\documentclass[11pt]{article}
\usepackage[margin=1in]{geometry}

\usepackage{graphicx} % Required for inserting images
\usepackage{hyperref}
\usepackage{amsmath}
\usepackage{amssymb, amsthm, amsfonts, graphicx, color, subcaption, enumerate, bm, enumitem, array, mathtools}
\usepackage{mathdots}
\usepackage{bbm}
\usepackage{tikz}

\usepackage{pgfplots}
\pgfplotsset{compat=1.18}

\usetikzlibrary{arrows.meta}
\renewcommand{\epsilon}{\varepsilon}
\renewcommand{\tilde}{\widetilde}
\title{The Complexity of Distributed Minimum \\ Weight Cycle Approximation}

\author{\hspace{0.5cm} Yi-Jun Chang\footnote{National University of Singapore. ORCID: 0000-0002-0109-2432. Email: cyijun@nus.edu.sg}  \and Yanyu Chen\footnote{National University of Singapore. ORCID: 0009-0008-8068-1649. Email: yanyu.chen@u.nus.edu}\and Dipan Dey\footnote{University of Houston, USA. ORCID: 0009-0001-0675-8790. Email: ddey@central.uh.edu} \and Yonggang Jiang\footnote{ETH Zurich. ORCID: 0009-0002-8485-6676. Email: yjiang@mpi-inf.mpg.de} \hspace{0.5cm}\and Gopinath Mishra\footnote{Institute of Mathematical Sciences, Chennai, India. ORCID: 0000-0003-0540-0292. Email: gopianjan117@gmail.com} \and Hung Thuan Nguyen\footnote{National University of Singapore. ORCID: 0009-0006-7993-2952.  Email: hung@u.nus.edu} \and Mingyang Yang\footnote{National University of Singapore. ORCID: 0009-0006-8971-2064. Email: myangat@u.nus.edu} }

\date{}

\usepackage{thmtools, thm-restate}

\usepackage[style=trad-alpha,natbib=true,maxcitenames=20]{biblatex}
\usepackage{algpseudocode,algorithm}

\usepackage[colorinlistoftodos,prependcaption,textsize=tiny]{todonotes}

\newcommand{\dist}{\operatorname{dist}}

\newcommand{\CONGEST}{\mathsf{CONGEST}}
\newcommand{\ID}{\operatorname{ID}}
\newcommand{\poly}{\operatorname{poly}}

\newcommand{\exponential}{\mathsf{Exponential}}

\newcommand{\MWC}{\mathsf{MWC}}
\newcommand{\RP}{\mathsf{RPaths}}
\newcommand{\APSP}{\mathsf{APSP}}

\newcommand{\apxMWC}{\mathsf{Apx}\text{-}\mathsf{MWC}}
\newcommand{\OPT}{\mathsf{OPT}}

\newcommand{\LDD}{\mathsf{LDD}}

\newcommand{\capture}{\mathcal{E}_{\operatorname{capture}}}
\newcommand{\congestion}{\mathsf{congestion}}
\newcommand{\dilation}{\mathsf{dilation}}
\newcommand{\arrival}{\tau_{\operatorname{arrive}}}

\usepackage{cleveref}

\newtheorem{lemma}{Lemma}[section]

\newtheorem{conjecture}[lemma]{Conjecture}

\newtheorem{definition}[lemma]{Definition}

\newtheorem{observation}[lemma]{Observation}
\newtheorem{proposition}[lemma]{Proposition}

\newcommand{\bin}{\{0,1\}}
\newcommand{\disj}{\mathsf{disj}}
\newcommand{\inprod}[2]{\langle #1, #2\rangle}

\begin{document}

\maketitle
\begin{abstract}
We study the \emph{Minimum Weight Cycle} ($\MWC$) problem in the $\mathsf{CONGEST}$ model of distributed computing.
For undirected weighted graphs, we give a randomized $(k+1)$-approximation algorithm for every \underline{real number} $k \geq (1+\sqrt{5})/2 \approx 1.618$. The algorithm runs in
\(
\tilde{O}\left(n^{\frac{k+1}{2k+1}} + D\right)
\)
rounds, where $n$ is the number of nodes and $D$ is the unweighted diameter of the graph. Varying $k$ therefore yields a smooth trade-off between approximation ratio and round complexity.

On the lower bound side, assuming the Erd\H{o}s girth conjecture, for every \underline{integer} $k \geq 1$ and every $\epsilon > 0$, any randomized $(k+1-\epsilon)$-approximation algorithm for $\MWC$ requires
\(
\tilde{\Omega}\left(n^{\frac{k+1}{2k+1}}+D\right)
\)
rounds. The lower bound holds for both directed unweighted graphs and undirected weighted graphs, even on graphs of diameter $\Theta(\log n)$. This gives a nearly tight characterization of the round complexity of approximate $\MWC$.

Our algorithm is based on a new connection between $\MWC$ and the low-diameter decomposition of Miller, Peng, and Xu (SPAA~2013). This connection has consequences beyond the $\mathsf{CONGEST}$ model: for every real number $k \geq 1$, this leads to a $(k+1)$-approximation algorithm for $\MWC$ with $\tilde{O}(mn^{1/k})$ work and $\tilde{O}(1)$ depth in the parallel setting, as well as an $\tilde{O}(n^{1/k})$-round algorithm in the broadcast congested clique model. We complement the latter bound with a matching $\tilde{\Omega}(n^{1/k})$
lower bound for $(k+1-\epsilon)$-approximation in the broadcast
congested clique model, for every integer $k\geq1$, assuming the Erd\H{o}s girth conjecture.
\end{abstract}
\thispagestyle{empty}
\newpage
\thispagestyle{empty}
{\small
\tableofcontents
}
\newpage
\pagenumbering{arabic}
\section{Introduction}

We study the \emph{Minimum Weight Cycle} ($\MWC$) problem in the $\CONGEST$ model~\cite{peleg2000distributed} of distributed computing. In this model, the communication network is represented by a graph $G=(V,E)$, where each node $v\in V$ corresponds to a computing device and each edge $e\in E$ corresponds to a communication link. Computation proceeds in synchronous rounds; in each round, every node may exchange $O(\log n)$ bits with each of its neighbors. Throughout the paper, unless otherwise stated, we write $n=|V|$ and $m=|E|$ for the number of nodes and edges in the graph under consideration. The notation $\tilde{O}(\cdot)$, $\tilde{\Omega}(\cdot)$, and $\tilde{\Theta}(\cdot)$ suppresses factors that are polylogarithmic in $n$.

The goal of $\MWC$ is to find a cycle of minimum total edge weight. The weight of such a cycle is the \emph{girth} of the graph, a fundamental graph parameter with a long history in graph theory and algorithms. The problem can be viewed as a \emph{round-trip} analogue of shortest paths: rather than finding the cheapest route from one point to another, it asks for the cheapest way to leave a point and return to it.  

In the centralized setting, $\MWC$ is 
closely connected to shortest-path computation.
For dense graphs, undirected weighted $\MWC$ and directed weighted
\emph{All-Pairs Shortest Paths} ($\APSP$) are equivalent under
subcubic reductions~\cite{williams2010subcubic}.
The recent breakthrough by Anthropic’s Claude~\cite{alman2026truly} gives an
$O(n^{2.9995})$-time algorithm for $\APSP$, implying truly subcubic algorithms
for $\MWC$ and the other problems in this equivalence class.
For sparse graphs, sparsity-preserving reductions connect $\MWC$
to several fundamental graph problems with 
$\tilde{O}(mn)$-time algorithms~\cite{agarwal2018fine}.

Despite its fundamental role in graph algorithms, $\MWC$ has only recently been systematically studied in the $\CONGEST$ model. Manoharan and Ramachandran~\cite{10.1145/3662158.3662801,manoharan2024computing} conducted a comprehensive investigation across the four basic settings:
\[
\{\text{undirected unweighted},\ \text{undirected weighted},\ \text{directed unweighted},\ \text{directed weighted}\}.
\]

For the weighted settings, following the standard assumption in prior work~\cite{10.1145/3662158.3662801}, we assume that all edge weights are positive integers bounded by a polynomial in $n$, so that each edge weight can be represented using $O(\log n)$ bits. Formally, throughout the paper, we write $w\colon E\to\mathbb{N}{\geq 1}$ for the edge-weight function and assume that $w_{\max}=\max_{e\in E}w(e)\in n^{O(1)}$.

A central challenge in understanding $\MWC$ in the $\CONGEST$ model is that its round complexity appears to extend beyond the most familiar complexity classes:
\[
\left\{
\tilde{\Theta}(n),
\tilde{\Theta}(\sqrt{n}+D),
\tilde{\Theta}(D)
\right\}.
\]
Indeed, as we discuss below, prior bounds for approximate $\MWC$ leave gaps involving intermediate complexities such as $n^{2/3}$ and $n^{4/5}$. Bounds of this form typically cannot be obtained by directly applying standard algorithmic or lower bound techniques, and resolving them may require substantially new ideas.

\subsection{Prior Work}\label{subsec:prior-work}

We briefly overview prior work on the $\MWC$ problem in the $\CONGEST$ model. Throughout the paper, we write $w(C)$ for the weight of a cycle $C$, and we write $\OPT$ for the minimum weight of a cycle in the communication network $G$. The approximate version of $\MWC$ is defined as follows.

\begin{definition}[$\apxMWC$]
For any real number $\alpha \geq 1$, the goal of the $\alpha$-$\apxMWC$ problem is to compute a cycle $C$ such that
\[
    w(C) \leq \alpha \cdot \OPT.
\]
\end{definition}

In the distributed setting, we require the output to be represented as follows. All nodes know the value $w(C)$ of the output cycle $C$, and for every edge of the graph, both of its endpoints know whether the edge belongs to $C$.

We write $D$ for the undirected unweighted diameter of the communication network $G$. Since $\apxMWC$ is inherently a \emph{global} problem, $\Omega(D)$ is a trivial lower bound for any approximation ratio $\alpha$. Indeed, consider a graph obtained by attaching a path of  $\Omega(D)$ nodes to a cycle $C$. In the required output format, even the node at the far end of the path must learn the value $w(C)$, which requires $\Omega(D)$ rounds.

\paragraph{The undirected unweighted setting.}
In undirected unweighted graphs, \citet{10.1145/2332432.2332504} showed that the \emph{exact} $\MWC$ problem can be solved in $O(n)$ rounds in the $\CONGEST$ model. On the lower bound side, \citet{10.5555/2095116.2095207} showed an $\tilde{\Omega}(\sqrt{n})$ lower bound for $(2-\epsilon)$-$\apxMWC$, for any  $\epsilon > 0$.

For approximation algorithms, \citet{peleg2012distributed} showed that $(2 - 1/g)$-$\apxMWC$ can be solved in $\tilde{O}(\sqrt{ng} + D)$ rounds, where $g = \OPT$ is the girth of the communication network $G$. This round complexity was subsequently improved to $\tilde{O}(\sqrt{n} + D)$ by \citet{10.1145/3662158.3662801}, matching the $\tilde{\Omega}(\sqrt{n})$ lower bound.

\paragraph{The remaining settings.}
We now turn to the remaining three settings: undirected weighted, directed unweighted, and directed weighted graphs. In all three settings, the \emph{exact} $\MWC$ problem can be solved in $\tilde{O}(n)$ rounds~\cite{bernstein2019distributed}. This is \emph{optimal}, as \citet{10.1145/3662158.3662801} showed an $\tilde{\Omega}(n)$ lower bound for $(2-\epsilon)$-$\apxMWC$, for any constant $\epsilon > 0$, in all three settings.

The more intriguing regime is the approximation regime $\alpha \geq 2$, where the $\tilde{\Omega}(n)$ lower bound no longer applies. For $\alpha$-$\apxMWC$, \citet{10.1145/3662158.3662801} established a lower bound of
\[
    \tilde{\Omega}(\sqrt{n}).
\]
In undirected weighted graphs, this lower bound holds even for approximation ratios $\alpha$ as large as $O(w_{\max}/\sqrt{n})$, where $w_{\max}$ is the maximum edge weight. Since $w_{\max}$ may be an arbitrary polynomial in $n$, the lower bound applies even to polynomial approximation ratios. In directed graphs, weighted or unweighted, the lower bound holds for \emph{every} approximation ratio $\alpha$.

These lower bounds are complemented by the following upper bounds of \citet{10.1145/3662158.3662801}:
\begin{align*}
    (2+\epsilon)\text{-}\apxMWC
        &\quad \text{in } \tilde{O}(n^{2/3} + D) \text{ rounds}
        && \text{for undirected weighted graphs}, \\
    (2+\epsilon)\text{-}\apxMWC
        &\quad \text{in } \tilde{O}(n^{4/5} + D) \text{ rounds}
        && \text{for directed weighted graphs}, \\
    2\text{-}\apxMWC
        &\quad \text{in } \tilde{O}(n^{4/5} + D) \text{{rounds}}
        && \text{for directed unweighted graphs}.
\end{align*}
Here $\epsilon > 0$ can be an arbitrarily small constant.

Together, these results reveal a gap in our understanding of  $\apxMWC$ in the $\CONGEST$ model. In the regime $\alpha \geq 2$, the best known lower bound is $\tilde{\Omega}(\sqrt{n})$, while the best known upper bounds are $\tilde{O}(n^{2/3}+D)$ and $\tilde{O}(n^{4/5}+D)$. This naturally leads to the following questions:
\begin{center}
\begin{description}
    \item[(Q1)] Can we narrow or close the gap between the upper and lower bounds?
    \item[(Q2)] More generally, what is the optimal round complexity as a function of the approximation ratio?
\end{description}
\end{center}

\subsection{Our Contributions}

We make progress on both questions raised above. 

\paragraph{Upper bound.} Our first contribution is an improved upper bound for approximate $\MWC$ in undirected weighted graphs. Throughout the paper, we say that an event occurs \emph{with high probability} if it happens with probability $1 - 1/\poly(n)$.

\begin{restatable}[Upper bound]{theorem}{kupper} \label{thm:mainUB} For any real number $k \geq 1$, the $(k+1)$-$\apxMWC$ problem in undirected weighted graphs can be solved with high probability in the $\CONGEST$ model in \[ \tilde{O}\left( n^{\frac{k+1}{2k+1}} + n^{\frac{1}{k}} + D \right) \] rounds. In particular, when $k \geq \frac{1+\sqrt{5}}{2} \approx 1.618$, the bound becomes \[ \tilde{O}\left( n^{\frac{k+1}{2k+1}} + D \right). \] \end{restatable} 

We emphasize that the parameter $k$ in \Cref{thm:mainUB} can be any real number at least $1$. This gives a continuous tradeoff between approximation ratio and round complexity.  There is no need to extend the theorem below $k=1$, since the \emph{exact} $\MWC$ problem can already be solved in $\tilde{O}(n)$ rounds~\cite{bernstein2019distributed}.

Moreover, we can slightly improve the approximation ratio essentially for free: decreasing $k$ additively by $\Theta(\log\log n / \log n)$ increases the round complexity only by a polylogarithmic factor, which is absorbed by the $\tilde{O}(\cdot)$ notation. Consequently, the same round complexity bound in \Cref{thm:mainUB} also holds for $(k+1-\epsilon)$-$\apxMWC$, for any real number $0 < \epsilon \in O(\log\log n / \log n)$. 

Compared with the previous upper bound of \citet{10.1145/3662158.3662801}, our algorithm is faster whenever $k > 1.5$ is a constant: its round complexity is then strictly below $\tilde{O}(n^{2/3}+D)$. Thus, while the previous algorithm gives a $(2+\epsilon)$-approximation in $\tilde{O}(n^{2/3}+D)$ rounds, our result yields faster algorithms when one allows for a larger approximation ratio.

\paragraph{Lower bound.} We complement the upper bound with a matching lower bound, conditional on the \emph{Erd\H{o}s girth conjecture}~\cite{erdos1964extremal}. This conjecture, which asserts the existence of graphs that simultaneously have high girth and high edge density, has played an important role in many lower bounds, including lower bounds for graph spanners~\cite{thorup2005approximate}. 
The Erd\H{o}s girth conjecture is proved for $k \in \{1,2,3,5\}$: the case $k=1$ is trivial, while the cases $k \in \{2,3,5\}$ follow from the constructions of Wenger~\cite{wenger1991extremal} and Benson~\cite{benson1966minimal}.

For convenience, we use the following bipartite form of the conjecture. This version is equivalent to the usual formulation up to constant factors, since every graph contains a bipartite subgraph with at least half of its edges. 

\begin{conjecture}[Erd\H{o}s girth conjecture] \label{girth_conj} For all integers $k \geq 1$ and $n \geq 1$, there exists a bipartite graph $G=(V,E)$ with \[ |E| \in \Omega\left(n^{1+\frac{1}{k}}\right) \quad \text{and} \quad\operatorname{girth}(G) > 2k. \] \end{conjecture}

Since the graph in \Cref{girth_conj} is bipartite, the condition $\operatorname{girth}(G)>2k$ is equivalent to $\operatorname{girth}(G)\geq 2k+2$. Under this conjecture, we prove the following lower bound.

\begin{restatable}[Lower bound]{theorem}{klower}
\label{thm:klower}
Assuming the Erd\H{o}s girth conjecture, for every integer $k \geq 1$ and every real number {$\epsilon>0$ such that $1/\epsilon$ is polynomially bounded in $n$}, any distributed algorithm that solves the $(k+1-\epsilon)$-$\apxMWC$ problem with high probability requires
\[
    \tilde{\Omega}\left(n^{\frac{k+1}{2k+1}}\right)
\]
rounds in the $\CONGEST$ model. This lower bound holds both for directed unweighted graphs and for undirected weighted graphs, even when the hop diameter is $O(\log n)$.
\end{restatable}

{Combining this bound with the $\Omega(D)$ lower bound from \Cref{subsec:prior-work} gives a worst-case lower bound of $\tilde{\Omega}(n^{(k+1)/(2k+1)}+D)$ rounds.}
{For directed unweighted graphs, the lower bound holds for every $\epsilon>0$. The restriction on $1/\epsilon$ is needed only to make the edge weights in the undirected construction positive integers of polynomial size.}
For every constant integer $k \geq 2$, this lower bound is strictly better than the previously known $\tilde{\Omega}(\sqrt{n})$ lower bound of \citet{10.1145/3662158.3662801}.

More importantly, for all integers $k \geq 2$, our upper and lower bounds \emph{match} up to polylogarithmic factors: by the parameter shifting observation above, the upper bound in \Cref{thm:mainUB} applies to $(k+1-\epsilon)$-$\apxMWC$, while \Cref{thm:klower} rules out faster algorithms for the same approximation ratio. Consequently, we establish
\[
    \tilde{\Theta}\left(n^{\frac{k+1}{2k+1}} + D\right)
\]
as the \emph{tight} round complexity of $(k+1-\epsilon)$-$\apxMWC$, up to polylogarithmic factors, for every integer $k \geq 2$ and {every $\epsilon>0$ such that $1/\epsilon$ is polynomially bounded in $n$ and $\epsilon\in O(\log\log n/\log n)$}.

  \Cref{fig-tradeoff} illustrates the resulting round--approximation tradeoff and compares it with the previous bounds. For approximation ratios below $2$, prior work already gives a tight $\tilde{\Theta}(n)$ bound. {For each integer $k\geq 2$, the bounds match up to polylogarithmic factors at approximation ratio $k+1-\epsilon$.} 
  
  %Recall that \citet{10.1145/3662158.3662801} showed that a $(2+\varepsilon)$-approximation is achievable in $\tilde{O}(n^{2/3}+D)$ rounds. \textcolor{blue}{Thus, the round complexity of the $\apxMWC$ problem remains open for approximation ratios between $2+\varepsilon$ and $3-\varepsilon$.} %\gopinath{I have added this sentence to present the complete current status of the problem.}\yijun{I would tend to not have this sentence because even above 3 the upper and lower bounds are not matched outside of the discrete points, so in a sense the region between 2 and 3 is not particularly special.}

\begin{figure}[ht]
\centering
\begin{tikzpicture}
\begin{axis}[
    width=13cm,
    height=12cm,
    xlabel={Approximation Ratio},
    ylabel={Exponent of $n$ in Round Complexity},
    xmin=1, xmax=6,
    ymin=0.48, ymax=1.02,
    domain=0:5,
    samples=200,
    legend style={
        at={(0.97,0.97)},
        anchor=north east,
        font=\tiny
    },
    legend cell align=left,
    grid=both,
]

% ===============================
% Our Upper Bound Curve
% ===============================
\addplot[blue, thick, domain=1:5, samples=201]
    ({x+1}, {max((x+1)/(2*x+1), 1/x)});
\addlegendentry{Our upper bound: $(k+1)$-$\apxMWC$, $\tilde{O}\left(n^{\frac{k+1}{2k+1}} + n^{\frac{1}{k}} + D\right)$}

% ===============================
% Our Lower Bound: horizontal pieces with white right endpoints
% ===============================
\addplot[red, thick] coordinates {(2, 3/5) (3, 3/5)};
\addlegendentry{Our lower bound: $(k+1-\epsilon)$-$\apxMWC$, $\tilde{\Omega}\left(n^{\frac{k+1}{2k+1}}\right)$}

\addplot[red, thick, forget plot] coordinates {(3, 4/7) (4, 4/7)};
\addplot[red, thick, forget plot] coordinates {(4, 5/9) (5, 5/9)};
\addplot[red, thick, forget plot] coordinates {(5, 6/11) (6, 6/11)};

% White circles indicating the excluded right endpoints
\addplot[
    red,
    only marks,
    mark=*,
    forget plot,
    mark options={draw=red, fill=white, fill opacity=1, line width=0.8pt},
]
coordinates {
    (3, 3/5)
    (4, 4/7)
    (5, 5/9)
    (6, 6/11)
};

% Filled circles indicating the left endpoints
\addplot[
    red,
    only marks,
    mark=*,
    forget plot,
]
coordinates {
    (2, 3/5)
    (3, 4/7)
    (4, 5/9)
    (5, 6/11)
    (6, 7/13)
};

% ===============================
% MR PODC 2024 Upper Bound: (2+eps)-approx
% ===============================
\addplot[cyan, thick] coordinates {(2, 2/3) (6, 2/3)};
\addlegendentry{\cite{10.1145/3662158.3662801} upper bound: $(2+\epsilon)$-$\apxMWC$, $\tilde{O}(n^{2/3}+D)$}

\addplot[
    cyan,
    only marks,
    mark=*,
    forget plot,
    mark options={draw=cyan, fill=white, fill opacity=1, line width=0.8pt},
]
coordinates {(2, 2/3)};

% ===============================
% MR Lower Bound: sqrt(n)
% ===============================
\addplot[orange, thick] coordinates {(1, 0.5) (6, 0.5)};
\addlegendentry{\cite{10.1145/3662158.3662801} lower bound: any approximation ratio, $\tilde{\Omega}(\sqrt{n})$}

% ===============================
% Tight bound below 2
% ===============================
\addplot[black, thick] coordinates {(1, 1) (2, 1)};
\addlegendentry{\cite{10.1145/3662158.3662801,bernstein2019distributed} tight bound: $(2-\epsilon)$-$\apxMWC$, $\tilde{\Theta}(n)$}

\addplot[
    black,
    only marks,
    mark=*,
    forget plot,
    mark options={draw=black, fill=white, fill opacity=1, line width=0.8pt},
]
coordinates {(2, 1)};

\end{axis}
\end{tikzpicture}
\caption{
Tradeoff between approximation ratio and the exponent of $n$ in the round complexity, ignoring polylogarithmic factors.
Assuming the Erd\H{o}s girth conjecture, the red lower-bound exponent immediately to the left of each integer ratio $k+1$ agrees with the blue upper-bound exponent at $k+1$ for $k\geq 2$. Each red step includes its left endpoint (filled circle) and excludes its right endpoint (open circle).
}
\label{fig-tradeoff}
\end{figure}

\paragraph{Technical contribution.}
Our main algorithmic contribution is a new use of the Miller--Peng--Xu (MPX) low-diameter decomposition~\cite{miller2013parallel} as a \emph{cycle-finding primitive}. For a suitable choice of parameters, we show that with non-negligible probability, one cluster contains an entire fixed minimum weight cycle while its center remains close to the cycle. This event allows us to extract from the cluster a cycle whose weight is within a factor $k+1$ of optimal. By itself, this idea already gives an $\tilde{O}(n^{\frac{k+1}{2k}}+D)$-round algorithm for every real number $k\geq 1$.

Obtaining the sharper bound $\tilde{O}(n^{\frac{k+1}{2k+1}}+D)$ requires overcoming the congestion caused by repeatedly computing such decompositions. Our key refinement is to treat minimum weight cycles with few edges and many edges differently. In the former case, we exploit hop-bounded approximate SSSP to reduce congestion. In the latter case, we initiate clusters only from a sampled subset of nodes, which substantially increases the success probability relative to the number of cluster centers. Although this restricted decomposition may fail to capture the entire cycle, we prove that only one short portion of the cycle can be missing and recover it using additional hop-bounded shortest-path computations.

A further technical ingredient is a robustness analysis showing that $(1+\epsilon)$-approximate SSSP suffices, even though a faithful implementation of the MPX decomposition requires exact SSSP and some standard properties of MPX decomposition fail under approximate distances; see~\cite[Appendix~A]{RozhonHMGZ23}. More broadly, our MPX-based approach provides a flexible framework that can be combined with shortest-path primitives from different computational models. Beyond $\CONGEST$, this yields our results for low-depth parallel computation and the broadcast congested clique.

\paragraph{Further application 1: parallel computation.} We first consider the standard \emph{work-depth} model of parallel computation~\cite{blelloch1996programming}. In this model, the \emph{work} of an algorithm is the total number of operations performed, while the \emph{depth} is the length of the longest chain of dependencies. Thus, the work measures the overall computational cost, and the depth measures the amount of inherently sequential computation that remains after parallelization.

\begin{restatable}[Parallel computation]{theorem}{parallelupper}
\label{thm:parallelUB}
For every real number $k \geq 1$, the $(k+1)$-$\apxMWC$ problem in undirected weighted graphs can be solved with high probability using $\tilde{O}(mn^{1/k})$ work and $\tilde{O}(1)$ depth. 
\end{restatable}

We compare \Cref{thm:parallelUB} with prior work. The weighted undirected $\apxMWC$ problem has a rich history in the centralized setting. Lingas and Lundell~\cite{lingas2009efficient} gave an $\tilde{O}(n^2)$-time $2$-approximation algorithm, and Roditty and Tov~\cite{roditty2013approximating} subsequently improved the approximation factor to $4/3$ with the same running time. The first subquadratic-time approximation algorithms for $\apxMWC$ were obtained by Ducoffe~\cite{ducoffe2021faster}, who gave an $\tilde{O}(m+n^{5/3})$-time $2$-approximation. More recently, Kadria, Roditty, Sidford, Vassilevska Williams, and Zwick~\cite{kadria2022algorithmic} obtained a general tradeoff, giving a $2k$-approximation in $\tilde{O}(m+n^{1+1/k})$ time for every integer $k \geq 1$. Their recent follow-up work~\cite{kadria2023improved} further improves the approximation guarantee to $4k/3$.

It is not clear to us whether these algorithms directly imply low-depth parallel algorithms with comparable work. In terms of the tradeoff between work and approximation ratio, \Cref{thm:parallelUB} and the above centralized algorithms are not directly comparable. The algorithms of Kadria, Roditty, Sidford, Vassilevska Williams, and Zwick~\cite{kadria2022algorithmic,kadria2023improved} achieve a better time bound $\tilde{O}(m+n^{1+1/k})$, with approximation ratio $4k/3$. In comparison, \Cref{thm:parallelUB} gives a low-depth parallel algorithm with $\tilde{O}(mn^{1/k})$ work and approximation ratio $k+1$, which is better than $4k/3$ when $k > 3$. One advantage of our approach is that the use of MPX low-diameter decomposition yields a smooth tradeoff: our algorithm works for every real number $k \geq 1$, rather than only for integer values of $k$ as in the existing algorithms.

\paragraph{Further application 2: broadcast congested clique.} We next consider the \emph{broadcast congested clique} model~\cite{drucker2014power}. In this model, every pair of nodes can communicate directly, but in each round each node must broadcast the same $O(\log n)$-bit message to all other nodes. Thus, unlike in $\CONGEST$, communication is not restricted to the edges of the input graph; however, the broadcast restriction still creates a nontrivial communication bottleneck. This model is essentially the same as the  number-in-hand model of multi-party communication complexity with a shared blackboard~\cite{kushilevitz1997communication}: a message written on the blackboard is visible to all players, just as a broadcast message is received by all nodes.

\begin{restatable}[Broadcast congested clique upper bound]{theorem}{cliqueupper}
\label{thm:cliqueUB}
For every real number $k \geq 1$, the $(k+1)$-$\apxMWC$ problem in undirected weighted graphs can be solved with high probability in \[\tilde{O}\left(n^{1/k}\right)\] rounds in the broadcast congested clique model.
\end{restatable}

Our $\CONGEST$ lower bound construction also implies a matching lower bound in the broadcast congested clique model.

\begin{restatable}[Broadcast congested clique lower bound]{theorem}{cliquelower}
\label{thm:cliqueLB}
Assuming the Erd\H{o}s girth conjecture, for every integer $k \geq 1$ and every real number {$\epsilon>0$ such that $1/\epsilon$ is polynomially bounded in $n$}, any distributed algorithm that solves the $(k+1-\epsilon)$-$\apxMWC$ problem with high probability requires
\[
    \tilde{\Omega}\left(n^{1/k}\right)
\]
rounds in the broadcast congested clique model. This lower bound holds both for directed unweighted graphs and for undirected weighted graphs.
\end{restatable}

Therefore, in the broadcast congested clique, the upper and lower bounds also match, up to polylogarithmic factors, for $(k+1-\epsilon)$-approximation, for every integer {$k\geq1$}. {For $k=1$, each node can broadcast its incident edges, one per round, so that after $O(n)$ rounds every node can compute an exact minimum weight cycle locally.}

\subsection{Independent Work}

Independently, Chechik, Lifshitz, and Mukhtar~\cite{chechik2026girth} also study $\apxMWC$ in the $\CONGEST$ model. In particular, they independently obtain the same lower bound as our \Cref{thm:klower}. Their lower bound relies on the same  Erd\H{o}s girth conjecture as ours, after reindexing the parameters.

On the algorithmic side, the result most closely related to our main upper bound is their tradeoff for undirected weighted graphs: for every \emph{integer} $k \geq 2$, they solve $(2k-1+o(1))$-$\apxMWC$ in $\tilde{O}\left(n^{\frac{k+1}{2k+1}}+D\right)$ rounds. Our upper bound gives a strictly stronger guarantee at the same round complexity: \Cref{thm:mainUB} solves $(k+1)$-$\apxMWC$ in $\tilde{O}\left(n^{\frac{k+1}{2k+1}}+D\right)$ rounds for every \emph{real} number $k \geq \frac{1+\sqrt{5}}{2} \approx 1.618$. Thus, for every $k>2$, our approximation ratio is better under the same round complexity bound. 

The two works take substantially different algorithmic approaches. Our approach is built around a new connection between $\MWC$ and MPX low-diameter decompositions, which forms the main technical basis of our upper bounds, whereas the approach of Chechik, Lifshitz, and Mukhtar does not rely on MPX decompositions.

%\yijun{I revised the text above a bit to highlight the difference, feel free to expand.}
 
Their work also obtains results in settings not addressed by our work. In particular, they solve $f$-$\apxMWC$ in undirected unweighted graphs in $\tilde{O}(n^{1/f}+D)$ rounds, for any {constant} integer $f>2$. For directed graphs, they give $\tilde{O}(n^{2/3}+D)$-round algorithms for $2$-$\apxMWC$ in the unweighted setting and for $(2+\epsilon)$-$\apxMWC$ in the weighted setting, improving over the previous $\tilde{O}(n^{4/5}+D)$ upper bounds of \citet{10.1145/3662158.3662801}.

%\gopinath{We were wondering if we should add some text comparing with the techniques of \cite{chechik2026girth}.}\yijun{that could be a good idea, but I don't think I understand their work enough to do this myself. We could say that we are using very different ideas, in particular, their approach does not involve mpx (which is the basis of our algorithms) at all}

\subsection{Additional Related Work}

In the centralized setting, the $\MWC$ problem has been studied extensively since the 1970s. In the unweighted case, Itai and Rodeh~\cite{itai1978finding} showed that exact $\MWC$ can be solved in $\min\{O(mn), O(n^\omega)\}$ time for both directed and undirected graphs, where $\omega$ denotes the matrix-multiplication exponent. Later, Roditty and Vassilevska Williams~\cite{roditty2011minimum} reduced exact $\MWC$ to the minimum weight triangle problem, obtaining $O(Mn^\omega)$-time algorithms for undirected graphs with integer weights in $[1,M]$ and directed graphs with integer weights in $[-M,M]$ and no negative cycles.

There has also been considerable recent progress on approximation algorithms. We have already discussed the undirected weighted case, so we focus here on the remaining settings. For undirected unweighted graphs, Kadria, Roditty, Sidford, Vassilevska Williams, and Zwick~\cite{kadria2022algorithmic} gave an $\tilde{O}(n^{1+1/k})$-time algorithm that, for every integer $k \geq 1$, returns a cycle of weight at most $2k\lceil g/2\rceil$, where $g = \OPT$. For directed graphs, \citet{chechik2020constant} gave an $\tilde{O}(m^{1+1/k})$-time algorithm achieving both $O(k\log k)$-$\apxMWC$ and $O(k\log\log n)$-$\apxMWC$, for every integer $k \geq 1$. This improves upon the earlier work of \citet{pachocki2018approximating}, which achieved an $O(k\log n)$ approximation within the same running time.

Next, we turn to the distributed setting. Beyond $\MWC$, distributed algorithms for finding short cycles have been studied extensively. \citet{CPZ21} used expander decompositions and routing to list all triangles in $\tilde{O}(n^{1/3})$ rounds in the $\CONGEST$ model, matching the $\tilde{\Omega}(n^{1/3})$ lower bound of \citet{IzumiL17}. More generally, a substantial body of work has investigated the round complexity of detecting $k$-node cycles in the $\CONGEST$ model~\cite{censorhillel_et_al:LIPIcs.DISC.2020.33,drucker2014power,fraigniaud2024even,fraigniaud_et_al:LIPIcs.ICALP.2025.80,korhonen_et_al:LIPIcs.OPODIS.2017.4}.

The $\MWC$ problem has also been studied in other distributed models. In the congested clique model, where every pair of nodes can exchange $O(\log n)$ bits per round, \citet{censorhillel_et_al:LIPIcs.DISC.2020.33} showed that an additive-one approximation for undirected unweighted $\MWC$ can be computed in $O(1)$ rounds.

A closely related problem is \emph{Replacement Paths} ($\RP$). Given a shortest path $P$ from a source $s$ to a target $t$, $\RP$ asks, for every edge $e\in P$, for the shortest $s$-$t$ path avoiding $e$. Together with the corresponding subpath of $P$, such a replacement path forms a cycle containing $e$. As discussed earlier, this connection is also reflected in fine-grained complexity: $\MWC$ and $\RP$ belong to the same $\tilde{O}(n^3)$-time equivalence class for dense graphs~\cite{williams2010subcubic} and the same $\tilde{O}(mn)$-time equivalence class for sparse graphs~\cite{agarwal2018fine}. In the $\CONGEST$ model, the recent work of \citet{manoharan2024computing}  initiated a systematic study of $\RP$. Subsequently, \citet{chang2025optimal} showed that unweighted directed $\RP$ has tight round complexity $\tilde{\Theta}(n^{2/3}+D)$, providing another example of a global graph problem whose complexity lies outside the most familiar distributed complexity classes.

\section{Technical Overview}
  
In this section, we present an overview of our proofs of \Cref{thm:mainUB,thm:klower,thm:parallelUB,thm:cliqueUB,thm:cliqueLB}.

\subsection{The Guiding Tradeoff}

Our $\CONGEST$  upper and lower bounds are both governed by the same underlying congestion--dilation tradeoff. For the $(k+1-\epsilon)$-$\apxMWC$ problem, the target round complexity is
\[
  \tilde{\Theta}\left(n^{\frac{k+1}{2k+1}} + D\right),
\]
which is exactly the bound achieved by \Cref{thm:mainUB} whenever
$k \geq \frac{1+\sqrt{5}}{2}$.

The exponent $\frac{k+1}{2k+1}$ arises from optimizing the simple expression
\[
    f(\alpha)=\frac{n}{\alpha}+\alpha^{1+1/k}.
\]
Informally, the term $n/\alpha$ represents a \emph{dilation} cost: information may have to travel along paths of  $\tilde{\Theta}(n/\alpha)$ nodes. The term $\alpha^{1+1/k}$ represents a \emph{congestion} cost: $\tilde{\Theta}(\alpha^{1+1/k})$ bits of information may have to traverse the same edge. Balancing these two terms gives
\(\min_\alpha\left\{\frac{n}{\alpha}+\alpha^{1+1/k}\right\}
    \in
    \Theta\left(n^{\frac{k+1}{2k+1}}\right)\).
    
Both our upper and lower bounds realize this tradeoff. On the lower bound side, we construct hard instances showing that any $\CONGEST$ algorithm must incur either $\tilde{\Omega}(n/\alpha)$ dilation or $\tilde{\Omega}(\alpha^{1+1/k})$ congestion. On the upper bound side, we employ different algorithmic approaches for long-hop and short-hop minimum weight cycles. By choosing an appropriate threshold separating these two cases, together with the parameter $\alpha$ controlling the congestion--dilation tradeoff of the approximate SSSP computation underlying the MPX low-diameter decomposition, we obtain the desired congestion--dilation tradeoff.

%On the algorithmic side, we design a distributed \((k+1)\)-approximation for the minimum weight cycle in the $\CONGEST$ model whose round complexity matches the tradeoff discussed above, up to polylogarithmic factors and omitting the dependence on the diameter~\(D\).

\subsection{Capturing Cycles via MPX Low-Diameter Decompositions}

We begin by describing a connection between $\apxMWC$ and MPX low-diameter decompositions, which is a key ingredient underlying our algorithms.

First, since the parameter $k\geq1$ in
\Cref{thm:mainUB,thm:parallelUB,thm:cliqueUB} can be any
real number, any additional $(1+\epsilon)$ factor in the
approximation ratio can be absorbed by decreasing $k$ by
a $1-O(\epsilon)$ factor. Inspecting the round complexity
bounds in these theorems, as long as
$\epsilon\in O(\log\log n/\log n)$, the resulting increase
in the round complexity is at most a polylogarithmic
factor and is therefore absorbed into the $\tilde{O}(\cdot)$
notation. Thus, in the following discussion, it suffices
to explain how to obtain a $(1+\epsilon)(k+1)$-approximation
for a sufficiently small $\epsilon$. To accommodate the possibility that this
adjustment takes the parameter slightly below $1$, we
allow $k\geq1/2$ in the actual analysis underlying these theorems.

For simplicity, throughout the technical overview, we assume that we are given a parameter $d$ such that $\OPT = 2d$, and our goal is to find a cycle of weight at most $(1+\epsilon)(k+1)\cdot 2d$. In the actual algorithm, we try $O(\epsilon^{-1}\log n)$ candidate values of $d$, ensuring that for one of them, the estimate $2d$ is within a $1+O(\epsilon)$ factor of $\OPT$.

\paragraph{MPX low-diameter decompositions.}
A key ingredient of our algorithm is a \emph{weighted} variant of the MPX low-diameter decomposition~\cite{miller2013parallel}. Each node $v$ independently samples a starting time of $-\delta_v$, where $\delta_v$ is drawn from the exponential distribution with parameter $\beta$. Starting from this time, node $v$ grows an SSSP tree.

It is helpful to view each SSSP tree as a continuously expanding wavefront, where traversing an edge $e$ of weight $w(e)$ takes time $w(e)$. Every node joins the cluster of the source whose wavefront reaches it first. In this way, the graph is partitioned into node-disjoint clusters, each equipped with a local SSSP tree rooted at its cluster center.

\paragraph{Analysis.}
We fix an arbitrary minimum weight cycle $C^\star$ of weight $2d$ and analyze how the decomposition behaves around it. We show how to choose $\beta$, as a function of $k$, $d$, and $n$, so that with probability $\Omega(n^{-1/k})$, one cluster captures all of $C^\star$. Moreover, the center of this cluster is close to the cycle: its distance to the nearest node of $C^\star$ is at most $kd(1+O(1/\log n))$.

\paragraph{Searching for a short cycle.}
Whenever the good event above occurs, we can find a cycle inside the cluster with weight at most
\[
    2 kd(1+O(1/\log n)) + 2d
    = 
    (1+\epsilon)(k+1)\cdot \OPT
\]
for some $\epsilon \in O(1/\log n)$, which meets our target approximation guarantee.

It remains to explain why such a cycle can be found within the cluster. It suffices to search over all cycles that consist of exactly one non-tree edge and otherwise only edges of the local SSSP tree. Indeed, at least one edge of $C^\star$ is not a tree edge. Taking such a non-tree edge $e$ together with the unique tree path between its endpoints {in} the local SSSP tree gives a cycle, and the good event ensures that this cycle satisfies the weight bound above.

Finally, to boost the success probability to $1-1/\poly(n)$, we repeat the decomposition independently $\tilde{O}(n^{1/k})$ times and output the lightest cycle found over all repetitions.

\subsection{Implementation via Approximate SSSP Computation}

The MPX low-diameter decomposition can be implemented using an undirected weighted SSSP computation. Conceptually, we add a virtual \emph{super source} connected to every node, where the weight of the edge incident to a node encodes its random starting time. To the best of our knowledge, all existing distributed SSSP algorithms can accommodate such a virtual super source.

A faithful implementation of the MPX low-diameter decomposition would use \emph{exact} SSSP. For our purposes, however, a sufficiently accurate $(1+\epsilon)$-approximation is enough. This is not immediate and requires careful analysis. In fact, some properties of the MPX low-diameter decomposition no longer hold if exact SSSP is replaced by $(1+\epsilon)$-approximate SSSP; see~\cite[Appendix~A]{RozhonHMGZ23} for a discussion. We identify the properties that remain valid under $(1+\epsilon)$-approximation and show that they suffice for our application. This is crucial for our algorithm, as the ability to use $(1+\epsilon)$-approximate rather than exact SSSP is what leads to the clean complexity upper bounds in \Cref{thm:mainUB,thm:parallelUB,thm:cliqueUB}.

\paragraph{Parallel computation.} We first consider the parallel setting. For $\epsilon \in \log^{-O(1)} n$, $(1+\epsilon)$-approximate SSSP in undirected weighted graphs can be computed with $\tilde{O}(m)$ work and $\tilde{O}(1)$ depth~\cite{andoni2020parallel,li2020faster,rozhovn2022undirected}. Since our approach repeats the decomposition $\tilde{O}(n^{1/k})$ times and these repetitions can be executed in parallel, the total cost becomes $\tilde{O}(mn^{1/k})$ work and $\tilde{O}(1)$ depth. This proves \Cref{thm:parallelUB}.

\paragraph{The broadcast congested clique model.}  The same idea also gives an immediate algorithm in the broadcast congested clique model. For $\epsilon \in \log^{-O(1)} n$, $(1+\epsilon)$-approximate SSSP in undirected weighted graphs can be computed in $\tilde{O}(1)$ rounds in this model~\cite{becker2021near}. Repeating the decomposition $\tilde{O}(n^{1/k})$ times therefore gives an overall round complexity of $\tilde{O}(n^{1/k})$, as stated in \Cref{thm:cliqueUB}.

\paragraph{The $\CONGEST$ model.}  It remains to understand what this direct implementation gives in the $\CONGEST$ model. For $\epsilon \in \log^{-O(1)} n$, the same approximate SSSP computation can be implemented using $\tilde{O}(n/\alpha+\alpha+D)$ rounds with congestion $\tilde{O}(\alpha)$, for any choice of parameter $\alpha$~\cite{becker2021near}. Since the decomposition is repeated $\tilde{O}(n^{1/k})$ times, the resulting collection of algorithms has congestion $\tilde{O}(\alpha n^{1/k})$ and dilation $\tilde{O}(n/\alpha+\alpha+D)$. Using \Cref{thm: dilation and congestion}, optimizing $\alpha$ to balance these two terms yields the round complexity
\[
    \tilde{O}\left(n^{\frac{k+1}{2k}} + D\right).
\]
This is already a nontrivial upper bound in the $\CONGEST$ model, but it falls short of our target
\[
    \tilde{O}\left(n^{\frac{k+1}{2k+1}} + D\right).
\]

\subsection{Sharpening the Round Complexity: Short-Hop Cycles}

We now explain how to sharpen the upper bound to achieve our target round complexity. The key idea is to handle \emph{short-hop} and \emph{long-hop} minimum weight cycles separately.

Let $\alpha$ be the parameter in our guiding congestion--dilation tradeoff, and define
\[
    h \in \tilde{\Theta}\left(\frac{n}{\alpha}\right)
\]
to be the threshold separating short-hop and long-hop cycles.

In the following discussion, we focus on the \emph{short-hop} regime, where we assume that an optimal minimum weight cycle $C^\star$ contains at most $h$ edges.

\paragraph{Hop-bounded approximate SSSP computation.}
For short-hop cycles, the key idea is to make the approximate SSSP computation \emph{hop-bounded}. We do so by adding a small perturbation $+\epsilon\cdot 2d/h$ to every edge weight. This changes $\OPT$ by at most a factor of $1+\epsilon$, while ensuring that every shortest path relevant to the approximate SSSP computation uses only $O(h/\epsilon)$ hops. For $\epsilon \in \log^{-O(1)} n$, such a hop-bounded $(1+\epsilon)$-approximate SSSP computation can be done with congestion $\tilde{O}(1)$ in $\tilde{O}(h)=\tilde{O}(n/\alpha)$ rounds~\cite{nanongkai2014distributed}.

This is the main advantage over the standard, non-hop-bounded approximate SSSP routine: the congestion drops from $\tilde{O}(\alpha)$ to $\tilde{O}(1)$. As discussed earlier, we repeat the decomposition independently for $\tilde{O}(n^{1/k})$ iterations, so the total congestion is now only $\tilde{O}(n^{1/k})$, while the dilation remains $\tilde{O}(h)=\tilde{O}(n/\alpha)$. Consequently, the overall round complexity is
\[
    \tilde{O}\left(n^{1/k}+n/\alpha\right),
\]
which is within our target bound. Indeed, the term $\tilde{O}(n/\alpha)$ is exactly the desired dilation term in the congestion--dilation tradeoff, and the additional term $\tilde{O}(n^{1/k})$ is dominated by our target bound $\tilde{O}\left(n^{\frac{k+1}{2k+1}}\right)$ whenever $k\ge (1+\sqrt{5})/2$. This proves \Cref{thm:mainUB} in the short-hop regime.

\subsection{Sharpening the Round Complexity: Long-Hop Cycles}

Next, we turn to the \emph{long-hop} regime, where a minimum weight cycle $C^\star$ contains at least $h$ edges, for
\(h\in\tilde{\Theta}(n/\alpha)\),
where $\alpha$ is the parameter in our target congestion--dilation tradeoff.

\paragraph{Sharpening the upper bound via congestion reduction.}
Recall that implementing each MPX low-diameter decomposition via approximate SSSP costs $\tilde{O}(n/\alpha+\alpha+D)$ rounds and incurs congestion $\tilde{O}(\alpha)$. Since the decomposition is repeated $\tilde{O}(n^{1/k})$ times, the total cost is $\tilde{O}(\alpha n^{1/k})$ congestion and $\tilde{O}(n/\alpha+\alpha+D)$ dilation.
Thus, to achieve our target congestion--dilation tradeoff, and hence prove \Cref{thm:mainUB}, it suffices to reduce the congestion from $\tilde{O}(\alpha n^{1/k})$ to $\tilde{O}(\alpha^{1+1/k})$. Equivalently, it suffices to reduce the number of repetitions from $\tilde{O}(n^{1/k})$ to $\tilde{O}(\alpha^{1/k})$. We show that this is indeed achievable in the long-hop regime.

\paragraph{Restricting the MPX low-diameter decomposition.}
For long-hop cycles, the key observation is that, when $C^\star$ contains many edges, it is unnecessary to initiate a cluster from every node in the MPX low-diameter decomposition. Instead, we randomly sample only $\tilde{\Theta}(\alpha)$ \emph{skeleton nodes} and initiate clusters exclusively from these nodes. Since $C^\star$ contains at least $h\in\tilde{\Theta}(n/\alpha)$ edges, a standard sampling argument shows that, with high probability, it contains $\tilde{\Omega}(1)$ skeleton nodes. Conceptually, restricting the decomposition to the sampled skeleton nodes replaces the parameter $n$ by $\alpha$ in the MPX analysis. As a result, the number of repetitions decreases from $\tilde{O}(n^{1/k})$ to $\tilde{O}(\alpha^{1/k})$, yielding the desired improvement in congestion.

This reduction in repetitions comes at the expense of a weaker capture guarantee. Previously, we showed that a single execution of the decomposition captures the entire minimum weight cycle $C^\star$ with probability $\Omega(n^{-1/k})$. Once cluster centers are restricted to the sampled skeleton nodes, such a guarantee is no longer possible. Instead, we prove that all skeleton nodes on $C^\star$ are captured by a single cluster with probability $\Omega(\alpha^{-1/k})$. In fact, our analysis establishes an even stronger structural property. The sampled skeleton nodes partition $C^\star$ into \emph{segments}, each containing only $\tilde{O}(n/\alpha)$ edges. We show that every segment, except the one farthest from the cluster center, is guaranteed to lie entirely within the cluster.

\paragraph{Patching the missing segment.} It remains to recover the missing segment. We do this by running $(1+\epsilon)$-approximate hop-bounded SSSP from every skeleton node, with hop bound $\tilde{O}(n/\alpha)$. This allows the two endpoints of the missing segment to detect the segment and obtain a sufficiently accurate estimate of its weight.

As discussed earlier, a single hop-bounded approximate SSSP computation has congestion $\tilde{O}(1)$ and takes $\tilde{O}(n/\alpha)$ rounds~\cite{nanongkai2014distributed}. Since we have only $\tilde{O}(\alpha)$ skeleton nodes and perform $\tilde{O}(\alpha^{1/k})$ repetitions, the total patching cost is congestion $\tilde{O}(\alpha^{1+1/k})$ and dilation $\tilde{O}(n/\alpha)$. This is exactly the desired congestion--dilation tradeoff.

There is one final subtlety: the approximate shortest paths used for patching must create simple cycles, rather than merely retracing paths already present in the tree associated with the cluster. We rule out this degeneracy by running the hop-bounded approximate SSSP computation on some modified graphs, which can still be simulated efficiently in the original graph.

\subsection{Lower Bound}

We conclude the technical overview by sketching the proofs of our lower bounds, \Cref{thm:klower,thm:cliqueLB}. While these lower bounds hold both for directed unweighted graphs and for undirected weighted graphs, in the technical overview we only consider the undirected weighted setting.

\paragraph{Hard instances.}
Assuming the Erd\H{o}s girth conjecture, we begin with a base graph $H$ on $\Theta(\alpha)$ nodes with girth at least $2k+2$ and $\Theta(\alpha^{1+1/k})$ edges. Every edge of $H$ is assigned weight~1.

From $H$, we derive two subgraphs $G_1$ and $G_2$. For each edge of $H$, we independently decide whether it appears in $G_1$, in $G_2$, in both, or in neither. Thus, the edge sets of $G_1$ and $G_2$ encode two $\Theta(\alpha^{1+1/k})$-bit strings, where the presence of an edge represents bit~1 and its absence represents bit~0.

Finally, for every node $v$ of $H$, we connect its two copies $v_1$ and $v_2$ in $G_1$ and $G_2$ by a simple path of  $\Theta(n/\alpha)$ nodes. All edges on these paths are assigned negligible  weight. For simplicity, in the following discussion we assume they have zero weight.

\paragraph{Model-specific settings.}
For the broadcast congested clique model, we set $\alpha=n$. In fact, the lower bound already holds when every connecting path is a single edge, so $G_1$ and $G_2$ are simply joined by a matching.

For the $\CONGEST$ model, $\alpha$ is the parameter in our guiding congestion--dilation tradeoff. We additionally attach a standard overlay tree to reduce the diameter to $O(\log n)$, following the framework underlying $\tilde{\Omega}(\sqrt{n})$ lower bounds in the $\CONGEST$ model~\cite{das2011distributed}. The overlay tree edges are assigned extremely large weight so that they never participate in a minimum weight cycle.

\paragraph{Analysis.}
Consider any edge $\{u,v\}$ of the base graph $H$. If the corresponding edges $\{u_1,v_1\}$ in $G_1$ and $\{u_2,v_2\}$ in $G_2$ are both present, then together with the two zero-weight connecting paths between $u_1,u_2$ and $v_1,v_2$, they form a cycle of total weight~2.

Otherwise, no such pair exists. By the girth assumption on $H$, every cycle contained in $G_1\cup G_2$ together with the connecting paths must traverse at least $2k+2$ base graph edges, and therefore has weight at least $2k+2$. The overlay tree cannot create a lighter cycle because its edges are prohibitively expensive.

Any $(k+1-\varepsilon)$-$\apxMWC$ algorithm must distinguish between the cases
\[
    \OPT\le2
    \qquad\text{and}\qquad
    \OPT\ge2k+2.
\]
The above discussion implies that such an algorithm must determine whether the edge sets of $G_1$ and $G_2$ intersect, thereby solving a set-disjointness instance of size $\Theta(\alpha^{1+1/k})$ across the natural cut separating $G_1$ and $G_2$.

\paragraph{Lower bound in the broadcast congested clique model.}
Recall that $\alpha=n$ in the broadcast congested clique model, so the resulting set-disjointness instance has size $\Theta(n^{1+1/k})$. By the classical $\Omega(N)$ communication lower bound for randomized set-disjointness on $N$-bit inputs~\cite{Razborov1992}, solving this instance requires $\Omega(n^{1+1/k})$ bits of communication. On the other hand, the broadcast congested clique can transmit only $\tilde{O}(n)$ bits per round over the entire network. It follows that every algorithm requires $\widetilde{\Omega}(n^{1/k})$ rounds, proving \Cref{thm:cliqueLB}.

\paragraph{Lower bound in the $\CONGEST$ model.}
In the $\CONGEST$ model, there are essentially two ways to communicate this information. One option is to send it along the connecting paths of  $\Theta(n/\alpha)$ nodes, incurring a \emph{dilation} cost of $\Omega(n/\alpha)$ rounds. The other is to route it through the overlay tree, which incurs a \emph{congestion} cost of $\tilde{\Omega}(\alpha^{1+1/k})$ rounds.

Intuitively, every algorithm must pay the smaller of these two costs, matching our guiding congestion--dilation tradeoff. The moving-cut framework~\cite{haeupler2020network,haeupler2021universally}, which generalizes the original construction of \citet{das2011distributed}, formalizes this intuition and yields the lower bound
\[
    \tilde{\Omega}\left(
        \max_{\alpha}
        \min\left\{
            \frac{n}{\alpha},
            \alpha^{1+1/k}
        \right\}
    \right)
    =
    \tilde{\Omega}\left(
        n^{\frac{k+1}{2k+1}}
    \right),
\]
thereby proving \Cref{thm:klower}.

\paragraph{Comparison with prior work.}
The previous lower bound of \citet{10.1145/3662158.3662801} also used a reduction from set-disjointness to establish an $\tilde{\Omega}(\sqrt{n})$ lower bound for any-factor approximation. Their construction follows the framework of \citet{das2011distributed}, originally developed for proving $\tilde{\Omega}(\sqrt{n})$ lower bounds for MST, SSSP, and many other distributed graph problems.

The same work also established an $\tilde{\Omega}(n)$ lower bound for $(2-\epsilon)$-$\apxMWC$. Their construction connects two bipartite graphs by a perfect matching, without using long paths or an overlay tree. Since the resulting set-disjointness instance has size $\Omega(n^2)$ and the two graphs are connected by only $O(n)$ edges, the reduction yields an $\tilde{\Omega}(n)$ lower bound.

Our lower bound instead draws inspiration from the recent $\tilde{\Omega}(n^{2/3})$ lower bound for the $\RP$ problem by \citet{chang2025optimal}, which showed that encoding information using the edges of a graph allows the framework of \citet{das2011distributed} to go beyond the $\tilde{\Omega}(\sqrt{n})$ barrier. More broadly, encoding quadratic amounts of information using edges has appeared in several recent lower bounds in the $\CONGEST$ model~\cite{10.5555/2095116.2095207,10.1145/3662158.3662801,manoharan2024computing}. Our lower bound combines this idea with the construction of \citet{10.1145/3662158.3662801} and the Erd\H{o}s girth conjecture to obtain the desired result.

\section{Preliminaries}\label{sec: preliminaries}

In this section, we define the distributed models considered in this paper, introduce the notions of congestion and dilation, and review a standard scheduling tool for distributed algorithms.

\paragraph{The $\CONGEST$ model.}
In the $\CONGEST$ model~\cite{peleg2000distributed}, the communication network is represented by a graph $G=(V,E)$, where each node corresponds to a computational device and each edge corresponds to a bidirectional communication link. Computation proceeds in synchronous rounds. In each round, every node performs arbitrary local computation, exchanges an $O(\log n)$-bit message with each of its neighbors, and updates its local state. Throughout the paper, we assume that each node $v$ has a unique identifier $\ID(v)$ of $O(\log n)$ bits and initially knows only its own identifier and the weights $w(e)$ of its incident edges $e$. 

\paragraph{The broadcast congested clique model.}
In this paper, we also consider the broadcast congested clique model~\cite{drucker2014power}, where the input graph is still $G=(V,E)$, but the communication network is the complete graph on $V$. Thus, unlike in $\CONGEST$, communication is not restricted to the edges of the input graph $G$. However, communication is restricted to \emph{broadcast}: in each round, each node sends a single $O(\log n)$-bit message, and the same message is received by all other nodes.  As in $\CONGEST$, computation proceeds in synchronous rounds with unlimited local computation.

\paragraph{Congestion and dilation.}
We use the following notions of congestion and dilation to describe the cost of scheduling a collection of distributed algorithms together. For a distributed algorithm $\mathcal{A}$ in the $\CONGEST$ model and an edge $e \in E$, let $c_{\mathcal{A}}(e)$ be the number of rounds in which $\mathcal{A}$ sends a message over $e$. The congestion of $\mathcal{A}$ is
$\congestion(\mathcal{A}) = \max_{e \in E} c_{\mathcal{A}}(e)$,
that is, the maximum number of times any single edge is used during the execution of $\mathcal{A}$. The dilation of $\mathcal{A}$, denoted $\dilation(\mathcal{A})$, is the number of communication rounds of $\mathcal{A}$.

We extend these definitions to a collection of algorithms $\{\mathcal{A}_1,\mathcal{A}_2,\ldots,\mathcal{A}_k\}$. We define
\begin{align*}
   \congestion(\mathcal{A}_1,\mathcal{A}_2,\ldots,\mathcal{A}_k)
    &=
    \max_{e \in E} \sum_{i=1}^k c_{\mathcal{A}_i}(e), \\
    \dilation(\mathcal{A}_1,\mathcal{A}_2,\ldots,\mathcal{A}_k)
    &=
    \max_{i \in [k]} \dilation(\mathcal{A}_i).
\end{align*}
Observe that we always have
\[
    \congestion(\mathcal{A}_1,\mathcal{A}_2,\ldots,\mathcal{A}_k)
    \leq
    \sum_{i=1}^k \congestion(\mathcal{A}_i).
\]
Suppose now that we want to execute $\mathcal{A}_1,\mathcal{A}_2,\ldots,\mathcal{A}_k$ concurrently, treating each algorithm as a black box. The congestion term is a lower bound on the number of rounds because an edge can carry only one $O(\log n)$-bit message per round, while the dilation term is a lower bound on the number of rounds because each individual algorithm must still complete its own communication pattern. Thus, any such schedule requires at least
\[
    \Omega(
        \congestion(\mathcal{A}_1,\mathcal{A}_2,\ldots,\mathcal{A}_k)
        +     \dilation(\mathcal{A}_1,\mathcal{A}_2,\ldots,\mathcal{A}_k)
    )
\]
communication rounds. The following well-known result shows that this lower bound can be matched up to a polylogarithmic factor.

\begin{proposition}[Scheduling distributed algorithms~\cite{ghaffari2015near}]\label{thm: dilation and congestion}
Any collection of independent distributed algorithms $\{\mathcal{A}_1, \mathcal{A}_2, \ldots, \mathcal{A}_k\}$ in the $\CONGEST$ model can be executed in 
\[\tilde{O}(\congestion(\mathcal{A}_1, \mathcal{A}_2, \ldots, \mathcal{A}_k)+\dilation(\mathcal{A}_1, \mathcal{A}_2, \ldots, \mathcal{A}_k))\]
rounds with high probability.
\end{proposition}

\section{Capturing Cycles via MPX Low-Diameter Decompositions}

In this section, we show how a variant of the MPX low-diameter decomposition can be used to extract cycles. In \Cref{subsect:mpx}, we present the variant of the MPX low-diameter decomposition used in our algorithm. In \Cref{subsect:capture}, we formalize what it means for a set of nodes to be captured by a cluster in the decomposition. In \Cref{subsect:capture-cycle-safe}, we study what the capture event implies when the captured set lies on a cycle. In \Cref{subsect:prob}, we analyze the probability of the capture event.

The connection to $\apxMWC$ is as follows. Informally, with an appropriate choice of parameters, if a cluster captures the node set of a minimum weight cycle, then the entire cycle is contained in the cluster and is not too far from the cluster center. This allows us to extract a cycle whose weight is within the desired approximation factor of $\OPT$.

\subsection{Decompositions via Random Shifts}\label{subsect:mpx}
We first recall the standard unweighted MPX low-diameter decomposition~\cite{miller2013parallel}, parameterized by a real number $\beta>0$. Each node $v$ independently samples a random shift $\delta_v$ from the exponential distribution $\exponential(\beta)$. Then each node $x$ joins the cluster of a node $v$ minimizing
\[
    \dist(v,x)-\delta_v.
\]
Equivalently, each node $v$ starts growing a BFS ball at time $-\delta_v$, and each node joins the cluster of the first growing ball that reaches it. 

\paragraph{Our MPX low-diameter decomposition.} In our setting, we need two modifications of this standard procedure. First, the input graph is weighted, so BFS is replaced by weighted SSSP. Second, as discussed in the technical overview, for the long-hop part of our algorithm for \Cref{thm:mainUB}, clusters are initiated only from a subset $S\subseteq V$ of skeleton nodes, rather than from all nodes. We therefore use the following weighted MPX low-diameter decomposition, parameterized by $(G,S,k,d)$.

Informally, $k$ corresponds to the parameter $k$ in \Cref{thm:mainUB,thm:parallelUB,thm:cliqueUB}, while $d$ should be viewed as a parameter for which $2d$ approximates $\OPT$. Both parameters, together with $|S|$, determine the choice of the exponential distribution parameter $\beta$.

Let $G=(V,E)$ be an undirected weighted graph, let $S\subseteq V$ be the set of candidate cluster centers {with $|S|\geq2$}, and let {$k\geq 1/2$} and $d>0$ be real numbers. Define
\[
    \beta=\frac{\ln |S|}{kd}
\]
and
\[
\delta_{\max}=\left(k+1+\frac{k}{\ln |S|}\right)d.
\]
Each node $s\in S$ independently samples
\[
    \delta_s \sim \exponential(\beta),
\]
and then truncates the shift by setting
\[
    \hat{\delta}_s=\min\{\delta_s,\delta_{\max}\}.
\]

Each node $v\in V$ joins the cluster of a node $s\in S$ minimizing
\[
    \dist_G(s,v)-\hat{\delta}_s.
\]
The resulting clusters form a partition of $V$, where each cluster is associated with a center in $S$.

Equivalently, each center $s\in S$ starts growing a ball at time $-\hat{\delta}_s$, where traversing an edge $e$ takes time $w(e)$. Each node joins the cluster of the first growing ball that reaches it. The truncation of the shifts ensures that all nodes can agree on a global starting time for the entire process that is no later than the time at which any cluster starts growing.

The following calculation shows that, with good probability, no shift is truncated, and hence the clustering based on the truncated shifts $\{\hat{\delta}_s\}_{s\in S}$ coincides with the clustering based on the original shifts $\{{\delta}_s\}_{s\in S}$.

\begin{lemma}[Probability of no truncation]\label{lem:no-truncation-prob} We have \[ \Pr\left[ \forall s\in S,\; \hat{\delta}_s=\delta_s \right] = \Pr\left[ \forall s\in S,\; \delta_s\leq \delta_{\max} \right] \geq 1-\frac{1}{e}|S|^{-1/k}. \] \end{lemma}

\begin{proof}
By the definition of the exponential distribution, for each $s\in S$,
\[
    \Pr[\delta_s>t]=e^{-\beta t}.
\]
Substituting the values of $\beta$ and $\delta_{\max}$, we get
\[
\begin{aligned}
    \Pr[\delta_s>\delta_{\max}]
    &=
    e^{
        -\frac{\ln |S|}{kd}
        \cdot
        \left(k+1+\frac{k}{\ln |S|}\right)d
    } \\
    &=
    e^{
        -\frac{k+1}{k}\ln |S|-1
    } \\
    &=
    \frac{1}{e}|S|^{-(1+1/k)}.
\end{aligned}
\]
Taking a union bound over all $s\in S$ gives
\[
    \Pr[\exists s\in S:\delta_s>\delta_{\max}]
    \leq
    |S|\cdot \frac{1}{e}|S|^{-(1+1/k)}
    =
    \frac{1}{e}|S|^{-1/k}.
\]
Thus, with probability at least $1-\frac{1}{e}|S|^{-1/k}$, no shift is truncated, and hence $\hat{\delta}_s=\delta_s$ for all $s\in S$.
\end{proof}

\subsection{Capturing a Node Set}\label{subsect:capture}

For any node $v\in V$ and any subset $U \subseteq V$, we write
\[
    \dist(v,U)=\min_{u\in U}\dist(v,u).
\]
 
We now define the key event used in the analysis.

\begin{definition}[Capture event]
For a non-empty subset $R\subseteq S$, let $\capture(R)$ be the event that there exists a node $s\in S$ such that
\[
    \dist(s,R)-\hat{\delta}_s
    <
    \dist(s',R)-\hat{\delta}_{s'} - d
\]
for every $s'\in S\setminus\{s\}$. When this happens, we say that \underline{the cluster centered at $s$ captures $R$}.
\end{definition}

Equivalently, $\capture(R)$ says the following. After the shifts are truncated, the ball grown from some center $s$ reaches the set $R$ first, and every other growing ball reaches $R$ more than $d$ time units later. 

In particular, if $R$ itself has weak diameter at most $d$, then all nodes in $R$ must join the cluster centered at $s$. 
 In the application with $S=V$, we take $R$ to be the node set of a minimum weight cycle $C^\star$; if $2d=\OPT$, then this set has weak diameter at most $d$, so $\capture(R)$ implies that the entire cycle $C^\star$ is within a cluster.

\begin{proposition}[Distance between cluster center and $R$]\label{prop:capture-close}
Suppose {$|R|\geq2$ and} $\capture(R)$ occurs, and let $s\in S$ be the center of a cluster that captures $R$. Then
\[
    \dist(s,R)< \hat{\delta}_s - d \leq \delta_{\max}-d.
\]
\end{proposition}

\begin{proof}
Since the cluster of $s$ captures $R$, by the definition of $\capture(R)$, we have
\[
    \dist(s,R)-\hat{\delta}_s
    <
    \dist(s',R)-\hat{\delta}_{s'}-d
\]
for every $s'\in S\setminus\{s\}$.

Choose any node {$r\in R\setminus\{s\}$, which exists since $|R|\geq2$}. Since $R\subseteq S$, we may apply the above inequality with $s'=r$. Using $\dist(r,R)=0$, we get
\[
    \dist(s,R)-\hat{\delta}_s
    <
    \dist(r,R)-\hat{\delta}_r-d
    =
    -\hat{\delta}_r-d
    \leq
    -d.
\]
Hence
\[
    \hat{\delta}_s>\dist(s,R)+d.
\]
By definition of the truncation, $\hat{\delta}_s\leq \delta_{\max}$. Therefore,
\[
    \dist(s,R)< \hat{\delta}_s - d \leq \delta_{\max}-d. \qedhere
\]
\end{proof}

\paragraph{Connection to $\apxMWC$.}
\Cref{prop:capture-close} is the key link between the capture event and the $(k+1+\epsilon)$-$\apxMWC$ problem. Suppose {$S=V$,} $2d=\OPT${,} and that the node set of a minimum weight cycle $C^\star$ is captured by some cluster. Then \Cref{prop:capture-close} implies that the cluster center $s$ is within distance at most $kd(1+O(1/\log n))$ of $C^\star$. Let $r$ be a node of $C^\star$ closest to $s$.

Consider the local SSSP tree rooted at $s$ that is computed during the MPX low-diameter decomposition. Since $C^\star$ is a cycle, at least one edge $e$ of $C^\star$ is not an edge of this tree. The cycle formed by $e$ together with the unique tree path between its endpoints has weight at most
\[
    2\dist(s,r)+w(C^\star)
    \leq
    2kd\left(1+O\left(\frac{1}{\log n}\right)\right)+2d
    =
    (k+1)\left(1+O\left(\frac{1}{\log n}\right)\right)\cdot \OPT.
\]
Thus, in this cluster, by examining all cycles consisting of exactly one non-tree edge with respect to this local SSSP tree, and tree edges otherwise, we can recover a $(k+1)(1+O(1/\log n))$-approximate minimum weight cycle.

\subsection{Safely Capturing a Cycle}\label{subsect:capture-cycle-safe}

For a node $v\in V$ and a center $s\in S$, define the \emph{arrival time} of the growing ball from $s$ at $v$ by
\[
    \arrival(s,v)=\dist(s,v)-\hat{\delta}_s.
\]
Thus, in the MPX low-diameter decomposition, each node $v$ joins the cluster centered at a node minimizing $\arrival(s,v)$.

Later, we will implement the MPX low-diameter decomposition using $(1+\epsilon)$-approximate SSSP rather than exact SSSP. To ensure that the cluster assignment remains unchanged for some nodes of interest, we introduce the following robustness notion.

\begin{definition}[$\epsilon$-safe nodes]\label{def:safety}
Let $\epsilon>0$. We say that a node $v$ is \emph{$\epsilon$-safe} for a center $s\in S$ if
\[
    \arrival(s,v)
    =
    \dist(s,v)-\hat{\delta}_s
    <
    \dist(s',v)-\hat{\delta}_{s'}-2\epsilon\delta_{\max}
    =
    \arrival(s',v)-2\epsilon\delta_{\max}
\]
for every $s'\in S\setminus\{s\}$.
\end{definition}

Equivalently, $v$ is $\epsilon$-safe for $s$ if the growing ball from $s$ reaches $v$ more than $2\epsilon\delta_{\max}$ time units before any other growing ball. Consequently, the assignment of $v$ to the cluster centered at $s$ is preserved even when the arrival times are computed only approximately, up to an additive error of at most $2\epsilon\delta_{\max}$.

For the long-hop part of the algorithm, we do not try to capture the entire minimum weight cycle. Instead, we sample skeleton nodes and aim to capture only the skeleton nodes lying on the cycle. We therefore need some additional terminology to describe how these skeleton nodes partition the cycle, and which parts of the cycle are certified to be $\epsilon$-safe once the capture event occurs.

\paragraph{Segments.}
Let $C$ be a cycle, and let
\[
    R=S\cap V(C).
\]
Assume that {$|R|\geq2$}. The nodes of $R$ partition $C$ into \emph{$R$-segments}: each $R$-segment is a maximal subpath of $C$ whose endpoints are in $R$ and whose internal nodes are not in $R$.

\paragraph{Visibility and antipodal segment.}
Suppose $R$ is captured by the cluster centered at $s\in S$. Let $q\in R$ be any node satisfying
\[
    \dist(s,q)=\dist(s,R).
\]
For an $R$-segment $P$, we say that $P$ is \emph{visible from $q$} if there exists an endpoint $r\in R$ of $P$ such that some $q$-$r$ subpath of $C$ containing $P$ has weight at most
\[
    d-2\epsilon\delta_{\max}.
\]
If $w(C)\leq 2(d-2\epsilon\delta_{\max})$, then at most one $R$-segment is not visible from $q$. If such a segment exists, we call it the \emph{antipodal segment}; otherwise, we say that there is no antipodal segment. See \Cref{fig:visibility-antipodal-minimal} for an illustration.

\begin{figure}[t]
\centering

\definecolor{Pone}{RGB}{31,119,180}
\definecolor{Ptwo}{RGB}{23,190,207}
\definecolor{Pthree}{RGB}{44,160,44}
\definecolor{Pfour}{RGB}{214,39,40}
\definecolor{Pfive}{RGB}{255,127,14}
\definecolor{Psix}{RGB}{148,103,189}

\begin{tikzpicture}[
    scale=1,
    every node/.style={font=\normalsize},
    edge/.style={gray!35, line width=0.6pt},
    P1/.style={Pone, line width=1.9pt},
    P2/.style={Ptwo, line width=1.9pt},
    P3/.style={Pthree, line width=1.9pt},
    P4/.style={Pfour, line width=2.5pt},
    P5/.style={Pfive, line width=1.9pt},
    P6/.style={Psix, line width=1.9pt},
    vtx/.style={circle, fill=black, inner sep=1.2pt},
    Rvtx/.style={circle, draw=black, fill=white, thick, inner sep=3.8pt},
    Qvtx/.style={circle, draw=black, fill=orange!20, very thick, inner sep=4.0pt},
    center/.style={circle, draw=black, fill=gray!15, inner sep=2pt},
    plabel/.style={fill=white, inner sep=1.2pt}
]

% -------------------------------------------------
% A concrete 20-node cycle C
% -------------------------------------------------
\def\r{3.65}
\def\rlab{4.45}
\def\pin{3}

\foreach \i in {0,...,19}{
    \pgfmathsetmacro{\ang}{-90 + 18*\i}
    \coordinate (v\i) at (\ang:\r);
    \coordinate (vlab\i) at (\ang:\rlab);
}

% Base cycle
\foreach \i in {0,...,18}{
    \pgfmathtruncatemacro{\j}{\i+1}
    \draw[edge] (v\i) -- (v\j);
}
\draw[edge] (v19) -- (v0);

% -------------------------------------------------
% R-segments
% R = {v_0, v_3, v_6, v_9, v_12, v_16}
% -------------------------------------------------

% P_1: v_0 -> v_3
\foreach \a/\b in {0/1,1/2,2/3}{
    \draw[P1] (v\a) -- (v\b);
}

% P_2: v_3 -> v_6
\foreach \a/\b in {3/4,4/5,5/6}{
    \draw[P2] (v\a) -- (v\b);
}

% P_3: v_6 -> v_9
\foreach \a/\b in {6/7,7/8,8/9}{
    \draw[P3] (v\a) -- (v\b);
}

% P_4: v_9 -> v_12, antipodal segment
\foreach \a/\b in {9/10,10/11,11/12}{
    \draw[P4] (v\a) -- (v\b);
}

% P_5: v_12 -> v_16
\foreach \a/\b in {12/13,13/14,14/15,15/16}{
    \draw[P5] (v\a) -- (v\b);
}

% P_6: v_16 -> v_0
\foreach \a/\b in {16/17,17/18,18/19,19/0}{
    \draw[P6] (v\a) -- (v\b);
}

% Vertices and labels
\foreach \i in {0,...,19}{
    \node[vtx] at (v\i) {};
}
\foreach \i in {1,...,19}{
    \node[font=\small] at (vlab\i) {$v_{\i}$};
}

% Highlight R vertices
\foreach \i in {3,6,9,12,16}{
    \node[Rvtx] at (v\i) {};
}
\node[Qvtx] at (v0) {};

% Cluster center and closest R-node
\node[center] (s) at (0,-5.25) {$s$};
\draw[dashed, gray!70] (s) -- (v0);
\node[plabel, font=\small] at (vlab0) {$q=v_0$};

% Segment labels inside the cycle
\node[plabel, text=Pone]   at (-63:\pin) {$P_1$};
\node[plabel, text=Ptwo]   at ( -9:\pin) {$P_2$};
\node[plabel, text=Pthree] at ( 45:\pin) {$P_3$};

\node[plabel, text=Pfour] at (99:\pin) {$P_4$};
\node[text=Pfour, font=\small] at (99:2.3) {antipodal};

\node[plabel, text=Pfive] at (162:\pin) {$P_5$};
\node[plabel, text=Psix]  at (234:\pin) {$P_6$};

\end{tikzpicture}

\caption{
Consider a 20-node cycle $C$ with unit-weight edges: the set $R=\{v_0,v_3,v_6,v_9,v_{12},v_{16}\}$ partitions $C$ into six $R$-segments; with $q=v_0$ and $w(C)=2(d-2\epsilon\delta_{\max})=20$, the segment $P_4$ is the antipodal segment.
}
\label{fig:visibility-antipodal-minimal}
\end{figure}

\begin{proposition}[Safely capturing a cycle]\label{lem:capture-cycle-safe}
Let $C$ be a cycle such that
\[
    w(C)\leq 2(d-2\epsilon\delta_{\max}),
    \quad\ 
    R=S\cap V(C),\quad\ \text{and}\quad\ |R|\geq2.
\]
Suppose $R$ is captured by the cluster centered at $s\in S$. Then every node of $C$, except possibly for the internal nodes of the antipodal segment, is $\epsilon$-safe for $s$.
In particular, if $R=V(C)$, then every node of $C$ is $\epsilon$-safe for $s$.
\end{proposition}

\begin{proof}
Let $q\in R$ satisfy $\dist(s,q)=\dist(s,R)$. Since $R$ is captured by the cluster centered at $s$, for every $s'\in S\setminus\{s\}$ we have
\begin{equation}
    \dist(s,R)-\hat{\delta}_s
    <
    \dist(s',R)-\hat{\delta}_{s'}-d.
    \label{eq:capture-condition-safe}
\end{equation}

We first show that every node on every $R$-segment visible from $q$ is $\epsilon$-safe for $s$. Let $P$ be an $R$-segment visible from $q$, and let $u$ be any node on $P$. By visibility, there is an endpoint $r\in R$ of $P$ and a $q$-$r$ subpath $Q$ of $C$ containing $P$ such that
\begin{equation}
    w(Q)\leq d-2\epsilon\delta_{\max}.
    \label{eq:visible-subpath-short}
\end{equation}
Since $u$ lies on $Q$, this implies
\begin{equation}
    \dist_Q(q,u)+\dist_Q(u,r)
    =
    w(Q)
    \leq
    d-2\epsilon\delta_{\max}.
    \label{eq:path-distance-sum}
\end{equation}
Moreover, since $\dist(s,q)=\dist(s,R)$, we have
\begin{equation}
    \dist(s,u)
    \leq
    \dist(s,q)+\dist_Q(q,u)
    =
    \dist(s,R)+\dist_Q(q,u).
    \label{eq:s-to-u-via-q}
\end{equation}
Also, for every $s'\in S\setminus\{s\}$,
\begin{equation}
    \dist(s',R)
    \leq
    \dist(s',r)
    \leq
    \dist(s',u)+\dist_Q(u,r).
    \label{eq:sp-to-R-via-u}
\end{equation}

Now fix any $s'\in S\setminus\{s\}$. Combining the above inequalities, we get
\[
\begin{aligned}
    \dist(s,u)-\hat{\delta}_s
    &\leq
    \dist(s,R)+\dist_Q(q,u)-\hat{\delta}_s
    && \text{by \eqref{eq:s-to-u-via-q}} \\
    &<
    \dist(s',R)-\hat{\delta}_{s'}-d+\dist_Q(q,u)
    && \text{by \eqref{eq:capture-condition-safe}} \\
    &\leq
    \dist(s',u)+\dist_Q(u,r)-\hat{\delta}_{s'}-d+\dist_Q(q,u)
    && \text{by \eqref{eq:sp-to-R-via-u}} \\
    &=
    \dist(s',u)-\hat{\delta}_{s'}
    -\left(d-\dist_Q(q,u)-\dist_Q(u,r)\right) \\
    &\leq
    \dist(s',u)-\hat{\delta}_{s'}-2\epsilon\delta_{\max}
    && \text{by \eqref{eq:path-distance-sum}}.
\end{aligned}
\]
Thus $u$ is $\epsilon$-safe for $s$.

By the definition of the antipodal segment, all nodes of $C$, except possibly for the internal nodes of the antipodal segment, lie on $R$-segments visible from $q$. Hence all such nodes are $\epsilon$-safe for $s$.

Finally, if $R=V(C)$, then every $R$-segment has no internal nodes. Thus even if an antipodal segment exists, it has no internal nodes to exclude, and every node of $C$ is $\epsilon$-safe for $s$.
\end{proof}

The full strength of \Cref{lem:capture-cycle-safe} is needed only for the long-hop case in the proof of \Cref{thm:mainUB}. In all other applications, we use the simpler case $R=V(C)$, where \Cref{lem:capture-cycle-safe} implies that every node of $C$ is $\epsilon$-safe for the center $s$.

\subsection{Probability Analysis}\label{subsect:prob}

In the following discussion, we show that the capture event occurs with probability $\Omega\left(|S|^{-1/k}\right)$. Consequently, $\tilde{O}\left(|S|^{1/k}\right)$ independent repetitions suffice to ensure that the event occurs in at least one decomposition with high probability.

The following lemma, a simple consequence of the memoryless property of the exponential distribution, is a key ingredient in analyzing the probability of the capture event. Variants of this lemma appear already in the original MPX paper~\cite{miller2013parallel} and have since become standard tools in analyses of MPX low-diameter decompositions. For example, it follows as a special case of \cite[Lemma 3.6]{haeupler2016faster} by setting $k=2$. We therefore omit the proof.

\begin{lemma}[Gap between the 1st and 2nd arrivals~\cite{miller2013parallel}]
\label{lem:mpx-gap}
Let $S$ be a finite non-empty set, and let $\{a_s\}_{s\in S}$ be fixed real numbers.
For each $s\in S$, let $\delta_s\sim\exponential(\beta)$ be sampled independently.
Then, for every $d\geq 0$,
\[
    \Pr\left[
        \exists s\in S\ \text{such that}\ 
        a_s-\delta_s
        <
        a_{s'}-\delta_{s'}-d
        \text{ for every } s'\in S\setminus\{s\}
    \right]
    \geq
    e^{-\beta d}.
\]
\end{lemma}

We are ready to lower bound the probability of the capture event.

\begin{proposition}[Probability of the capture event]\label{lem:capture-probability}
For every non-empty subset $R\subseteq S$,
\[
    \Pr[\capture(R)]
    \geq
    \left(1-\frac{1}{e}\right)|S|^{-1/k}.
\]
\end{proposition}

\begin{proof}
By setting $a_s=\dist(s,R)$ for each $s\in S$, \Cref{lem:mpx-gap} implies that, with probability at least
\[
    e^{-\beta d}=|S|^{-1/k},
\]
there exists a node $s\in S$ such that
\[
    \dist(s,R)-\delta_s
    <
     \dist(s',R)-\delta_{s'} - d
\]
for every $s'\in S\setminus\{s\}$.

Let $\mathcal{E}_{\mathrm{gap}}(R)$ denote this event, and let $\mathcal{E}_{\mathrm{no\text{-}trunc}}$ denote the event that no shift is truncated, i.e.,
\[
    \delta_s\leq \delta_{\max}
\]
for every $s\in S$. By \Cref{lem:no-truncation-prob},
\[
    \Pr[\overline{\mathcal{E}_{\mathrm{no\text{-}trunc}}}]
    \leq
    \frac{1}{e}|S|^{-1/k}.
\]
Therefore,
\[
\begin{aligned}
    \Pr[\mathcal{E}_{\mathrm{gap}}(R)\cap \mathcal{E}_{\mathrm{no\text{-}trunc}}]
    &\geq
    \Pr[\mathcal{E}_{\mathrm{gap}}(R)]
    -
    \Pr[\overline{\mathcal{E}_{\mathrm{no\text{-}trunc}}}] \\
    &\geq
    |S|^{-1/k}
    -
    \frac{1}{e}|S|^{-1/k} \\
    &=
    \left(1-\frac{1}{e}\right)|S|^{-1/k}.
\end{aligned}
\]

On the event $\mathcal{E}_{\mathrm{no\text{-}trunc}}$, we have
\[
    \hat{\delta}_s=\delta_s
\]
for every $s\in S$. Hence, if $\mathcal{E}_{\mathrm{gap}}(R)$ also occurs, then there exists $s\in S$ such that
\[
    \dist(s,R)-\hat{\delta}_s
    <
    -d+\dist(s',R)-\hat{\delta}_{s'}
\]
for every $s'\in S\setminus\{s\}$. This is exactly the event $\capture(R)$. Therefore,
\[
    \Pr[\capture(R)]
    \geq
    \Pr[\mathcal{E}_{\mathrm{gap}}(R)\cap \mathcal{E}_{\mathrm{no\text{-}trunc}}]
    \geq
    \left(1-\frac{1}{e}\right)|S|^{-1/k}. \qedhere
\]
\end{proof}

\section{Implementation via Approximate SSSP Computation}

We next explain how to implement the MPX low-diameter decomposition using approximate SSSP. Besides distance estimates, we also need the approximate SSSP algorithm to output a tree, which will be used both to define the clusters and to support cycle detection within each cluster. In \Cref{subsect:treelike}, we review the tree-like version of approximate SSSP that provides this additional structure. In \Cref{subsect:decomposition}, we use it to define the $(1+\epsilon)$-approximate MPX low-diameter decomposition $\LDD(G,S,k,d,\epsilon)$ and prove that $\epsilon$-safe nodes are assigned to the intended clusters. Finally, in \Cref{subsect:cycle}, we show how the resulting tree is used in the cycle detection step of our $\apxMWC$ algorithm.

\subsection{Tree-Like Distance Approximation}\label{subsect:treelike}

We begin by defining tree-like $(1+\epsilon)$-approximate SSSP.

\begin{definition}[$(1+\epsilon)$-approximate SSSP]
Let $G=(V,E,w)$ be a connected undirected graph with positive edge weights, let $x\in V$ be a source node, and let $\epsilon>0$. A function $\tilde{\dist}(x,\cdot):V\to\mathbb{R}_{\geq 0}$ is a $(1+\epsilon)$-approximate SSSP distance estimate from $x$ if, for every node $v\in V$,
\[
    \dist_G(x,v)
    \leq
    \tilde{\dist}(x,v)
    \leq
    (1+\epsilon)\dist_G(x,v).
\]
\end{definition}

\begin{definition}[Tree-like distance estimate]
Let $G=(V,E,w)$ be a connected undirected graph with positive edge weights, and let $x\in V$ be a source node. A distance estimate $\tilde{\dist}(x,\cdot):V\to\mathbb{R}_{\geq 0}$ is tree-like with respect to $x$ if $\tilde{\dist}(x,x)=0$ and, for every node $v\neq x$, there exists an edge $\{u,v\}\in E$ such that
\[
    \tilde{\dist}(x,u)
    \leq
    \tilde{\dist}(x,v)-w(u,v).
\]
\end{definition}

The tree-like condition allows us to recover a spanning tree $T$ rooted at the source node $x$: for each node $v\neq x$, choose one neighbor $u$ satisfying $\tilde{\dist}(x,u)
    \leq
    \tilde{\dist}(x,v)-w(u,v)$,
and make $u$ the parent of $v$. Since edge weights are positive, the estimates strictly decrease along parent pointers, so this process defines a spanning tree rooted at $x$. Moreover, if $\tilde{\dist}(x,\cdot)$ is also a $(1+\epsilon)$-approximate SSSP distance estimate, then the tree distance from $x$ to every node $v$ is at most $\tilde{\dist}(x,v)$, and hence at most $(1+\epsilon)\dist_G(x,v)$. Thus $T$ is a $(1+\epsilon)$-approximate SSSP tree.

We use the transformation of \citet{RozhonHMGZ23}, which shows how to convert any $(1+\epsilon)$-approximate SSSP algorithm into one with the tree-like property using only polylogarithmically many calls to the original algorithm. Since our bounds suppress polylogarithmic factors, we may essentially assume tree-likeness for free.

The tree-like property is important because it provides an approximate shortest-path tree, not merely approximate distance labels. This allows us to use the same cycle extraction argument as in the exact MPX low-diameter decomposition; see \Cref{subsect:capture}.

\subsection{The Decomposition \texorpdfstring{$\LDD(G,S,k,d,\epsilon)$}{LDD(G,S,k,d,epsilon)}}\label{subsect:decomposition}
We now define the approximate implementation of the  MPX low-diameter decomposition using tree-like $(1+\epsilon)$-approximate SSSP. Recall that in the  decomposition algorithm, each center $s\in S$ starts growing a ball at time $-\hat{\delta}_s$. Equivalently, the arrival time of $s$ at a node $v$ is
\[
    \arrival(s,v) = \dist_G(s,v)-\hat{\delta}_s.
\]
To implement this using SSSP with positive edge weights, we add a virtual super source $x$ and connect it to each center $s\in S$ by a virtual edge of weight
\[
    2\delta_{\max}-\hat{\delta}_s.
\]
Since $\hat{\delta}_s\in[0,\delta_{\max}]$, these virtual edges have weights in $[\delta_{\max},2\delta_{\max}]$. Moreover, minimizing
\[
    2\delta_{\max}-\hat{\delta}_s+\dist_G(s,v)
\]
over $s\in S$ is equivalent to minimizing the original arrival time $\arrival(s,v) = \dist_G(s,v)-\hat{\delta}_s$.

Let $G^+$ denote the augmented graph obtained from $G$ by adding the virtual super source $x$ and the virtual edges $\{x,s\}$ of weight $2\delta_{\max}-\hat{\delta}_s$ for all $s\in S$. We write
\[
    \LDD(G,S,k,d,\epsilon)
\]
for the decomposition obtained by running tree-like $(1+\epsilon)$-approximate SSSP from $x$ in $G^+$. 

Specifically, each node $v\in V$ joins the cluster of the center $s\in S$ that is the second node on the $x$-$v$ path in the resulting spanning tree $T$ of $G^+$. For each center $s\in S$ whose cluster is non-empty, let $T_s$ denote the subtree of $T$ rooted at $s$. We call $T_s$ the \emph{local approximate SSSP tree} of the cluster centered at $s$. Observe that $T_s$ spans exactly the nodes assigned to the cluster centered at $s$.

\paragraph{Remark on edge weights and the safety margin.}
In this work, as in many distributed graph algorithms, we assume that edge weights are polynomially bounded positive integers. This ensures that each edge weight can be represented using $O(\log n)$ bits.

This assumption motivates our choice of assigning weight $2\delta_{\max}-\hat{\delta}_s$, rather than the more natural $\delta_{\max}-\hat{\delta}_s$, to the edge between the virtual super source and each center $s\in S$. The latter choice could produce a zero-weight edge, whereas the former guarantees that every virtual edge has weight in the range $[\delta_{\max},2\delta_{\max}]$. Since this simply adds the same offset $\delta_{\max}$ to every shifted distance, the resulting decomposition is unchanged.

This choice is also compatible with the margin $2\epsilon\delta_{\max}$ in the definition of $\epsilon$-safety (\Cref{def:safety}). {Indeed, for every node $v$ to which we apply the safety guarantee, we already have}
\[
    \dist_G(s,v)-\hat{\delta}_s\leq 0.
\]
Hence the exact distance from the virtual super source to $v$ through $s$ is at most
\[
    2\delta_{\max}-\hat{\delta}_s+\dist_G(s,v)
    \leq
    2\delta_{\max}.
\]
Therefore, a $(1+\epsilon)$-approximate SSSP computation incurs an additive error of at most $2\epsilon\delta_{\max}$ on these relevant distances, exactly matching the safety margin.

Finally, the sampled shifts are real numbers, so the virtual
edge weights may also be real rather than integral.
To obtain polynomially bounded integer weights, we may round all edge weights to a sufficiently fine granularity and then scale them. By choosing a sufficiently small inverse-polynomial granularity, every relevant distance changes by at most an additive $1/\poly(n)$ term, which can be absorbed into the $(1+\epsilon)$ approximation factor. 

Specifically, we may replace each augmented edge weight $\ell(e)$ by
$\widehat\ell(e)=\lceil n^3\ell(e)\rceil$.
In our applications, 
$1/2 \leq \ell(e)\in O(nw_{\max}\log n)$, so the scaled weights
are at most $O(n^4w_{\max}\log n) \subseteq n^{O(1)}$.

Run tree-like approximate SSSP with accuracy $\epsilon/2$ on these integer weights and divide the returned distance estimates by $n^3$. The resulting approximation factor is at most $(1+\epsilon/2)(1+ O(n^{-3}))$, which is less than $1+\epsilon$ for sufficiently large $n$, as we have $\epsilon\in\log^{-O(1)}n$. 

The discussion above suggests that $\epsilon$-safety should be enough to preserve the cluster assignment under the approximate SSSP implementation. The next lemma makes this formal.

\begin{lemma}[Safety preserves cluster assignment]\label{lem:safe-preserves-cluster}
Let $v\in V$ be $\epsilon$-safe for a center $s\in S$. Suppose 
\[
    \dist_G(s,v)-\hat{\delta}_s \leq 0.
\]
Then $v$ is assigned to the cluster centered at $s$ in $\LDD(G,S,k,d,\epsilon)$, regardless of the choice of the tree-like $(1+\epsilon)$-approximate SSSP tree.
\end{lemma}

\begin{proof}
Let $G^+$ be the augmented graph with virtual super source $x$, where the edge $\{x,s'\}$ has weight $2\delta_{\max}-\hat{\delta}_{s'}$ for each $s'\in S$. For each center $s'\in S$, write
\[
    D_{s'}(v)
    =
    2\delta_{\max}-\hat{\delta}_{s'}+\dist_G(s',v).
\]
This is the weight of the shortest $x$-$v$ path in $G^+$ that goes through $s'$.

Since $v$ is $\epsilon$-safe for $s$, for every $s'\neq s$ we have
\begin{equation}
     D_s(v)<D_{s'}(v)-2\epsilon\delta_{\max}. \label{eq:safe}  
\end{equation}

Thus, among all centers, the exact shortest $x$-$v$ path in $G^+$ goes through $s$, with a margin of more than $2\epsilon\delta_{\max}$ over every other center. In particular,
\[
    \dist_{G^+}(x,v)=D_s(v).
\]
Moreover, by the assumption $\dist_G(s,v)-\hat{\delta}_s\leq 0$,
\[
    D_s(v)
    =
    2\delta_{\max}+\dist_G(s,v)-\hat{\delta}_s
    \leq
    2\delta_{\max}.
\]

Now consider the tree $T$ produced by the tree-like $(1+\epsilon)$-approximate SSSP computation from $x$. Since the exact shortest $x$-$v$ distance in $G^+$ is $D_s(v)$, the approximation guarantee implies that the weight of the $x$-$v$ path in $T$ is at most $(1+\epsilon)D_s(v)$. Moreover, as shown above, $D_s(v)\leq 2\delta_{\max}$. Therefore, \[ \dist_T(x,v) \leq (1+\epsilon)D_s(v) \leq D_s(v)+2\epsilon\delta_{\max}. \]

Suppose, for contradiction, that $v$ is assigned to a cluster centered at some $s'\neq s$. Then the $x$-$v$ path in $T$ goes through $s'$, so its weight is at least $D_{s'}(v)$. Therefore,
\[
    D_{s'}(v)
    \leq
    D_s(v)+2\epsilon\delta_{\max},
\]
contradicting \eqref{eq:safe}. 
Hence the $x$-$v$ path in $T$ must go through $s$, and so $v$ is assigned to the cluster centered at $s$.
\end{proof}

\subsection{Cycle Detection Inside Clusters}\label{subsect:cycle}

For every node $v\in V$ assigned to the cluster centered at $s\in S$, define
\[
    \tilde{\dist}(s,v)
    =
    \tilde{\dist}(x,v)-\tilde{\dist}(x,s).
\]
We use this quantity as the approximate distance from $s$ to $v$ within its cluster. By the tree-like property, if $p(z)$ is the parent of a node $z$ in the tree $T$, then
\[
    \tilde{\dist}(x,p(z))
    \leq
    \tilde{\dist}(x,z)-w(p(z),z).
\]
Therefore, along the unique $s$-$v$ path in $T_s$, the values $\tilde{\dist}(x,\cdot)$ decrease toward $s$ by at least the corresponding edge weights. Telescoping over this path gives
\[
    \dist_{T_s}(s,v)
    \leq
    \tilde{\dist}(x,v)-\tilde{\dist}(x,s)
    =
    \tilde{\dist}(s,v).
\]
In particular, $\tilde{\dist}(s,v)\geq \dist_{T_s}(s,v) \geq \dist_G(s,v)$.

\paragraph{Cycle detection step.}
We run the \emph{cycle detection step} over all clusters of $\LDD(G,S,k,d,\epsilon)$. For each cluster centered at $s$, it enumerates all edges $e=\{u,v\}$ such that both endpoints $u$ and $v$ lie in the cluster, but $e$ is not an edge of the local approximate SSSP tree $T_s$. For each such edge $e$, the algorithm forms the cycle consisting of $e$ together with the unique $u$-$v$ path in $T_s$. It assigns this cycle the \emph{weight estimate}
\[
    \tilde{\dist}(s,u)+\tilde{\dist}(s,v)+w(u,v).
\]

The estimate assigned to each constructed cycle is an upper bound on its true weight. Indeed, the unique $u$-$v$ path in $T_s$ is contained in the union of the $s$-$u$ path and the $s$-$v$ path in $T_s$, and hence has weight at most
\[
    \dist_{T_s}(s,u)+\dist_{T_s}(s,v)
    \leq
    \tilde{\dist}(s,u)+\tilde{\dist}(s,v).
\]
Therefore, the true weight of the cycle formed by $e=\{u,v\}$ and the $u$-$v$ tree path is indeed at most the weight estimate $\tilde{\dist}(s,u)+\tilde{\dist}(s,v)+w(u,v)$. The following result gives a sufficient condition under which the cycle detection step returns a cycle of small weight estimate.

\begin{proposition}[A non-tree edge $\rightarrow$ a cycle of small weight]\label{lem:safe-nontree-edge-gives-cycle}
Run $\LDD(G,S,k,d,\epsilon)$, and let $C$ be a cycle such that
\[
    w(C)\leq 2(d-2\epsilon\delta_{\max}),
    \quad\
    R=S\cap V(C), \quad\ \text{and} \quad\ |R|\geq2.
\]
Suppose $R$ is captured by the cluster centered at $s\in S$. Let $e=\{u,v\}$ be an edge of $C$ such that both $u$ and $v$ are $\epsilon$-safe for $s$. If $e$ is not an edge of the local approximate SSSP tree $T_s$, then the cycle detection step finds a cycle with weight estimate at most $2\delta_{\max}$.
\end{proposition}

\begin{proof}
Let $q\in R$ be a node closest to $s$, so
\begin{equation}\label{eq:q-closest}
    \dist_G(s,q)=\dist_G(s,R).
\end{equation}
By \Cref{prop:capture-close},
\begin{equation}\label{eq:capture-close-lower}
    \hat{\delta}_s>\dist_G(s,R)+d
\end{equation}
and
\begin{equation}\label{eq:capture-close-upper}
    \dist_G(s,q) = \dist_G(s,R) <\hat{\delta}_s-d   \leq \delta_{\max}-d.
\end{equation}

We first show that both endpoints of $e$ are assigned to the cluster centered at $s$. Since $u$ and $v$ lie on $C$, each of them is within distance at most $w(C)/2$ from $q$ along the cycle. Hence, for each $y\in\{u,v\}$,
\begin{align}
    \dist_G(s,y)-\hat{\delta}_s
    &<
    \left(\dist_G(s,q)+\frac{w(C)}{2}\right)
    -
    \left(\dist_G(s,R)+d\right) \notag\\
    &=
    \frac{w(C)}{2}-d \notag\\
    &\leq
    -2\epsilon\delta_{\max}
    \leq 0.
    \label{eq:uy-negative-shift}
\end{align}
Here the first inequality uses \eqref{eq:capture-close-lower}, the equality uses \eqref{eq:q-closest}, and the final inequality uses the assumption
\(w(C)\leq 2(d-2\epsilon\delta_{\max})\).

Since $u$ and $v$ are $\epsilon$-safe for $s$, \Cref{lem:safe-preserves-cluster} and \eqref{eq:uy-negative-shift} imply that both $u$ and $v$ are assigned to the cluster centered at $s$.

We next bound the approximate distances from $s$ to $u$ and $v$ used by the cycle detection step. Let $G^+$ be the augmented graph with virtual super source $x$. For each $y\in\{u,v\}$, the exact $x$-$y$ distance through $s$ is
\begin{equation}\label{eq:exact-through-s}
    2\delta_{\max}-\hat{\delta}_s+\dist_G(s,y).
\end{equation}
By \eqref{eq:uy-negative-shift}, this quantity is at most $2\delta_{\max}$. Since the tree-like SSSP tree is a $(1+\epsilon)$-approximate SSSP tree, and since 
\[
    {\tilde{\dist}(x,s)\geq2\delta_{\max}-\hat{\delta}_s},
\]
we have, by the definition of our distance estimate $\tilde{\dist}(\cdot,\cdot)$,
\begin{align}
    \tilde{\dist}(s,y)
    &=
    \tilde{\dist}(x,y)-\tilde{\dist}(x,s) \notag\\
    &\leq
    (1+\epsilon)\left(2\delta_{\max}-\hat{\delta}_s+\dist_G(s,y)\right)
    -
    (2\delta_{\max}-\hat{\delta}_s) \notag\\
    &=
    \dist_G(s,y)
    +
    \epsilon\left(2\delta_{\max}-\hat{\delta}_s+\dist_G(s,y)\right) \notag\\
    &<
    \dist_G(s,y)+2\epsilon\delta_{\max},
    \label{eq:approx-distance-bound}
\end{align}
where the last inequality uses \eqref{eq:uy-negative-shift}.

Since $e\notin E(T_s)$, the cycle detection step considers the cycle formed by $e$ together with the unique $u$-$v$ path in $T_s$. By definition of the weight estimate and by \eqref{eq:approx-distance-bound}, the assigned estimate is  
\begin{align}
    \tilde{\dist}(s,u)+\tilde{\dist}(s,v)+w(u,v)
    &<
    \dist_G(s,u)+\dist_G(s,v)+w(u,v)+4\epsilon\delta_{\max}.
    \label{eq:estimated-cycle-bound}
\end{align}

It remains to relate the right-hand side to the weight of $C$. Since $u,v,q\in V(C)$ and $e=\{u,v\}$ is an edge of $C$,
\begin{equation}\label{eq:cycle-metric-bound}
    \dist_G(s,u)+\dist_G(s,v)+w(u,v)
    \leq
    2\dist_G(s,q)+w(C).
\end{equation}
Combining \eqref{eq:estimated-cycle-bound} and \eqref{eq:cycle-metric-bound}, the weight estimate of the considered cycle is at most
\begin{equation}\label{eq:before-final-bound}
    2\dist_G(s,q)+w(C)+4\epsilon\delta_{\max}.
\end{equation}
Finally, by \eqref{eq:capture-close-upper} and the assumption \(w(C)\leq 2(d-2\epsilon\delta_{\max})\), 
the bound in \eqref{eq:before-final-bound} is smaller than
\[
    2(\delta_{\max}-d)+2(d-2\epsilon\delta_{\max})+4\epsilon\delta_{\max}
    =
    2\delta_{\max}.
\]
Therefore the cycle detection step finds a cycle with weight estimate at most $2\delta_{\max}$.
\end{proof}

\Cref{lem:safe-nontree-edge-gives-cycle} captures the only property of the cycle detection step that we will need: if $\capture(R)$ occurs and the cycle contains a non-tree edge whose endpoints are $\epsilon$-safe for the cluster center, then the algorithm detects a cycle of weight estimate at most $2\delta_{\max}$.

We first apply \Cref{lem:safe-nontree-edge-gives-cycle} in the special case $S=V$. Then, for a minimum weight cycle $C^\star$, we have $R=S\cap V(C^\star)=V(C^\star)$. Thus, when $\capture(V(C^\star))$ occurs, all nodes of $C^\star$ are $\epsilon$-safe for the capturing center by \Cref{lem:capture-cycle-safe}. Moreover, since the local approximate SSSP tree is a tree, it cannot contain all edges of $C^\star$; hence $C^\star$ contains a non-tree edge whose endpoints are $\epsilon$-safe. Therefore, \Cref{lem:safe-nontree-edge-gives-cycle} applies directly.

Later, in the long-hop part of the algorithm for \Cref{thm:mainUB}, we apply \Cref{lem:safe-nontree-edge-gives-cycle} with $S$ equal to the sampled skeleton nodes. The difficulty there is that the antipodal segment may contain nodes that are not guaranteed to be $\epsilon$-safe by \Cref{lem:capture-cycle-safe}, so a suitable non-tree edge may not exist.

The following proposition records the resulting guarantee for the case $S=V$, including both its success probability and approximation ratio. 

\begin{proposition}[Approximation ratio and success probability]\label{lem:ldd-cycle-detection-short}
Let $C^\star$ be a minimum weight cycle. Run $\LDD(G,S,k,d,\epsilon)$ with $S=V$, followed by the cycle detection step. Suppose $\epsilon \in O(1/\log n)$ and
\[
    2d\left(1-O\left(\frac{1}{\log n}\right)\right)
    \leq
    \OPT
    \leq
    2(d-2\epsilon\delta_{\max}).
\]
With probability $\Omega(n^{-1/k})$, the cycle detection step finds a cycle whose weight estimate is at most
\[
    (k+1)\left(1+O\left(\frac{1}{\log n}\right)\right)\OPT.
\]
\end{proposition}

\begin{proof}
Since $S=V$, we have
\[
    R=S\cap V(C^\star)=V(C^\star).
\]
By \Cref{lem:capture-probability}, the event $\capture(V(C^\star))$ occurs with probability $\Omega(n^{-1/k})$. Condition on this event, and let $s\in S$ be the center of the cluster that captures $V(C^\star)$.

By \Cref{lem:capture-cycle-safe}, applied with $R=V(C^\star)$ and using the assumption
\[
    w(C^\star)=\OPT\leq 2(d-2\epsilon\delta_{\max}),
\]
every node of $C^\star$ is $\epsilon$-safe for $s$.

We claim that at least one edge of $C^\star$ is not an edge of the local approximate SSSP tree $T_s$. Indeed, by \Cref{lem:safe-preserves-cluster}, every node of $C^\star$ is assigned to the cluster centered at $s$. Since $T_s$ is a tree, it cannot contain all edges of the cycle $C^\star$. Thus, there exists an edge $e=\{u,v\}$ of $C^\star$ that is not an edge of $T_s$.

Both endpoints $u$ and $v$ are $\epsilon$-safe for $s$, and $e$ is a non-tree edge. Therefore, by \Cref{lem:safe-nontree-edge-gives-cycle}, the cycle detection step finds a cycle with weight estimate at most $2\delta_{\max}$.

Since $S=V$, we have $\delta_{\max} \in (k+1)(1+O(1/\log n))d$. Together with
\[
    2d\leq \left(1+O\left(\frac{1}{\log n}\right)\right)\OPT,
\]
this implies
\[
    2\delta_{\max}
    \leq
    (k+1)\left(1+O\left(\frac{1}{\log n}\right)\right)\OPT.
\]
This proves the {proposition}.
\end{proof}

\section{Parallel and Distributed Algorithms}\label{sec:upperbounds}

In this section, we turn the one-run guarantee (\Cref{lem:ldd-cycle-detection-short}) from the previous section into parallel and distributed algorithms for $\apxMWC$. The first step is to guess the scale parameter $d$ and amplify the success probability by independent repetitions; this is done in \Cref{subsect:reduction}. We then describe the implementation of the full procedure in \Cref{subsect:impl}, where all steps except for global aggregation and broadcasting are reduced to $(1+\epsilon)$-approximate SSSP computations. These ingredients are finally instantiated in the work-depth, broadcast congested clique, and $\CONGEST$ models in \Cref{subsect:upperbounds}.

\subsection{Parameter Guessing and Success Probability Amplification}\label{subsect:reduction}

To turn the one-run guarantee of \Cref{lem:ldd-cycle-detection-short} into {a} $(k+1)\left(1+O\left(\frac{1}{\log n}\right)\right)$-$\apxMWC$ algorithm that succeeds with high probability, we try geometrically spaced candidate values for the scale parameter $d$, and for each value of $d$ we repeat the decomposition enough times to amplify the success probability. In total, this uses $\tilde{O}(n^{1/k})$ independent calls to $\LDD(G,V,k,d,\epsilon)$, each followed by the cycle detection step, and returns with high probability a cycle of weight at most
\[
    (k+1)\left(1+O\left(\frac{1}{\log n}\right)\right)\OPT.
\]

It suffices to consider {$1/2\leq k\in O(\log n)$}. Indeed, if $k>C\log n$ for a sufficiently large constant $C$, then we may instead run the algorithm with parameter $k'=C\log n$. This only improves the approximation ratio, and all the upper bounds considered in this paper (\Cref{thm:mainUB,thm:parallelUB,thm:cliqueUB}) change by at most a constant factor.

We set $\epsilon=1/(k\log^2 n)$. Since $k\in O(\log n)$, this is inverse-polylogarithmic.

\paragraph{Parameter guessing.}
Since edge weights are positive integers and $w_{\max}\in n^{O(1)}$, we have $1\leq \OPT\leq n w_{\max} \in n^{O(1)}$. Let $\eta=1/\log n$, and try all candidate values $2d=(1+\eta)^i$ between $1$ and {$2n w_{\max}$}. The number of candidates is $O(\log(nw_{\max})/\eta)=O(\log^2 n)$.

It remains to observe that one of these candidates satisfies the two conditions of \Cref{lem:ldd-cycle-detection-short}. Since $\delta_{\max}=(k+1+k/\ln n)d$ and $\epsilon=1/(k\log^2 n)$, we have $2\epsilon\delta_{\max}/d\in O(1/\log^2 n)$. Thus the condition
\[
    \OPT\leq 2(d-2\epsilon\delta_{\max})
\]
is satisfied whenever {$2d\geq (1+\eta)\OPT$, for sufficiently large $n$}. On the other hand, the condition
\[
    2d\left(1-O\left(\frac{1}{\log n}\right)\right)\leq \OPT
\]
is satisfied whenever $2d\leq (1+O(1/\log n))\OPT$.

Therefore, for sufficiently large $n$, every value $2d$ in the interval $[(1+\eta)\OPT,(1+\eta)^2\OPT]$ satisfies both requirements of \Cref{lem:ldd-cycle-detection-short}. This interval has multiplicative width $1+\eta$ and lies within $[1,2nw_{\max}]$, so it contains a candidate value $2d$.

\paragraph{Success probability amplification.}
For this candidate $d$, \Cref{lem:ldd-cycle-detection-short} shows that one run of $\LDD(G,V,k,d,\epsilon)$ followed by the cycle detection step succeeds with probability $\Omega(n^{-1/k})$. Thus $\Theta(n^{1/k}\log n)$ independent repetitions succeed with high probability. Since there are only $O(\log^2 n)$ candidate values of $d$, the total number of repetitions over all candidates is $\tilde{O}(n^{1/k})$.

Returning the cycle with minimum weight estimate over all repetitions and all candidate values of $d$ gives, with high probability, a cycle of weight at most
\[
    (k+1)\left(1+O\left(\frac{1}{\log n}\right)\right)\OPT,
\]
because every weight estimate produced by the cycle detection step is an upper bound on the true cycle weight.

\paragraph{Adjusting the approximation ratio.}
We next explain how to remove the extra $1+O(1/\log n)$
factor from the approximation guarantee without changing
the complexity bounds asymptotically.
For a target approximation ratio $k+1$ with $k\geq1$,
run the algorithm with parameter
\[
    k'=\frac{k+1}{1+C/\log n}-1,
\]
where $C$ is a sufficiently large constant bounding the
approximation error. For sufficiently large $n$, we have
$k'\geq1/2$, so the preceding analysis applies to this choice.
The resulting approximation ratio is at most
$(k'+1)(1+C/\log n)=k+1$.
Moreover, $1/k'-1/k=O(1/\log n)$, so the factors $n^{1/k}$
and $\alpha^{1/k}$ increase by at most constant factors
when $k$ is replaced by $k'$, for $1\leq\alpha\leq n$.

\subsection{Implementation of the Full Procedure}\label{subsect:impl}

We now go through the procedure described above and explain how its steps are implemented. The main point is that, except for global aggregation and broadcasting, all required operations can be reduced to $(1+\epsilon)$-approximate SSSP computations. Throughout this section, we use the transformation of \citet{RozhonHMGZ23} to assume that all approximate SSSP computations are tree-like, at the cost of only polylogarithmic overhead.

\paragraph{Implementing $\LDD(G,V,k,d,\epsilon)$.}
Recall that $\LDD(G,V,k,d,\epsilon)$ is implemented by adding a virtual super source $x$ and connecting it to each node $s\in V$ by an edge of weight $2\delta_{\max}-\hat{\delta}_s$. Running tree-like $(1+\epsilon)$-approximate SSSP from $x$ in the augmented graph $G^+$ gives a rooted tree $T$. Each node $v$ joins the cluster of the first real node on the $x$-$v$ path in $T$.

The SSSP computation immediately gives each node $v$ its estimate $\tilde{\dist}(x,v)$ and its parent in $T$. However, for the cycle detection step, the useful quantity is not $\tilde{\dist}(x,v)$ but the estimate from the cluster center. Namely, if $v$ is assigned to the cluster centered at $s$, then we need
\[
    \tilde{\dist}(s,v)=\tilde{\dist}(x,v)-\tilde{\dist}(x,s).
\]
Thus, every node $v$ has to learn two pieces of information about its cluster center $s$: the identity $\ID(s)$ and the value $\tilde{\dist}(x,s)$. Since both can be represented using $O(\log n)$ bits, this task is simply to let each center $s$ broadcast $O(\log n)$ bits of information to all nodes in its subtree $T_s$.

We explain how to broadcast one bit from each center to its cluster using one additional $(1+\epsilon)$-approximate SSSP computation. After the tree $T$ is computed, keep only the edges of $T$ and give all other edges infinite weight. {The children of $x$ in $T$ are the active cluster centers. Only the virtual edges joining $x$ to these centers receive finite weights.} To broadcast one bit $b_s\in\{0,1\}$ from each center $s$, set the weight of the virtual edge $\{x,s\}$ to one of two sufficiently separated values depending on $b_s$. Then run $(1+\epsilon)$-approximate SSSP from $x$ in this graph. Since the only finite paths from $x$ to nodes in $T_s$ enter $T_s$ through the virtual edge $\{x,s\}$, the resulting distance estimates allow every node in $T_s$ to recover the bit $b_s$.

{Repeating this procedure $O(\log n)$ times} lets every node $v$ learn $\ID(s)$ and $\tilde{\dist}(x,s)$ for its cluster center $s$. Consequently, every node can compute $\tilde{\dist}(s,v)=\tilde{\dist}(x,v)-\tilde{\dist}(x,s)$. At the end of this step, every node knows its cluster center, its parent in the local tree $T_s$, and its approximate distance estimate from the cluster center.

\paragraph{Cycle detection.}
The cycle detection step is \emph{local} once the above information is available. Each node $v$ sends the same $O(\log n)$-bit message to all its neighbors, containing its cluster center, its parent in $T$, and its value $\tilde{\dist}(s,v)$, where $s$ is the center of the cluster containing $v$. Hence this step is directly implementable in one round not only in the $\CONGEST$ model, but also in the broadcast congested clique model.

After this message exchange, each edge $e=\{u,v\}$ can be inspected locally. If $u$ and $v$ belong to different clusters, then $e$ is ignored. If they belong to the same cluster centered at $s$ and $e$ is not a tree edge of $T_s$, then $e$ defines a candidate cycle: the cycle consisting of $e$ together with the unique $u$-$v$ path in $T_s$. Its weight estimate is
\[
    \tilde{\dist}(s,u)+\tilde{\dist}(s,v)+w(u,v).
\]
Thus each edge can locally decide whether it gives a candidate cycle and, if so, compute its weight estimate.

\paragraph{Global aggregation.}
After the cycle detection step, each candidate cycle is represented by its non-tree edge $e=\{u,v\}$, together with its weight estimate. We need to find a candidate cycle of minimum weight estimate over all edges, all repetitions, and all candidate values of $d$.

This aggregation is straightforward in the models considered here. In the work-depth model, there are at most $\tilde{O}(m n^{1/k})$ candidate cycles in total, so the minimum can be found using $\tilde{O}(m n^{1/k})$ work and $\tilde{O}(1)$ depth. In the broadcast congested clique model, each node first locally computes the best candidate cycle among the edges incident to it, over all repetitions and all candidate values of $d$. Then all nodes broadcast their local minima, and the global minimum is found in one additional round. In the $\CONGEST$ model, the same aggregation can be performed along a BFS tree in $O(D)$ rounds once the candidates are known.

\paragraph{Recovering the actual cycle and its weight.}
The preceding steps identify an edge $e=\{u,v\}$ whose associated candidate cycle has minimum weight estimate. To output the actual cycle, it remains to identify the tree path between $u$ and $v$ in the local tree $T_s$. This can again be reduced to $(1+\epsilon)$-approximate SSSP computations.

Run SSSP in the tree $T_s$ from $u$ and from $v$. For each tree edge, compare the directions of the parent pointers induced by these two SSSP trees. An edge lies on the unique $u$-$v$ path in $T_s$ if and only if these two directions are opposite. Therefore, the nodes can mark exactly the tree edges on the $u$-$v$ path, and together with the non-tree edge $\{u,v\}$ this gives the desired cycle.

Once the edges of the output cycle are marked, computing its actual weight is another aggregation task, now using summation rather than minimization. Broadcasting the resulting weight to all nodes is also straightforward: it takes one round in the broadcast congested clique model and $O(D)$ rounds in the $\CONGEST$ model. In the work-depth model, the computation is centralized, so no broadcasting is needed.
 
\subsection{Consequences in Parallel and Distributed Models}\label{subsect:upperbounds}

The remaining task is to instantiate the procedure above in three models. The work-depth and broadcast congested clique implementations give \Cref{thm:parallelUB,thm:cliqueUB}. The direct $\CONGEST$ implementation gives a slightly weaker bound than \Cref{thm:mainUB}, which will later be improved {by} refining the approach with additional ideas.

Recall from \Cref{subsect:reduction} that we use $\epsilon=1/(k\log^2 n)$. Since it suffices to consider {$1/2\leq k\in O(\log n)$}, we have $\epsilon\in\log^{-O(1)} n$. By the transformation of \citet{RozhonHMGZ23}, we may assume that all approximate SSSP computations are tree-like, at the cost of only polylogarithmic overhead. By the previous discussion, $\tilde{O}(n^{1/k})$ independent executions of the decomposition and the cycle detection step suffice to obtain, with high probability, a cycle of weight at most $(k+1)(1+O(1/\log n))\OPT$.

\parallelupper*

\begin{proof}
For $\epsilon\in\log^{-O(1)} n$, $(1+\epsilon)$-approximate SSSP in undirected weighted graphs can be computed with $\tilde{O}(m)$ work and $\tilde{O}(1)$ depth~\cite{andoni2020parallel,li2020faster,rozhovn2022undirected}. Hence one execution of the decomposition, together with the auxiliary SSSP computations described above, costs $\tilde{O}(m)$ work and $\tilde{O}(1)$ depth.

We perform the $\tilde{O}(n^{1/k})$ executions in parallel. The total work is therefore $\tilde{O}(m n^{1/k})$, while the depth remains $\tilde{O}(1)$. Recovering the candidate cycle with the smallest estimated weight and computing its actual weight can also be done within the same asymptotic bounds. Therefore, with high probability, the algorithm solves the $(k+1)(1+O(1/\log n))$-$\apxMWC$ problem with $\tilde{O}(m n^{1/k})$ work and $\tilde{O}(1)$ depth.
 
Applying the parameter adjustment of \Cref{subsect:reduction} gives a $(k+1)$-approximation with the same asymptotic work and depth.
\end{proof}

\cliqueupper*

\begin{proof}
For $\epsilon\in\log^{-O(1)} n$, $(1+\epsilon)$-approximate SSSP in undirected weighted graphs can be computed in $\tilde{O}(1)$ rounds in this model~\cite{becker2021near}. Hence one execution of the decomposition and the associated cycle detection costs $\tilde{O}(1)$ rounds.

Running the required $\tilde{O}(n^{1/k})$ executions sequentially takes $\tilde{O}(n^{1/k})$ rounds. Recovering a candidate cycle with the smallest estimated weight, computing its actual weight, and broadcasting the result can all be performed within the same asymptotic round complexity. Therefore, with high probability, the algorithm solves the $(k+1)(1+O(1/\log n))$-$\apxMWC$ problem in $\tilde{O}(n^{1/k})$ rounds.

Applying the parameter adjustment of \Cref{subsect:reduction} gives a $(k+1)$-approximation with the same asymptotic round complexity.
\end{proof}

\paragraph{The direct $\CONGEST$ implementation.}
The same approach gives a direct $\CONGEST$ algorithm, but with a weaker round complexity than \Cref{thm:mainUB}. For $\epsilon \in\log^{-O(1)} n$, the $(1+\epsilon)$-approximate SSSP algorithm of \citet{becker2021near} can be implemented with congestion $\tilde{O}(\alpha)$ and using $\tilde{O}(n/\alpha+\alpha+D)$ rounds, for any parameter $\alpha$.

We need $\tilde{O}(n^{1/k})$ executions. The resulting collection of algorithms {has} congestion $\tilde{O}(\alpha n^{1/k})$ and dilation $\tilde{O}(n/\alpha+\alpha+D)$. By \Cref{thm: dilation and congestion}, these computations can be scheduled in
\[
    \tilde{O}\left(\alpha n^{1/k}+\frac{n}{\alpha}+D\right)
\]
rounds. The remaining aggregation and broadcasting steps take only ${O}(D)$ additional rounds.

Balancing $\alpha n^{1/k}$ and $n/\alpha$ gives
$\alpha=n^{(k-1)/(2k)}$, and hence the direct implementation runs in
\[
    \tilde{O}\left(n^{\frac{k+1}{2k}}+D\right)
\]
rounds.
 
Finally, we remove the extra $1+O(1/\log n)$ approximation
factor. If $k\geq1+2C/\log n$, where $C$ is the constant
in \Cref{subsect:reduction}, the parameter adjustment gives
$k'\geq1$. Thus, the above implementation applies with
$\alpha=n^{(k'-1)/(2k')}\geq1$, and the round complexity
is affected by at most a constant factor. Otherwise,
$k < 1+2C/\log n$, so $n^{(k+1)/(2k)} \in \Omega(n)$.
Therefore, we may use the existing $\tilde{O}(n)$-round algorithm for
the \emph{exact} $\MWC$ problem~\cite{bernstein2019distributed}.
 
In all cases, we obtain a $(k+1)$-$\apxMWC$ algorithm using
$\tilde{O}(n^{(k+1)/(2k)}+D)$ rounds in the $\CONGEST$ model.

\section{Sharpening the Round Complexity}

In this section, we show how to sharpen the round complexity bound $\tilde{O}(n^{(k+1)/(2k)}+D)$ above to achieve our target round complexity $\tilde{O}(n^{(k+1)/(2k+1)}+D)$ for any real number $k \geq (1+\sqrt{5})/2 \approx 1.618$. The key idea is to handle \emph{short-hop} and \emph{long-hop} minimum weight cycles separately. In \Cref{subsect:hop-bounded-sssp}, we describe a hop-bounded approximate SSSP primitive that will be used in both cases. In \Cref{subsect:short-hop-cycles}, we handle the case where some minimum weight cycle has few edges, by slightly perturbing the edge weights and then using the hop-bounded SSSP primitive to implement the MPX low-diameter decomposition. In \Cref{subsect:sampling-skeleton-nodes}, we describe the sampling of skeleton nodes used in the long-hop case. In \Cref{subsect:extended-cycles}, we introduce the extended cycle detection step, which patches the possible antipodal segment. In \Cref{subsect:long-hop-cycles}, we combine this patching step with the sampling of skeleton nodes to handle long-hop minimum weight cycles.

\subsection{Hop-Bounded Approximate SSSP}\label{subsect:hop-bounded-sssp}

We use the following hop-bounded version of approximate SSSP in both short-hop and long-hop parts of the algorithm. In the short-hop case, after a small perturbation of the edge weights, the promise that some minimum weight cycle has at most $h$ edges implies that all shortest paths have few hops. In the long-hop case, the same primitive will be used to obtain approximate distance information between the endpoints of the missing segment, namely the antipodal segment whose internal nodes are not guaranteed to be $\epsilon$-safe by \Cref{lem:capture-cycle-safe}.

\begin{definition}[$h$-hop distance]
Let $G=(V,E,w)$ be an undirected graph with positive edge weights, and let $h\geq 1$. For two nodes $u,v\in V$, the $h$-hop distance between $u$ and $v$ is
\[
    \dist_G^{(h)}(u,v)
    =
    \min\{w(P): \text{$P$ is a $u$-$v$ path in $G$ with at most $h$ edges}\}.
\]
If no such path exists, we set $\dist_G^{(h)}(u,v)=\infty$.
\end{definition}

\begin{definition}[$h$-hop-bounded approximate SSSP]\label{def-hop-sssp}
Let $G=(V,E,w)$ be an undirected graph with positive edge weights, let $x\in V$, and let $\epsilon>0$. A function
\[
    \tilde{\dist}(x,\cdot):V\to\mathbb{R}_{\geq 0}\cup\{\infty\}
\]
is a $(1+\epsilon)$-approximate $h$-hop-bounded SSSP distance estimate from $x$ if, for every node $v\in V$,
\[
    \dist_G(x,v)
    \leq
    \tilde{\dist}(x,v)
    \leq
    (1+\epsilon)\dist_G^{(h)}(x,v).
\]
\end{definition}

The hop-bounded SSSP primitive underlying the lemma below is already known~\cite{nanongkai2014distributed}. We include a proof for completeness, and because we need two additional properties that are not stated in the prior work~\cite{nanongkai2014distributed}: tree-likeness of the distance estimates and a bound on the height of the corresponding approximate SSSP tree $T$. Both properties will be useful later.

\begin{lemma}[Tree-like hop-bounded approximate SSSP~\cite{nanongkai2014distributed}]\label{lem:tree-like-hop-bounded-sssp}
Let $G=(V,E,w)$ be an undirected graph with positive polynomially bounded integer edge weights, let $x\in V$ be a source node, and let $h\geq 1$ and $\epsilon\in(0,1)$. There is a deterministic distributed algorithm in the $\CONGEST$ model that computes a $(1+\epsilon)$-approximate $h$-hop-bounded SSSP distance estimate $\tilde{\dist}(x,\cdot)$ from $x$.

Moreover, the algorithm outputs parent pointers defining a tree $T$ rooted at $x$ with height $O((h/\epsilon)\log n)$ certifying the tree-likeness property: for every node $v\neq x$ with $\tilde{\dist}(x,v) < \infty$, if $p(v)$ is the parent of $v$ in $T$, then
\[
    \tilde{\dist}(x,p(v))
    \leq
    \tilde{\dist}(x,v)-w(p(v),v).
\] 

The algorithm has congestion ${O}(\log n)$ and takes  ${O}((h/\epsilon) \log n)$ rounds. 
\end{lemma}
\begin{proof}
We use rounding and scaling, together with guessing of the distance scale. Observe that every finite $h$-hop distance lies in $[1,h\cdot w_{\max}]$, where
\[
    w_{\max}=\max_{e\in E} w(e) \in n^{O(1)},
\]
so $h\cdot w_{\max}\in n^{O(1)}$. We try all $O(\log n)$ scales
\[
    d\in \{1,2,4,\ldots,2^{\lceil \log(hw_{\max})\rceil}\}.
\]

\paragraph{One distance scale.}
Fix a scale $d$. Set
\[
    \mu_d=\epsilon\cdot\frac{d}{h},
\]
and replace each edge weight $w(e)$ by the scaled integer weight
\[
    w_d(e)=\left\lceil \frac{w(e)}{\mu_d}\right\rceil .
\]
For every path $P$ with at most $h$ edges, we have
\begin{equation}\label{eq:hop-bounded-rounding}
    \mu_d w_d(P)
    \leq
    w(P)+h\mu_d
    =
    w(P)+\epsilon d .
\end{equation}

For this fixed scale $d$, consider the unweighted graph obtained by subdividing each edge $e$ into a path of $w_d(e)$ edges. We run BFS from $x$ in this subdivided graph up to depth
\[
    R=\left\lceil \frac{3h}{\epsilon}\right\rceil .
\]
Equivalently, this BFS can be simulated in the original graph by forwarding the BFS wave across an edge $e$ with delay $w_d(e)$. Let $D_d(v)$ be the BFS distance from $x$ to $v$ in the subdivided graph, if this distance is at most $R$. The corresponding distance estimate in the original graph $G$ is $\mu_d D_d(v)$.

\paragraph{Approximation guarantee.}
Let $v$ be any node. If $\dist_G^{(h)}(x,v)=\infty$, then the upper bound requirement in \Cref{def-hop-sssp} is vacuous, so assume $\dist_G^{(h)}(x,v)<\infty$. The case $v=x$ is trivial, so assume $v\neq x$. Choose a scale $d$ such that
\[
    d\leq \dist_G^{(h)}(x,v) \leq 2d.
\]
Let $P$ be an $h$-hop shortest path from $x$ to $v$, so
\[
    w(P)=\dist_G^{(h)}(x,v).
\]
Since $P$ has at most $h$ edges, \eqref{eq:hop-bounded-rounding} gives
\[
    \mu_d w_d(P)
    \leq
    w(P)+\epsilon d
    =
    \dist_G^{(h)}(x,v)+\epsilon d
    \leq
    (1+\epsilon)\dist_G^{(h)}(x,v),
\]
where the last inequality uses $d\leq \dist_G^{(h)}(x,v)$.

We check that this path is within the BFS depth for scale $d$. Since $\dist_G^{(h)}(x,v)\leq 2d$, \eqref{eq:hop-bounded-rounding} gives
\[
    w_d(P)
    \leq
    \frac{w(P)+\epsilon d}{\mu_d}
    =
    \frac{\dist_G^{(h)}(x,v)+\epsilon d}{\mu_d}
    \leq
    \frac{(2+\epsilon)d}{\epsilon d/h}
    =
    \frac{(2+\epsilon)h}{\epsilon}
    \leq
    R,
\]
where the last inequality uses $\epsilon\leq 1$ and the definition of $R$.
Therefore the BFS for scale $d$ reaches $v$, and since BFS computes the
shortest distance in the subdivided graph up to depth $R$, we have
\[
    D_d(v)\leq w_d(P).
\]
Multiplying by $\mu_d$ and using the bound above gives
\[
    \mu_dD_d(v)
    \leq
    \mu_dw_d(P)
    \leq
    (1+\epsilon)\dist_G^{(h)}(x,v).
\]

After running this BFS procedure for all scale guesses $d$, the algorithm defines $\tilde{\dist}(x,v)$ to be the minimum value of $\mu_dD_d(v)$ over all scales for which $D_d(v)$ is finite. Therefore,
\[
    \tilde{\dist}(x,v)
    \leq
    (1+\epsilon)\dist_G^{(h)}(x,v).
\]
On the other hand, for every scale $d$, the value $\mu_dD_d(v)$ is the rounded weight of some actual $x$-$v$ path, and rounded edge weights only increase original edge weights. Hence $\mu_dD_d(v)\geq \dist_G(x,v)$ for every finite label, and so
\[
    \tilde{\dist}(x,v)\geq \dist_G(x,v).
\]
Thus $\tilde{\dist}(x,\cdot)$ is a $(1+\epsilon)$-approximate $h$-hop-bounded SSSP distance estimate.

\paragraph{Tree-likeness.}
For each node $v\neq x$ with finite final estimate, let $d(v)$ be any scale $d$ such that $\tilde{\dist}(x,v) = \mu_dD_d(v)$. In the BFS for scale $d(v)$, suppose the BFS wave first reaches $v$ by crossing the original edge $\{u,v\}$ from $u$ to $v$. We set $p(v)=u$. Then
\[
    D_{d(v)}(u)
    \leq
    D_{d(v)}(v)-w_{d(v)}(u,v).
\]
Multiplying by $\mu_{d(v)}$ gives
\[
    \mu_{d(v)}D_{d(v)}(u)
    \leq
    \mu_{d(v)}D_{d(v)}(v)-\mu_{d(v)}w_{d(v)}(u,v)
    \leq
    \mu_{d(v)}D_{d(v)}(v)-w(u,v),
\]
where the last inequality uses $\mu_{d(v)}w_{d(v)}(u,v)\geq w(u,v)$. Since the final estimate of $u$ is the minimum over all scales,
\[
    \tilde{\dist}(x,u)\leq \mu_{d(v)}D_{d(v)}(u).
\]
Therefore,
\[
    \tilde{\dist}(x,p(v))
    \leq
    \tilde{\dist}(x,v)-w(p(v),v),
\]
which is the desired tree-like property.

\paragraph{Tree height.}
The parent pointers define a tree rooted at $x$, since the estimate strictly decreases along every parent edge. We bound its height as follows. Consider any parent chain. Fix one scale $d$. Among the nodes on the chain whose final estimate is attained at scale $d$, the scaled labels $D_d(\cdot)$ strictly decrease along the chain, because the final estimates strictly decrease and the scaling factor $\mu_d$ is fixed. All such labels are integers between $0$ and $R$. Hence each scale appears at most $R+1$ times on the chain.

There are $O(\log n)$ scales, and $R\in O(h/\epsilon)$. Thus every parent chain has at most $O(R\log n)
    =
    O((h/\epsilon)\log n)$
 nodes. Therefore, $T$ has height $O((h/\epsilon)\log n)$.

\paragraph{Complexity.}
For each scale $d$, the bounded-depth BFS takes $O(h/\epsilon)$ rounds and has congestion $O(1)$: each directed edge forwards the BFS wave at most once for that scale. Since there are $O(\log n)$ scales, performing all BFS procedures sequentially takes
\[
    O((h/\epsilon)\log n)
\]
rounds and has total congestion $O(\log n)$.
\end{proof}

In the proof above, we do not attempt to optimize the tree height or the round complexity. For example, a simple pipelining argument can improve the round complexity to $O(h/\epsilon+\log n)$, but this improvement is not important for our application.

\subsection{Short-Hop Cycles}\label{subsect:short-hop-cycles}

We now handle the case where there exists a minimum weight cycle with few edges. We use the same MPX low-diameter decomposition and cycle detection procedure as in \Cref{subsect:impl}, with two changes. First, for each guessed parameter $d$, we slightly perturb the edge weights by adding a small amount to every edge. Under the promise that some minimum weight cycle has at most $h$ edges, this changes $\OPT$ by only a negligible factor. Second, we implement the required approximate SSSP computations using the hop-bounded primitive of \Cref{lem:tree-like-hop-bounded-sssp}. The perturbation ensures that the shortest paths relevant to these SSSP computations have few hops. In the proof below, we focus on these differences and avoid repeating the parts of the algorithm that were already discussed in detail in \Cref{sec:upperbounds}.

\begin{lemma}[Short-hop cycles]\label{lem:short-hop-cycles}
Suppose there exists a minimum weight cycle with at most $h$ edges. For every real number $k\geq 1$, $(k+1)$-$\apxMWC$ can be solved with high probability in
\[
    \tilde{O}\left(n^{1/k}+h+D\right)
\]
rounds in the $\CONGEST$ model.
\end{lemma}

\begin{proof}
{We first analyze the algorithm for $1/2\leq k\in O(\log n)$}, and we use the same parameter
\[
    \epsilon=\frac{1}{k\log^2 n}.
\]
We use the same guessing and repetition framework as in \Cref{subsect:reduction}. Thus, over all candidate values of $d$, we perform $\tilde{O}(n^{1/k})$ independent executions of the MPX low-diameter decomposition and the cycle detection step. Compared with \Cref{subsect:impl}, the execution for a candidate $d$ is run on a perturbed graph, and each approximate SSSP computation is implemented using \Cref{lem:tree-like-hop-bounded-sssp}.

\paragraph{Perturbing the weights.}
Fix a candidate value $d$. Define
\[
    \tau_d=\epsilon\cdot\frac{2d}{h},
\]
and let $G^{(d)}$ be the graph obtained from $G$ by replacing each edge weight $w(e)$ by
\[
    w^{(d)}(e)=w(e)+\tau_d.
\]
Let $\OPT^{(d)}$ be the minimum cycle weight in $G^{(d)}$.

Let $C^\star$ be a minimum weight cycle in $G$ with at most $h$ edges. Then
\[
    w^{(d)}(C^\star)
    =
    w(C^\star)+\tau_d |E(C^\star)|
    \leq
    \OPT+h\tau_d
    =
    \OPT+2\epsilon d.
\]
Hence, for every candidate $d$ with $2d\in \Theta(\OPT)$,
\[
    \OPT
    \leq
    \OPT^{(d)}
    \leq
    (1+O(\epsilon))\OPT.
\]
Thus, for the relevant candidates, the perturbation changes the optimum by only a $1+O(\epsilon)$ factor.

Recall from \Cref{subsect:reduction} that the admissible interval for $2d$ has multiplicative width $1+\Omega(1/\log n)$. Since $O(\epsilon) \subseteq o(1/\log n)$, the perturbation changes the optimum by much less than this width. Therefore, the same argument still guarantees that some candidate $d$ satisfies the hypotheses of \Cref{lem:ldd-cycle-detection-short} with respect to the perturbed graph $G^{(d)}$. To handle edge weights that are real rather than
integral, we may apply the same rounding and scaling procedure
as in \Cref{subsect:decomposition}.

\paragraph{Hop-bounded SSSP implementation.}
We next explain why the approximate SSSP computations used in \Cref{subsect:impl} can be made hop-bounded after the perturbation.

Consider the augmented graph used to implement the MPX low-diameter decomposition on $G^{(d)}$: we add a virtual super source $x$ and connect it to each center $s \in S=V$ by a virtual edge whose weight is at most $2\delta_{\max}$. Therefore, for every $v \in V$, every shortest $x$-$v$ path in the augmented graph has total weight at most $2\delta_{\max}$. On the other hand, every virtual edge has weight at least $\delta_{\max}$. Therefore, the part of any such shortest path lying in $G^{(d)}$ has weight at most $\delta_{\max}$.

Every edge of $G^{(d)}$ has weight at least $\tau_d$. Thus the part of every shortest $x$-$v$ path lying in $G^{(d)}$ has at most
\[
    \frac{\delta_{\max}}{\tau_d}
    \in
    O\left(\frac{kh}{\epsilon}\right)
    \subseteq
    \tilde{O}(h)
\]
edges, where we use $\delta_{\max}=(k+1+k/\ln n)d$ and $k\in O(\log n)$. Therefore, the SSSP computation implementing the MPX low-diameter decomposition can be replaced by \Cref{lem:tree-like-hop-bounded-sssp}, with hop parameter $\tilde{O}(h)$ and accuracy $\epsilon$.

As discussed in \Cref{subsect:impl}, except for the final global aggregation and broadcasting steps, all remaining tasks reduce to approximate SSSP computations, which are performed either in the augmented graph above or inside the approximate SSSP tree $T$, or one of its subtrees. By \Cref{lem:tree-like-hop-bounded-sssp}, the tree $T$ has height $\tilde{O}(h)$. Hence all of these computations can also be implemented using \Cref{lem:tree-like-hop-bounded-sssp} with hop parameter $\tilde{O}(h)$ and accuracy $\epsilon$.

\paragraph{Approximation guarantee.}
For some candidate value $d$, \Cref{lem:ldd-cycle-detection-short} applied to $G^{(d)}$ shows that, with probability $\Omega(n^{-1/k})$, one execution finds a cycle $C$ whose perturbed weight is at most
\[
    (k+1)\left(1+O\left(\frac{1}{\log n}\right)\right)\OPT^{(d)}.
\]
Since $\epsilon=1/(k\log^2 n)$, we have
\[
    \OPT^{(d)}
    \leq
    (1+O(\epsilon))\OPT
    \leq
    \left(1+O\left(\frac{1}{\log n}\right)\right)\OPT.
\]
Moreover, the original weight of any cycle is at most its perturbed weight. Hence a successful execution returns a cycle whose original weight is at most
\[
    (k+1)\left(1+O\left(\frac{1}{\log n}\right)\right)\OPT.
\]
Repeating over all candidates for a total of $\tilde{O}(n^{1/k})$ independent executions boosts the success probability to $1 - 1/\poly(n)$. {Applying the parameter adjustment of \Cref{subsect:reduction} gives a $(k+1)$-approximation for every target $k\geq 1$, without affecting the round complexity asymptotically.}

\paragraph{Round complexity.}
There are $\tilde{O}(n^{1/k})$ executions. In each execution, every approximate SSSP computation described above is implemented using \Cref{lem:tree-like-hop-bounded-sssp} with hop parameter $\tilde{O}(h)$, and hence has congestion $\tilde{O}(1)$ and dilation $\tilde{O}(h)$. Therefore, by \Cref{thm: dilation and congestion}, these computations can be executed using
\(
    \tilde{O}\left(n^{1/k}+h\right)
\)
rounds.
As in \Cref{subsect:impl}, each node keeps the best candidate cycle among all executions and incident edges. Aggregating these local minima over a BFS tree and broadcasting the selected value take $O(D)$ additional rounds. The total round complexity is therefore
\(\tilde{O}\left(n^{1/k}+h+D\right)\).
\end{proof}

\subsection{Sampling Skeleton Nodes}\label{subsect:sampling-skeleton-nodes}

Let $\alpha$ be a parameter. We sample each node independently with probability $\alpha/n$, call the sampled nodes \emph{skeleton nodes}, and write $S$ for the set of skeleton nodes.

The purpose of the sampling is to ensure that, on any sufficiently long cycle, consecutive skeleton nodes are not too far apart. Recall that, for $R=S\cap V(C)$, the nodes of $R$ partition a cycle $C$ into $R$-segments, namely maximal subpaths of $C$ whose endpoints are in $R$ and whose internal nodes are not in $R$. The lemma below shows that if $C$ has $h\in \Omega((n/\alpha)\log^2 n)$ edges, then all its $R$-segments have only $O((n/\alpha)\log n)$ edges with high probability. This is the property we need later: the possible antipodal segment that has to be patched is an $R$-segment, and hence it can be handled by hop-bounded SSSP.

\begin{lemma}[Skeleton nodes]\label{lem:skeleton-sampling}
Let $C$ be a fixed cycle with $h$ edges, and suppose
\[
    h\in \Omega\left(\frac{n \log^2 n}{\alpha}\right).
\]
Let $S$ be obtained by sampling each node independently with probability $\alpha/n$, and let $R=S\cap V(C)$. With high probability, the following hold.
\begin{description}
    \item[{At least two skeleton nodes:}] {$|R|\geq2$.}
    \item[Short segments:] Every $R$-segment of $C$ has at most $O((n/\alpha)\log n)$ edges.
    \item[Sample size:] $|S|\in \Theta(\alpha)$.
\end{description}
\end{lemma}

\begin{proof}
Let $L=\lceil c(n/\alpha)\ln n\rceil$ for a sufficiently large constant $c$. Consider any fixed block $B$ of $L$ consecutive nodes on $C$. The probability that $B$ contains no sampled node is
\[
    (1-\alpha/n)^L
    \leq
    e^{-\alpha L/n}
     \leq
    e^{-c\ln n}
    =
    n^{-c}.
\]
There are at most $n$ possible starting nodes for such a block, so by a union bound, with high probability, every block of $L$ consecutive nodes on $C$ contains a sampled node.

Since $L \leq h/2$ for sufficiently large $n$, we can find two disjoint blocks of $L$ consecutive nodes on $C$, so $|R| \geq 2$ when this event occurs. Since no $R$-segment can contain more than $L$ edges, every $R$-segment has $O((n/\alpha)\log n)$ edges.

It remains to bound $|S|$. We have $\mathbb{E}[|S|]=\alpha$. In the only regime where the lemma is relevant, $\alpha\in\Omega(\log^2 n)$; otherwise no cycle can satisfy the assumption $h\in\Omega((n/\alpha)\log^2 n)$, since $h\leq n$. Thus a Chernoff bound gives $|S|\in\Theta(\alpha)$ with high probability.
\end{proof}

\subsection{The Extended Cycle Detection Step}\label{subsect:extended-cycles}

We next describe the cycle detection procedure used in the long-hop case, for one fixed execution of $\LDD(G,S,k,d,\epsilon)$ with $S$ equal to the sampled skeleton nodes. The procedure extends the basic cycle detection step from \Cref{subsect:cycle} by adding a \emph{patching} step. 

The reason for this extension is the following. Let $C^\star$ be a minimum weight cycle, let $R=S\cap V(C^\star)$, and suppose $R$ is captured by the cluster centered at $s$. By \Cref{lem:capture-cycle-safe}, all nodes of $C^\star$, except possibly the internal nodes of the antipodal segment, are $\epsilon$-safe for $s$. Thus, if $C^\star$ contains an edge that is not in the local tree $T_s$ and whose endpoints are both $\epsilon$-safe, then the cycle detection step of \Cref{subsect:cycle} already succeeds by \Cref{lem:safe-nontree-edge-gives-cycle}. The remaining case is that every edge of $C^\star$ outside the antipodal segment is an edge of $T_s$. In this case, we patch the antipodal segment by finding an approximate patching path between its two skeleton endpoints.

\paragraph{The doubled graph $H$ for patching.} We need a patching path that, together with the 
path between its endpoints in the local tree $T_s$, contains a simple cycle. Merely
requiring the patch path to use a non-tree edge is insufficient,
since it may traverse that edge and immediately return.
Our idea is to assign parity bits to edges, with all tree
edges having parity zero, and require the patching path to have odd
parity. We enforce this condition using a \emph{doubled graph} $H$, constructed as follows.

Fix one execution of $\LDD(G,S,k,d,\epsilon)$ with approximate
SSSP tree $T$, and consider the forest $F=T-\{x\}=\bigcup_s T_s$. Set $\chi(e)=0$ for every tree edge
$e\in E(F)$. Using fresh randomness, independently assign
a uniform bit $\chi(e)\in\{0,1\}$ to every other edge in $G$.

The graph $H$ has two copies $v^{(0)},v^{(1)}$ of every node
$v$ in $G$. For each edge $e=\{u,v\}$ in $G$ and each $i\in\{0,1\}$, the graph $H$
contains the edge
\[
    \{u^{(i)},v^{(i\oplus\chi(e))}\}
\]
of weight $w(e)$, where $\oplus$ denotes addition modulo two.
Thus, each original edge in $G$ has exactly two realizations: within
the copies if $\chi(e)=0$ and between the copies if
$\chi(e)=1$. 

\paragraph{Patching candidates.}
Let
\[
    L_{\mathrm{seg}}=c_{\mathrm{seg}}\cdot \frac{n \log n}{\alpha}
\]
for a sufficiently large constant $c_{\mathrm{seg}}$, so that $L_{\mathrm{seg}}$ is at least the segment hop bound from \Cref{lem:skeleton-sampling}. For every skeleton node $r\in S$, we run $L_{\mathrm{seg}}$-hop $(1+\epsilon)$-approximate SSSP {of \Cref{lem:tree-like-hop-bounded-sssp}} from $r^{(0)}$ in {$H$}. For every skeleton node $r'\in S \setminus \{{r}\}$ such that both $r$ and $r'$ belong to the same cluster, let
\[
    \tilde{\dist}_{r}^{\mathrm{patch}}(r,r')
\]
denote the resulting distance estimate from $r^{(0)}$ to $(r')^{(1)}$.

Suppose $s$ is the center of the cluster that {contains} both $r$ and $r'$.
If $\tilde{\dist}_{r}^{\mathrm{patch}}(r,r') < \infty$, then we have
a candidate patching path whose weight estimate of the resulting cycle is
\[
    \tilde{\dist}(s,r)
    +
    \tilde{\dist}(s,r')
    +
    \tilde{\dist}_{r}^{\mathrm{patch}}(r,r').
\]

\paragraph{Simplicity of the cycle.} We next argue that, indeed, the patching path results in a simple cycle whose weight is at most the above estimate.  For any walk in $G$, define its
parity as the sum modulo two of $\chi(e)$ over its edge
traversals, with each traversal contributing one term.

 Let $\widehat W$ be
the path from $r^{(0)}$ to {$(r')^{(1)}$} in the tree
produced by the hop-bounded approximate SSSP computation
from $r^{(0)}$ in $H$. Replacing each {node} $v^{(i)}$ of
$\widehat W$ by $v$ gives an $r$-$r'$ walk $W$ in the original
graph $G$, possibly with repeated {nodes}. Let $Q$ be the
unique $r$-$r'$ path in the local tree $T_s$ in $G$.

We have $w(W)\leq\tilde{\dist}_{r}^{\mathrm{patch}}(r,r')$ and $w(Q)\leq\tilde{\dist}(s,r)+\tilde{\dist}(s,r')$, so it remains to show that $W\cup Q$ contains a simple cycle. The walk $W$
has parity one because $\widehat W$ joins opposite copies,
whereas $Q$ has parity zero because all its edges belong
to $F$. Following $W$ and then $Q$ in reverse
therefore gives a closed walk of parity one. If $W\cup Q$
contained no {simple} cycle, it would be a tree, so this closed walk
would traverse every edge an even number of times and have
parity zero, a contradiction.

\paragraph{The extended cycle detection step.}
The extended cycle detection step runs the cycle detection step from \Cref{subsect:cycle}, together with computing all patching candidates defined above. As before, after all candidates are generated, the algorithm keeps the candidate of minimum weight estimate.

The following lemma extends \Cref{lem:safe-nontree-edge-gives-cycle} by showing that, conditioned on $\capture(R)$, the extended cycle detection step {produces a candidate of small weight with probability at least $1/2$ over the random parity bits}.

\begin{lemma}[{Patching after capture}]\label{lem:patching-under-capture}
Let $C^\star$ be a minimum weight cycle, let $S\subseteq V$, and let $R=S\cap V(C^\star)$. Suppose {$|R|\geq2$} and every $R$-segment of $C^\star$ has at most $L_{\mathrm{seg}}$ edges. Run $\LDD(G,S,k,d,\epsilon)$, followed by the extended cycle detection step.
Assume that
\[
    w(C^\star)\leq 2(d-2\epsilon\delta_{\max}),
\]
and that $R$ is captured by the cluster centered at $s$. {Then, with probability at least $1/2$ over the parity bits, the extended cycle detection step produces a candidate whose weight estimate is at most}
\[
    (1+O(\epsilon))\,2\delta_{\max}.
\]
\end{lemma}

\begin{proof}
By \Cref{lem:capture-cycle-safe}, all nodes of $C^\star$, except possibly the internal nodes of the antipodal segment, are $\epsilon$-safe for $s$.

\paragraph{The easy case: a safe non-tree edge.}
First suppose $C^\star$ contains an edge $e=\{u,v\}$ such that both endpoints $u$ and $v$ are $\epsilon$-safe for $s$, and $e$ is not an edge of the local tree $T_s$. Then the cycle detection step of \Cref{subsect:cycle} succeeds by \Cref{lem:safe-nontree-edge-gives-cycle}, and finds a cycle whose weight estimate is at most $2\delta_{\max}$. Thus we may assume from now on that no such edge exists.

\paragraph{Structure of the remaining case.}
Let $P$ be the antipodal segment, and let $a,b\in R$ be its two endpoints. Let $Q=C^\star\setminus P$ be the complementary $a$-$b$ path on $C^\star$. Every edge of $Q$ has both endpoints $\epsilon$-safe for $s$, so by our assumption every edge of $Q$ is an edge of $T_s$. In particular, $Q$ is the tree path between $a$ and $b$ in $T_s$.

The segment $P$ must contain an edge outside $F=T-\{x\}=\bigcup_s T_s$, since otherwise $P\cup Q=C^\star$ would lie in the forest $F$. Since $P$ is simple, its parity is therefore one with probability $1/2$ over the independent parity bits.

\paragraph{Feasibility of the patch.}
{Assume that $P$ has parity one. Following $P$ in $H$ from $a^{(0)}$ then ends at $b^{(1)}$, giving a path of weight $w(P)$ and at most $L_{\mathrm{seg}}$ hops.} Hence the $L_{\mathrm{seg}}$-hop $(1+\epsilon)$-approximate SSSP from $a^{(0)}$ in {$H$} yields
\[
    \tilde{\dist}_{a}^{\mathrm{patch}}(a,b)
    \leq
    (1+\epsilon)w(P).
\]
\paragraph{Weight estimate.}
It remains to bound the weight estimate of the {candidate obtained by} patching $a$ and $b$. Let $q\in R$ be {a node with $\dist(s,q)=\dist(s,R)$}. By \Cref{prop:capture-close},
\[
    \dist(s,q)<\delta_{\max}-d.
\]
The path $Q$ contains $q$, and hence
\[
    \dist(s,a)+\dist(s,b)
    \leq
    2\dist(s,q)+w(Q)
    <
    2(\delta_{\max}-d)+w(Q).
\]
Moreover, as in the proof of \Cref{lem:safe-nontree-edge-gives-cycle}, the local estimates for the $\epsilon$-safe nodes $a$ and $b$ satisfy
\[
    \tilde{\dist}(s,a)+\tilde{\dist}(s,b)
    \leq
    \dist(s,a)+\dist(s,b)+O(\epsilon\delta_{\max}).
\]
Therefore the cycle due to patching $a$ and $b$ has weight estimate at most
\[
\begin{aligned}
    &\tilde{\dist}(s,a)
    +
    \tilde{\dist}(s,b)
    +
    \tilde{\dist}_{a}^{\mathrm{patch}}(a,b)  \\
    &\qquad\leq
    2(\delta_{\max}-d)+w(Q)+(1+\epsilon)w(P)+O(\epsilon\delta_{\max}) \\
    &\qquad=
    2\delta_{\max}-2d+w(C^\star)+\epsilon w(P)+O(\epsilon\delta_{\max}).
\end{aligned}
\]
Using $w(C^\star)\leq 2(d-2\epsilon\delta_{\max})$ and $w(P)\leq w(C^\star)\in O(\delta_{\max})$, this is at most
\[
    (1+O(\epsilon))\,2\delta_{\max}.
\]
{Since $P$ has parity one with probability $1/2$, this proves the lemma.}
\end{proof}

Next, we use \Cref{lem:patching-under-capture} to prove a long-hop analogue of \Cref{lem:ldd-cycle-detection-short}, establishing the approximation ratio and success probability of the extended cycle detection step for a single execution of $\LDD(G,S,k,d,\epsilon)$.
 
\begin{lemma}[Approximation ratio and success probability for long-hop cycles]\label{lem:ldd-cycle-detection-long}
Let $C^\star$ be a minimum weight cycle with $h$ edges, and suppose
\[
    h\in \Omega\left(\frac{n \log^2 n}{\alpha}\right).
\]
Assume $\alpha\in n^{\Omega(1)}$ and $\epsilon\in O(1/\log n)$. Sample each node independently with probability $\alpha/n$, and let $S$ be the resulting set of skeleton nodes. Suppose
\[
    2d\left(1-O\left(\frac{1}{\log n}\right)\right)
    \leq
    \OPT
    \leq
    2(d-2\epsilon\delta_{\max}).
\]
 One execution of $\LDD(G,S,k,d,\epsilon)$, followed by the extended cycle detection step with fresh independent parity bits, produces a candidate whose weight estimate is at most 
\[
    (k+1)\left(1+O\left(\frac{1}{\log n}\right)\right)\OPT
\]
 with probability $\Omega(\alpha^{-1/k})$.
\end{lemma}

\begin{proof} 
Let $R=S\cap V(C^\star)$. By \Cref{lem:skeleton-sampling}, with high probability, {$|R|\geq2$}, every $R$-segment of $C^\star$ has at most $L_{\mathrm{seg}}\in O((n/\alpha)\log n)$ edges, and $|S|\in\Theta(\alpha)$. {Fix any sample $S$ satisfying these conditions.} \Cref{lem:capture-probability} yields
\[
    \Pr[\capture(R)]
    \geq
    (1-1/e)|S|^{-1/k}
    =
    \Omega(\alpha^{-1/k}).
\]
Condition on the event $\capture(R)$, and let $s$ be the center of the cluster that captures $R$. 
Since
\[
    w(C^\star)=\OPT\leq 2(d-2\epsilon\delta_{\max}),
\]
all assumptions of \Cref{lem:patching-under-capture} are satisfied. {Hence, with probability at least $1/2$ over the parity bits, the extended cycle detection step produces a candidate whose weight estimate is at most}
\[
    (1+O(\epsilon))\,2\delta_{\max}.
\]

 Since $|S|\in\Theta(\alpha)$ and $\alpha\in n^{\Omega(1)}$, we have $\ln |S|\in \Omega(\log n)$ and  $\delta_{\max} \in (k+1)(1+O(1/\log n))d$.  Together with
\[
    2d
    \leq
    \left(1+O\left(\frac{1}{\log n}\right)\right)\OPT,
\]
we have
\[
    (1+O(\epsilon))\,2\delta_{\max}
    \leq
    (k+1)\left(1+O\left(\frac{1}{\log n}\right)\right)\OPT.
\]

\paragraph{Summary.} The guarantee of \Cref{lem:skeleton-sampling} holds with high probability. For any fixed sample $S$ meeting the 
guarantee of \Cref{lem:skeleton-sampling}, $\capture(R)$ occurs with probability
$\Omega(\alpha^{-1/k})$. Conditioned on $\capture(R)$, the extended
cycle detection step produces a desired candidate with
probability at least $1/2$. Hence the overall success
probability is $\Omega(\alpha^{-1/k})$.
\end{proof}

\subsection{Long-Hop Cycles}\label{subsect:long-hop-cycles}

We now combine the ingredients above to handle the long-hop case.  

\begin{lemma}[Long-hop cycles]\label{lem:long-hop-cycles}
Let $\alpha\in n^{\Omega(1)}$. Suppose there exists a minimum weight cycle with
\[
    h\in \Omega\left(\frac{n \log^2 n}{\alpha}\right)
\]
edges. For every real number $k\geq 1$, $(k+1)$-$\apxMWC$ can be solved with high probability in
\[
    \tilde{O}\left(\alpha^{1+1/k}+\frac{n}{\alpha}+D\right)
\]
rounds in the $\CONGEST$ model.
\end{lemma}

\begin{proof}
{We first analyze the algorithm for $1/2\leq k\in O(\log n)$.} Set
\[
    \epsilon=\frac{1}{k\log^2 n}.
\]
We use the same guessing and amplification framework as in \Cref{subsect:reduction}. We first sample the skeleton set $S$ by including each node independently with probability $\alpha/n$. Then, over all candidate values of $d$, we perform $\tilde{O}(\alpha^{1/k})$ independent executions of $\LDD(G,S,k,d,\epsilon)$, followed by the extended cycle detection step {with fresh independent parity bits}, as described in \Cref{subsect:sampling-skeleton-nodes,subsect:extended-cycles}.

\paragraph{Success probability.}
Let $C^\star$ be a minimum weight cycle with $h\in\Omega((n/\alpha)\log^2 n)$ edges. By \Cref{lem:skeleton-sampling}, with high probability, the skeleton set $S$ satisfies $|S|\in\Theta(\alpha)$, $|R|\geq2$ for $R=S\cap V(C^\star)$, and every $R$-segment of $C^\star$ has at most $O((n/\alpha)\log n)$ edges. Fix any such $S$ for the subsequent analysis.

By the same argument as in \Cref{subsect:reduction}, some candidate value $d$ satisfies
\[
    2d\left(1-O\left(\frac{1}{\log n}\right)\right)
    \leq
    \OPT
    \leq
    2(d-2\epsilon\delta_{\max}).
\]
For this candidate $d$, \Cref{lem:ldd-cycle-detection-long} shows that one execution succeeds with probability $\Omega(\alpha^{-1/k})$.   Repeating $\tilde{O}(\alpha^{1/k})$ times therefore boosts the success probability to $1-1/\poly(n)$.

By \Cref{lem:ldd-cycle-detection-long}, a successful execution {produces a candidate} whose weight estimate is at most
\[
    (k+1)\left(1+O\left(\frac{1}{\log n}\right)\right)\OPT.
\]
{Applying the parameter adjustment of \Cref{subsect:reduction} gives a $(k+1)$-approximation for every target $k\geq 1$, without affecting the round complexity asymptotically.}

\paragraph{Round complexity.}
Compared with the algorithm described in \Cref{sec:upperbounds}, there are two main differences: the patching part in the extended cycle detection step is new, and the number of repetitions is reduced from $\tilde{O}(n^{1/k})$ to $\tilde{O}(\alpha^{1/k})$.

We first consider all tasks other than the patching part. As discussed in \Cref{subsect:impl}, apart from some global aggregation and broadcasting tasks that can be done in $O(D)$ rounds, all tasks reduce to $(1+\epsilon)$-approximate SSSP computations.

As discussed in \Cref{subsect:upperbounds}, for $\epsilon=1/(k\log^2 n)\in\log^{-O(1)} n$, the $(1+\epsilon)$-approximate SSSP algorithm of~\cite{becker2021near} can be implemented with congestion $\tilde{O}(\alpha)$ in $\tilde{O}(n/\alpha+\alpha+D)$ rounds, for any parameter $\alpha$.

Over the $\tilde{O}(\alpha^{1/k})$ independent executions, the resulting collection has congestion $\tilde{O}(\alpha^{1+1/k})$ and dilation $\tilde{O}(n/\alpha+\alpha+D)$. By \Cref{thm: dilation and congestion}, these computations take
\[
    \tilde{O}\left(\alpha^{1+1/k}+\frac{n}{\alpha}+D\right)
\]
rounds in total, which is the stated round complexity bound.

It remains to account for the patching part in the extended cycle detection step, which is also repeated independently $\tilde{O}(\alpha^{1/k})$ times. In each execution, for each skeleton node $r\in S$, we run hop-bounded $(1+\epsilon)$-approximate SSSP from $r^{(0)}$ in the doubled graph {$H$} with hop bound
\[
    L_{\mathrm{seg}} \in \tilde{O}\left(\frac{n}{\alpha}\right).
\]
The doubled graph $H$ can be simulated in the original graph $G$ with only a constant-factor overhead in round complexity: each node simulates its two copies, so each physical edge corresponds to two edges in $H$. 

By \Cref{lem:tree-like-hop-bounded-sssp}, each such hop-bounded $(1+\epsilon)$-SSSP has congestion $\tilde{O}(1)$ and dilation $\tilde{O}(n/\alpha)$. Since there are $\Theta(\alpha)$ skeleton sources per execution and $\tilde{O}(\alpha^{1/k})$ executions, the full collection of patching computations has congestion $\tilde{O}(\alpha^{1+1/k})$ and dilation $\tilde{O}(n/\alpha)$. Another application of \Cref{thm: dilation and congestion} gives
\[
    \tilde{O}\left(\alpha^{1+1/k}+\frac{n}{\alpha}\right)
\]
rounds for the patching step, which is within the stated round complexity bound.
 
\paragraph{Recovering the selected cycle.} 
If the candidate with the smallest weight estimate comes
from the cycle detection step of \Cref{subsect:cycle},
we recover the corresponding cycle using the procedure
of \Cref{subsect:impl}. Otherwise, the candidate with the smallest weight estimate is due to patching: let $W$ be the walk in $G$
obtained by projecting the selected patching path in $H$.
Its two skeleton endpoints $r$ and $r'$ belong to the same
local tree $T_s$. Let $Q$ be the $r$-$r'$ path in $T_s$.
We recover a simple cycle in $W\cup Q$ as follows.

First, mark the edges of $W$ by tracing the stored path in $H$, taking $\tilde{O}(n/\alpha)$ rounds.
Mark the edges of $Q$ using two tree-like
$(1+\epsilon)$-approximate SSSP computations from $r$ and
$r'$ in $T_s$, as in \Cref{subsect:impl}.
Next, compute a spanning tree $T^\ast$ of $W\cup Q$ by another
such SSSP computation, rooted at $r$. Since $W\cup Q$ contains
a cycle, it has an edge $e=\{u,v\}$ outside $T^\ast$.
Select such an edge and use two SSSP computations in
$T^\ast$, as in \Cref{subsect:impl}, to mark the $u$-$v$
tree path. Together with $e$, this path forms a simple cycle.
\end{proof}

We are now ready to prove our $\CONGEST$ upper bound.

\kupper*

\begin{proof}
As before, it suffices to consider $1\leq k\in O(\log n)$. Set $\alpha=n^{k/(2k+1)}\in n^{\Omega(1)}$, so we have  
\[
    \frac{n}{\alpha}
    =
    n^{(k+1)/(2k+1)} =
    \alpha^{1+1/k}.
\]

Select $h \in \Theta\left(\frac{n}{\alpha}\log^2 n\right)$ to satisfy the precondition of \Cref{lem:long-hop-cycles}. The algorithm runs both the short-hop algorithm of \Cref{lem:short-hop-cycles} with hop parameter $h$ and the long-hop algorithm of \Cref{lem:long-hop-cycles} with parameter $\alpha$, and returns the best cycle found.

Let $C^\star$ be a minimum weight cycle. If $C^\star$ has at most $h$ edges, then the short-hop algorithm outputs a $(k+1)$-approximation with high probability in
\[
    \tilde{O}\left(n^{1/k}+h+D\right)
    =
    \tilde{O}\left(
        n^{1/k}
        +
        \frac{n}{\alpha}
        +
        D
    \right)
\]
rounds, by \Cref{lem:short-hop-cycles}.

Otherwise, $C^\star$ has more than $h$ edges, so the long-hop algorithm outputs a $(k+1)$-approximation with high probability in
\[
    \tilde{O}\left(
        \alpha^{1+1/k}
        +
        \frac{n}{\alpha}
        +
        D
    \right)
\]
rounds, by \Cref{lem:long-hop-cycles}.

Substituting $\alpha=n^{k/(2k+1)}$ into the two bounds gives
\[
    \tilde{O}\left(
        n^{(k+1)/(2k+1)}
        +
        n^{1/k}
        +
        D
    \right).
\]
Thus, in either case, the algorithm returns a $(k+1)$-approximate minimum weight cycle with high probability within the claimed round complexity.
\end{proof}

\section{Lower Bounds}
\label{sec:lower bounds}

In this section, we prove the lower bounds stated in
\Cref{thm:klower,thm:cliqueLB}. Our proofs are based on reductions from
set-disjointness using dense high-girth bipartite graphs given by the
Erd\H{o}s girth conjecture (\Cref{girth_conj}).
In \Cref{subsec:basic-lb-construction}, we introduce a basic lower bound graph
construction. In \Cref{subsec:bcc-lower bound}, we use this construction
directly to prove the lower bound in the broadcast congested clique model. We
then turn to the $\CONGEST$ model. In
\Cref{subsec:congest-lb-construction}, we modify the basic lower bound graph
construction. In \Cref{subsec:moving-cut}, we review the moving-cut framework
used in our analysis. In \Cref{subsec:disj-lb-fixed-graph}, we apply this
framework to prove a set-disjointness lower bound on the modified construction.
Finally, in \Cref{subsec:disj-to-mwc}, we reduce set-disjointness to
$(k+1-\epsilon)$-$\apxMWC$ and prove the $\CONGEST$ lower bound.

\subsection{A Basic Two-Copy Construction}
\label{subsec:basic-lb-construction}

We first describe a simple construction that turns set-disjointness into an
approximation gap for $\MWC$. This construction is used directly for the
broadcast congested clique lower bound. Later, in the $\CONGEST$ lower bound,
we replace the matching edges between the two copies by long paths and add an
overlay tree.
Let
\[
    H=(L\cup R,E_H)
\]
be a bipartite graph with two parts of equal size $\gamma$:
\[
    L=\{\ell_1,\ldots,\ell_\gamma\},
    \qquad
    R=\{\rho_1,\ldots,\rho_\gamma\}.
\]

\paragraph{Encoding bit strings using subgraphs of $H$.}
For two strings $x,y\in\{0,1\}^{E_H}$, define the graph
$\mathcal{G}(H,x,y)$ as follows. The graph contains two copies of the node
set $L\cup R$ of $H$, denoted by $A\cup B$ and $U\cup V$, respectively,
where
\[
    A=\{a_i:i\in[\gamma]\},\quad
    B=\{b_j:j\in[\gamma]\},\quad
    U=\{u_i:i\in[\gamma]\},\quad
    V=\{v_j:j\in[\gamma]\}.
\]
Here, $a_i$ and $u_i$ are the copies of $\ell_i$, while $b_j$ and $v_j$
are the copies of $\rho_j$.

Each string specifies a subgraph in one copy of $H$. For every edge
$e=\{\ell_i,\rho_j\}\in E_H$, we add the edge $\{a_i,b_j\}$ to the first
copy if $x_e=1$, and add the edge $\{u_i,v_j\}$ to the second copy if
$y_e=1$.

Finally, we add a perfect matching between corresponding nodes in the two
copies: for every $i\in[\gamma]$, we add the edge $\{a_i,u_i\}$, and for
every $j\in[\gamma]$, we add the edge $\{b_j,v_j\}$.

\paragraph{Directions or weights.} We use two variants of this construction. In the directed unweighted variant,
the edges are directed as follows:
\[
    a_i \to u_i,
    \qquad
    u_i \to v_j,
    \qquad
    v_j \to b_j,
    \qquad
    b_j \to a_i.
\]
Here $u_i\to v_j$ is present if and only if $y_{\{\ell_i,\rho_j\}}=1$, and
$b_j\to a_i$ is present if and only if $x_{\{\ell_i,\rho_j\}}=1$.
In the undirected weighted variant, the matching edges have weight $0$, and
all edges inside the two copies have weight $1$.

For $x,y\in\{0,1\}^{E_H}$, write
\[
    \inprod{x}{y}
    =
    \sum_{e\in E_H} x_e y_e.
\]
The following lemma establishes a gap in the value of $\OPT$ between the two
cases $\inprod{x}{y}\neq 0$ and $\inprod{x}{y}=0$.

\begin{lemma}[Approximation gap in the basic construction]
\label{lem:basic-cycle-gap}
Suppose $H$ has girth greater than $2k$. Then the following hold for $\mathcal{G}(H,x,y)$.\\

\begin{minipage}[t]{0.48\textwidth}

\textbf{Directed unweighted case:}
\begin{itemize}
    \item If $\inprod{x}{y}\neq 0$, then $\OPT \le 4$.
    \item If $\inprod{x}{y}=0$, then $\OPT \ge 4(k+1)$.
\end{itemize}
\end{minipage}
\hfill
\begin{minipage}[t]{0.48\textwidth}
\textbf{Undirected weighted case:}
\begin{itemize}
    \item If $\inprod{x}{y}\neq 0$, then $\OPT \le 2$.
    \item If $\inprod{x}{y}=0$, then $\OPT \ge 2(k+1)$.
\end{itemize}
\end{minipage}
\end{lemma}

\begin{proof}
We first consider the directed unweighted case. If $\inprod{x}{y}\neq 0$, then
there is an edge $e=\{\ell_i,\rho_j\}\in E_H$ with $x_e=y_e=1$. Hence both
edges
\[
    b_j\to a_i
    \qquad\text{and}\qquad
    u_i\to v_j
\]
are present. Together with $a_i\to u_i$ and $v_j\to b_j$, they form the
directed cycle
\[
\begin{tikzpicture}[baseline=-0.5ex]
\node (a) {$a_i$};
\node[right=1em of a] (u) {$u_i$};
\node[right=1em of u] (v) {$v_j$};
\node[right=1em of v] (b) {$b_j$};
\draw[->] (a) -- (u);
\draw[->] (u) -- (v);
\draw[->] (v) -- (b);
\draw[->,bend left=35] (b) to (a);
\end{tikzpicture}
\]
so $\OPT\le 4$.

Now suppose $\inprod{x}{y}=0$. Let $C$ be any directed cycle in
$\mathcal{G}(H,x,y)$. Contract the matching edges $a_i\to u_i$ and
$v_j\to b_j$. The remaining edges of $C$ project to a closed walk in $H$: both
$b_j\to a_i$ and $u_i\to v_j$ project to the edge $\{\ell_i,\rho_j\}$. {Since $\inprod{x}{y}=0$, at most one copy of each edge of $H$ is present, so distinct non-matching edges of $C$ project to distinct edges of $H$. The matching alone contains no cycle, so the projection is a non-empty closed trail and contains a cycle of length at least $2k+2$ by bipartiteness and the girth assumption.} The original cycle alternates between edges inherited from $H$ and matching edges, and
therefore has weight at least $2(2k+2)=4(k+1)$. Thus $\OPT\ge 4(k+1)$.

The proof for the undirected weighted case is analogous. If
$\inprod{x}{y}\neq 0$, then for some edge $e=\{\ell_i,\rho_j\}$, both
$\{a_i,b_j\}$ and $\{u_i,v_j\}$ are present. Together with the two zero-weight matching
edges $\{a_i,u_i\}$ and $\{b_j,v_j\}$, they form a cycle of weight $2$, so
$\OPT\le 2$.

Conversely, suppose $\inprod{x}{y}=0$. Contract the zero-weight matching
edges. Every cycle projects to {a non-empty closed trail in $H$ by the same argument}, and hence uses at least
$2k+2$ edges inherited from $H$. Since each edge  inherited from $H$ has weight $1$, every cycle has weight
at least $2k+2=2(k+1)$. Thus $\OPT\ge 2(k+1)$.
\end{proof}

\subsection{Lower Bound in the Broadcast Congested Clique Model}
\label{subsec:bcc-lower bound}

We first recall the set-disjointness problem. For any integer $b\ge 1$, Alice
and Bob receive strings $x,y\in\bin^b$, viewed as the characteristic vectors of
the sets
\[
    X=\{i\in[b]:x_i=1\}
    \qquad\text{and}\qquad
    Y=\{i\in[b]:y_i=1\}.
\]
Their goal is to determine whether $X$ and $Y$ are disjoint. Equivalently, since
$\inprod{x}{y}=\sum_{i=1}^b x_i y_i$ counts the elements in $X\cap Y$, define
\[
    \disj_b(x,y)
    =
    \begin{cases}
        1 & \text{if } \inprod{x}{y}=0,\\
        0 & \text{otherwise}.
    \end{cases}
\]
The randomized communication complexity of $\disj_b$ is $\Omega(b)$, even
with constant error~\cite{Razborov1992}.

We say that a randomized distributed algorithm is \emph{$\delta$-error} if it produces a correct
output with a probability of at least $1-\delta$ on every input instance.

\begin{proposition}[Broadcast congested clique lower bound from a high-girth graph]
\label{prop:bcc-lower bound}
Let $H=(L\cup R,E_H)$ be a bipartite graph with
$|L|=|R|=\gamma$ and $\operatorname{girth}(H)>2k$. Let $\epsilon > 0$ be any real number.
There exists a constant $\delta>0$ such that any $\delta$-error algorithm for
$(k+1-\epsilon)$-$\apxMWC$ in the broadcast congested clique model requires
\[
    \Omega\left(\frac{|E_H|}{\gamma\log\gamma}\right)
\]
rounds on the family of $4\gamma$-node graphs
\[
    \left\{
        \mathcal{G}(H,x,y)
        :
        x,y\in\{0,1\}^{E_H}
    \right\}.
\]
The lower bound holds for directed unweighted graphs and undirected graphs
with non-negative integer weights.
\end{proposition}

\begin{proof}
Let $\delta>0$ be a constant for which the randomized public-coin
communication complexity of set-disjointness is $\Omega(|E_H|)$. Suppose
there is a $\delta$-error algorithm $\mathcal{A}$ for
$(k+1-\epsilon)$-$\apxMWC$ on the graph family in the proposition, running
in $T$ rounds. We use $\mathcal{A}$ to construct a two-party communication
protocol for $\disj_{|E_H|}$.

Given $x,y\in\bin^{E_H}$, Alice simulates the nodes in $A\cup B$ of
$\mathcal{G}(H,x,y)$, while Bob simulates the nodes in $U\cup V$. Alice
knows all edges incident to $A\cup B$. Similarly, Bob knows
all edges incident to $U\cup V$. Thus, each party can initialize the states of
all nodes it simulates.

\paragraph{Simulation cost.}
The two parties simulate $\mathcal{A}$ round by round. In each round, Alice locally
computes the broadcast messages of the $2\gamma$ nodes in $A\cup B$ and sends
them to Bob, while Bob does the same for the $2\gamma$ nodes in $U\cup V$.
Since each broadcast message contains $O(\log\gamma)$ bits, one round of
$\mathcal{A}$ can be simulated using $O(\gamma\log\gamma)$
bits of two-party communication.  Hence the entire simulation uses $O(T\gamma\log\gamma)$ bits of communication.

\paragraph{Correctness.}
Suppose first that $\inprod{x}{y}\neq 0$. By
\Cref{lem:basic-cycle-gap}, $\OPT\le 4$ in the directed unweighted case and
$\OPT\le 2$ in the undirected weighted case. Therefore, with probability at
least $1-\delta$, $\mathcal{A}$ returns a cycle whose weight is at
most
\(
    4(k+1-\epsilon)<4(k+1)
\)
in the directed case, and at most
\(
    2(k+1-\epsilon)<2(k+1)
\)
in the weighted case.

On the other hand, if $\inprod{x}{y}=0$, then
\Cref{lem:basic-cycle-gap} guarantees that every cycle has weight at
least $4(k+1)$ in the directed case and at least $2(k+1)$ in the weighted
case. Thus, the output of $\mathcal{A}$ distinguishes the two cases with error
at most $\delta$. Consequently, the simulation yields a $\delta$-error
two-party protocol for $\disj_{|E_H|}$.

By the $\Omega(|E_H|)$ randomized communication lower bound for
set-disjointness,
\(
    T\gamma\log\gamma
    \in
    \Omega(|E_H|)
\). 
Therefore,
\(
    T
    \in
    \Omega\left(
        \frac{|E_H|}{\gamma\log\gamma}
    \right)\),
which proves the proposition.
\end{proof}

We are ready to prove \Cref{thm:cliqueLB}.

\begin{proof}[Proof of \Cref{thm:cliqueLB}]
Assume the Erd\H{o}s girth conjecture. For every integer $k\ge 1$, there
exists a bipartite graph $H$ with $|L|=|R|=\gamma$,
$\operatorname{girth}(H)>2k$, and
$|E_H|\in\Omega(\gamma^{1+1/k})$. Since
$\mathcal{G}(H,x,y)$ has $n=4\gamma$ nodes,
\Cref{prop:bcc-lower bound} gives the desired lower bound
\(
    \Omega\left(\frac{\gamma^{1/k}}{\log\gamma}\right)
    \subseteq
    \tilde{\Omega}\left(n^{1/k}\right)
\) for directed unweighted graphs and undirected graphs with \emph{non-negative}
integer weights.

\paragraph{Non-negative weights $\rightarrow$ positive weights.}
There is one remaining gap between \Cref{prop:bcc-lower bound} and
\Cref{thm:cliqueLB}: the proposition allows non-negative integer weights,
whereas the theorem requires positive integer weights. We now show that this
gap can be removed while preserving the approximation gap.

Let $w$ be the original non-negative weight function, set
\[
    M
    =
    \left\lceil
        \frac{2(k+1)n}{\epsilon}
    \right\rceil,
\]
and replace every edge weight by
\[
    w'(e)=Mw(e)+1.
\]
For every cycle $C$, since $C$ contains at most $n$ edges,
\[
    Mw(C)
    \le
    w'(C)
    \le
    Mw(C)+n.
\]
Hence, if $\inprod{x}{y}\neq 0$, then the original bound $\OPT\le 2$ gives
$\OPT'\le 2M+n$. If $\inprod{x}{y}=0$, then the original bound
$\OPT\ge 2(k+1)$ gives $\OPT'\ge 2(k+1)M$.

The gap is preserved because
\[
    (k+1-\epsilon)(2M+n)
    <
    2(k+1)M.
\]
Indeed, this inequality is equivalent to
$(k+1-\epsilon)n<2\epsilon M$, which follows immediately from the choice of
$M$. Thus, a $(k+1-\epsilon)$-approximation still distinguishes between $\inprod{x}{y}\neq 0$ and $\inprod{x}{y} = 0$.

Finally, if {$1/\epsilon$ is polynomially bounded in $n$}, then $M\in n^{O(1)}$, so all
weights are positive polynomially bounded integers, as required.
\end{proof}

\subsection{The Lower Bound Graphs for \texorpdfstring{$\CONGEST$}{CONGEST}}
\label{subsec:congest-lb-construction}

We now adapt the basic construction to the $\CONGEST$ model. The main idea is
to replace the matching edges between the two copies of $H=(L\cup R,E_H)$ with long paths,
while adding an overlay tree that keeps the diameter small. Refer to \Cref{fig:GgkdpHxy} for an illustration of the construction in the directed unweighted setting.

For integers $d\ge 2$ and $p\ge 1$, define $\mathcal{G}_{d,p}(H,x,y)$ as follows.

\paragraph{Replacing the matching with long paths.} Start with the graph
$\mathcal{G}(H,x,y)$. For every
$i\in[\gamma]$, replace the matching edge between $a_i$ and $u_i$ with a $d^p$-node path
\[
    Q_i=(q^i_0,q^i_1,\ldots,q^i_{d^p-1}),
    \qquad
    q^i_0=a_i,
    \qquad
    q^i_{d^p-1}=u_i.
\]
Similarly, for every $j\in[\gamma]$, replace the matching edge between $b_j$
and $v_j$ with a $d^p$-node path
\[
    R_j=(r^j_0,r^j_1,\ldots,r^j_{d^p-1}),
    \qquad
    r^j_0=b_j,
    \qquad
    r^j_{d^p-1}=v_j.
\]
We write
\[
    \mathcal{Q}=\{Q_i:i\in[\gamma]\}
    \qquad\text{and}\qquad
    \mathcal{R}=\{R_j:j\in[\gamma]\}.
\]

\paragraph{Adding an overlay tree.}  Next, add a complete $d$-ary tree $\mathcal{T}$ of depth $p$. Its nodes at
depth $h$ are denoted by
\[
    t^h_0,t^h_1,\ldots,t^h_{d^h-1}.
\]
For every position $s\in\{0,\ldots,d^p-1\}$, connect the leaf $t^p_s$ to
$q^i_s$ for every $i\in[\gamma]$, and to $r^j_s$ for every
$j\in[\gamma]$. Thus, the $s$th leaf of $\mathcal{T}$ provides an overlay
connection among all path nodes at position $s$.

\paragraph{Directed unweighted variant.}
Orient every path $Q_i$ from $a_i$ to $u_i$, and every path $R_j$ from $v_j$
to $b_j$. For each edge $e=\{\ell_i,\rho_j\}\in E_H$, orient the
edges inherited from $H$ as
\[
    b_j\to a_i
    \qquad\text{and}\qquad
    u_i\to v_j.
\]
The first edge is present if $x_e=1$, and the second if $y_e=1$.
Finally, orient every tree edge away from the root and every edge between a
tree leaf and a path node away from the leaf. In particular, no directed cycle
can contain an overlay edge.

 \paragraph{Undirected weighted variant.}
Assign weight $1$ to every edge inherited from $H$. Assign weight $0$ to every edge
of the paths in $\mathcal{Q}\cup\mathcal{R}$. Assign weight $2(k+1)$ to
every edge in the overlay tree, including the edges between tree leaves and path nodes.
Thus, any cycle containing an overlay edge has weight at least $2(k+1)$.

\begin{figure}[ht!]
    \centering
      \usetikzlibrary{arrows.meta,backgrounds}
\begin{tikzpicture}[
    scale=1.3,
    every label/.style={inner sep=1pt, font=\footnotesize}
]
\footnotesize
    % Define styles for different types of vertices
    \tikzset{
        vertex/.style={circle, draw, minimum size=4pt, inner sep=1pt, fill=white},
        filled/.style={circle, draw, fill=black, minimum size=4pt, inner sep=1pt},
        symbol/.style={inner sep=0pt, font=\footnotesize, text height=12pt},
        small/.style={font=\footnotesize},
        mapping/.style={teal!60!blue, thick},
        dir/.style={arrows={-Stealth[inset=0.7pt, length=5pt, angle'=25]}},
        line/.style={gray, very thin, dir},
        base/.style={gray!30, line width=.75em, dir}
    }
    
    % Draw vertices on level A with labels
    \node[] at (-2.4,0.25) {$Q_1$};
    \node[vertex] (a) at (-2,0) [label=below left:{$a_1=q^1_0$}] {};
    \node[vertex] (b) at (1,0) [label=below:{$q^1_1$}] {};
    \node[vertex] (c) at (2,0) [label=below:{$q^1_2$}] {};
    \node[] (ca) at (3.5,0) {};
    \node[small] at (4,0) {$\cdots$};
    \node[] (cb) at (4.5,0) {};
    \node[vertex] (d) at (8,0) [label=right:{$u_1=q^1_{d^p-1}$}] {};

    \node[symbol] (dda) at (-.9,-0.6) {$\ddots$};
    \node[symbol] (ddb) at (1,-0.7) {$\vdots$};
    \node[symbol] (ddc) at (2,-0.7) {$\vdots$};
    \node[symbol] (ddd) at (6.8,-0.6) {$\iddots$};
    
    \node[symbol] (ddda) at (-.9,-3.5) {$\iddots$};
    \node[symbol] (dddb) at (1,-3.8) {$\vdots$};
    \node[symbol] (dddc) at (2,-3.8) {$\vdots$};
    \node[symbol] (dddd) at (6.8,-3.5) {$\ddots$};
    
    % Draw vertices on level B with labels
    \node[] at (-0.2, -1) {$Q_\gamma$};
    \node[vertex] (e) at (0,-1.2) [label=below:{$a_\gamma=q^\gamma_0$}] {};
    \node[vertex] (f) at (1,-1.2) [label=below:{$q^\gamma_1$}] {};
    \node[vertex] (g) at (2,-1.2) [label=below:{$q^\gamma_2$}] {};
    \node[] (ga) at (3.5,-1.2) {};
    \node[small] at (4,-1.2) {$\cdots$};
    \node[] (gb) at (4.5,-1.2) {};
    \node[vertex] (h) at (6,-1.2) [label=below left:{$u_\gamma=q^\gamma_{d^p-1}$}] {};

    % Draw vertices on level C with labels
    \node[] at (-0.45,-2.72) {$R_\gamma$};
    \node[vertex] (e2) at (0,-3) [label=below:{$b_\gamma=r^\gamma_0$}] {};
    \node[vertex] (f2) at (1,-3) [label=below:{$r^\gamma_1$}] {};
    \node[vertex] (g2) at (2,-3) [label=below:{$r^\gamma_2$}] {};
    \node[] (ga2) at (3.5,-3) {};
    \node[small] at (4,-3) {$\cdots$};
    \node[] (gb2) at (4.5,-3) {};
    \node[vertex] (h2) at (6,-3) [label=below:{$v_\gamma=r^\gamma_{d^p-1}$}] {};

    % Draw vertices on level D with labels
    \node[] at (-2.4,-3.95) {$R_1$};
    \node[vertex] (i) at (-2,-4.2) [label=below:{$b_1=r^1_0$}] {};
    \node[vertex] (j) at (1,-4.2) [label=below:{$r^1_1$}] {};
    \node[vertex] (k) at (2,-4.2) [label=below:{$r^1_2$}] {};
    \node[] (ka) at (3.5,-4.2) {};
    \node[small] at (4, -4.2) {$\cdots$};
    \node[] (kb) at (4.5,-4.2) {};
    \node[vertex] (l) at (8,-4.2) [label=below:{$v_1=r^1_{d^p-1}$}] {};
    
    % Draw tree structure vertices
    \begin{scope}[xshift=1.5mm, yshift=0mm]
        \node at (5,3) {$\mathcal{T}$};
        \node[vertex] (o) at (3.3,3.6) [label=right:{$t_0^0$}] {};
        
        \node (left) at (2.7,3.2) {};
        \node (right) at (4.1,3.0) {};
        
        \node[symbol] (eleft) at (2.4,3.1) {$\iddots$};
        \node[symbol] (eright) at (4.5,2.7) {$\ddots$};
        \node[small] (mid) at (4,1) {$\cdots$};
        \node[vertex] (pq) at (2,2.8) [label=left:{$t^{p-2}_0$}] {};
        
        \node[vertex] (p) at (1,2) [label=right:{$t^{p-1}_0$}] {};
        \node[vertex] (q) at (2.9,2) [label=right:{$t^{p-1}_1$}] {};
        \node[vertex] (qq) at (5.1,2) [label=right:{$t^{p-1}_{d^{p-1}-1}$}] {};
        
        \node[vertex] (s) at (0.5,1) [label=above left:{$t^p_0$}] {};
        \node[vertex] (m) at (1.5,1) [label=left:{$t^p_1$}] {};
        \node[vertex] (n) at (2.5,1) [label=left:{$t^p_2$}] {};
        \node[vertex] (nn) at (3.3,1) [label=below:{$t^p_3$}] {};
        \node[small] at (4.0,2.0) {$\cdots$};
        \node[vertex] (rl) at (4.7,1) [label=below:{$t^p_{d^p-2}$}] {};
        \node[vertex] (r) at (5.5,1) [label=above right:{$t^p_{d^p-1}$}] {};
    \end{scope}

    % Draw path connections (ordinary black)
    \draw[dir] (a) -- (b) -- (c) -- (ca);
    \draw[dir] (e) -- (f) -- (g) -- (ga);
    \draw[dir] (ga2) -- (g2) -- (f2) -- (e2);
    \draw[dir] (ka) -- (k) -- (j) -- (i);
    \draw[dir] (cb) -- (d);
    \draw[dir] (gb) -- (h);
    \draw[dir] (h2) -- (gb2);
    \draw[dir] (l) -- (kb);
    
    % Draw tree connections
    \draw[line] (o) -- (left);
    \draw[line] (o) -- (right);
    \draw[line] (pq) -- (p);
    \draw[line] (pq) -- (q);
    \draw[line] (p) -- (s);
    \draw[line] (p) -- (m);
    \draw[line] (q) -- (n);
    \draw[line] (q) -- (nn);
    \draw[line] (qq) -- (r);
    \draw[line] (qq) -- (rl);
    
    \draw[line] (s) to [out=260,in=35] (a);
    \draw[line] (s) to [out=265,in=50] (e);
    \draw[line] (s) to [out=270,in=60] (e2);
    \draw[line] (s) to [out=276,in=0, looseness=1.45] (i);
    
    \draw[line] (m) to [out=269,in=35] (b);
    \draw[line] (m) to [out=269,in=50] (f);
    \draw[line] (m) to [out=272,in=60] (f2);
    \draw[line] (m) to [out=275,in=60] (j);
    
    \draw[line] (n) to [out=268,in=35] (c);
    \draw[line] (n) to [out=269,in=50] (g);
    \draw[line] (n) to [out=272,in=60] (g2);
    \draw[line] (n) to [out=275,in=60] (k);
    
    \draw[line] (r) to [out=355,in=115] (d);
    \draw[line] (r) to [out=300,in=85] (h);
    \draw[line] (r) to [out=260,in=120, looseness=1] (h2);
    \draw[line] (r) to [out=255,in=175, looseness=1.4] (l);

    % Helper points
    \node (vl) at (-1,-.7) {};
    \node (wl) at (-1,-3.6) {};    
    \node (vr) at (6.95,-.7) {};
    \node (wr) at (6.95,-3.6) {};

    % Thick gray "background/base" edges
    \begin{pgfonlayer}{background}
        \draw[base] (d) -- (wr);
        \draw[base] (vr) -- (l);
        \draw[base] (vr) -- (h2);
        \draw[base] (h) -- (l);
        \draw[base] (h2) -- (h);
        \draw[base] (wr) -- (h);
        \draw[base] (d) -- (l);

        \draw[base] (a) -- (wl);
        \draw[base] (vl) -- (i);
        \draw[base] (vl) -- (e2);
        \draw[base] (e) -- (i);
        \draw[base] (e2) -- (e);
        \draw[base] (wl) -- (e);
        \draw[base] (a) -- (i);
    \end{pgfonlayer}

    % Teal "actual/mapped" edges
    \draw[dir, mapping] (d) -- (wr);
    \draw[dir, mapping] (vr) -- (l);
    \draw[dir, mapping] (vr) -- (h2);
    \draw[dir, mapping] (h) -- (h2);
    \draw[dir, mapping] (h) -- (wr);

    \draw[dir, mapping] (wl) -- (a);
    \draw[dir, mapping] (i) -- (vl);
    \draw[dir, mapping] (e2) -- (vl);
    \draw[dir, mapping] (e2) -- (e);
    \draw[dir, mapping] (i) -- (a);

\end{tikzpicture}
    \caption{
An illustration of $\mathcal{G}_{d,p}(H,x,y)$. The thick gray edges depict
the base graph $H$, while the blue edges depict the two subgraphs of $H$
that encode the inputs $x$ and $y$.}
    \label{fig:GgkdpHxy}
\end{figure}

\begin{observation}[Size and diameter]
\label{obs:congest-lb-size}
The graph $\mathcal{G}_{d,p}(H,x,y)$ has
\[
    2\gamma d^p+\frac{d^{p+1}-1}{d-1}
    \in
    \Theta(\gamma d^p)
\]
nodes, and its undirected unweighted diameter is at most $2p+2$.
\end{observation}

\begin{proof}
The paths in $\mathcal{Q}\cup\mathcal{R}$ contain $2\gamma d^p$ nodes in
total, while $\mathcal{T}$ contains $(d^{p+1}-1)/(d-1)$ nodes. Since
$\gamma\ge 1$, the claimed size bound follows.

The diameter bound of $2p+2$ follows because every node in $\mathcal{G}_{d,p}(H,x,y)$ can reach the root of the overlay tree $\mathcal{T}$ by a path of at most $p+1$ edges.
\end{proof}

\begin{lemma}[Approximation gap in the $\CONGEST$ construction]
\label{lem:congest-cycle-gap}
Suppose $H$ has girth greater than $2k$. Then the following hold for $\mathcal{G}_{d,p}(H,x,y)$.\\

\begin{minipage}[t]{0.48\textwidth}
\textbf{Directed unweighted case:}
\begin{itemize}
    \item If $\inprod{x}{y}\neq 0$, then
    $\OPT\le 2d^p$.
    \item If $\inprod{x}{y}=0$, then
    $\OPT\ge 2(k+1)d^p$.
\end{itemize}
\end{minipage}
\hfill
\begin{minipage}[t]{0.48\textwidth}
\textbf{Undirected weighted case:}
\begin{itemize}
    \item If $\inprod{x}{y}\neq 0$, then
    $\OPT\le 2$.
    \item If $\inprod{x}{y}=0$, then
    $\OPT\ge 2(k+1)$.
\end{itemize}
\end{minipage}
\end{lemma}

\begin{proof}
The proof follows the same projection argument as
\Cref{lem:basic-cycle-gap}, with each matching edge in the basic construction
replaced by the corresponding path in $\mathcal{Q}\cup\mathcal{R}$.

Consider first the directed case. By construction, no directed cycle can use
an overlay edge. If $\inprod{x}{y}\neq 0$, then for some
$e=\{\ell_i,\rho_j\}\in E_H$, both edges corresponding to $e$
are present. Together with $Q_i$ and $R_j$, they form a directed cycle of
weight $2(d^p-1)+2=2d^p$.

Now suppose $\inprod{x}{y}=0$. Replace every path in
$\mathcal{Q}\cup\mathcal{R}$ with a single matching edge. Any directed cycle
then becomes a directed cycle in the basic construction. By
\Cref{lem:basic-cycle-gap}, the resulting cycle uses at least $2k+2$ edges
inherited from $H$. Since such edges alternate with the  matching
edges, the original cycle uses at least $2k+2$ paths from
$\mathcal{Q}\cup\mathcal{R}$. Therefore, its weight is at least
$(2k+2)((d^p-1)+1)=2(k+1)d^p$.

For the undirected weighted case, any cycle containing an overlay edge has
weight at least $2(k+1)$. For every remaining cycle, replacing each
zero-weight path in $\mathcal{Q}\cup\mathcal{R}$ with a single zero-weight
matching edge yields a cycle in $\mathcal{G}(H,x,y)$. This correspondence is
bijective and preserves the cycle weight. The claim therefore follows directly
from \Cref{lem:basic-cycle-gap}.
\end{proof}

\subsection{The Moving-Cut Framework}
\label{subsec:moving-cut}

We next review the moving-cut framework of
\citet{haeupler2020network,haeupler2021universally}, which we use to prove
that set-disjointness is hard on our $\CONGEST$ construction.
Informally, {a} moving cut is a combinatorial object that certifies distributed lower
bounds for communication problems. The framework builds on the approach of
\citet{das2011distributed}, who proved various
$\tilde{\Omega}(\sqrt{n})$ lower bounds on $O(\log n)$-diameter graphs by
tracking communication bottlenecks throughout the execution of a distributed
algorithm. Moving cuts were later introduced explicitly by
\citet{haeupler2020network} to establish network-coding gaps for simple
pairwise communication tasks, and were subsequently used by
\citet{haeupler2021universally} to prove universal lower bounds.

\paragraph{Distributed set-disjointness.}
We consider the set-disjointness problem in the $\CONGEST$ model, where the
input bits are distributed among a collection of source--sink pairs. Let
$G=(V,E)$ be a graph, and let
\[
    S=\{(s_i,t_i):i\in[b]\}
\]
be a collection of source--sink pairs, where a node may participate in
multiple pairs. Given $x,y\in\bin^b$, the source $s_i$ receives $x_i$ and
the sink $t_i$ receives $y_i$. The goal is for every node to output
$\disj_b(x,y)$.

\begin{definition}[Moving cuts~\cite{haeupler2020network,haeupler2021universally}]
\label{def:moving-cut}
A \underline{moving cut} for $(G=(V,E),S)$ is an assignment
$\ell:E\to\mathbb{Z}_{\ge 1}$. Its \underline{capacity} is
\[
    \lambda
    =
    \sum_{e\in E}\left(\ell(e)-1\right),
\]
and its \underline{distance} is
\[
    \min_{i,j\in[b]}
    \dist_{\ell}(s_i,t_j),
\]
where $\dist_{\ell}$ denotes shortest-path distance with respect to the edge
lengths $\ell$.
\end{definition}

Intuitively, the capacity of a moving cut measures the total increase relative to the original unit edge lengths, whereas its distance measures how far the cut separates every source from every sink.

\begin{lemma}[Lower bound via moving cuts~\cite{haeupler2021universally}]
\label{lem:moving-cut}
Suppose $(G,S)$ admits a moving cut of distance at least $\beta$ and
capacity strictly less than $b$. Then there exists a constant $\delta>0$ such
that every $\delta$-error distributed algorithm for computing $\disj_b$ on
$(G,S)$ requires $\tilde{\Omega}(\beta)$ rounds in the $\CONGEST$ model.
This holds even with public randomness and even if the entire graph $G$ and
the terminal pairs $S$ are known to all nodes.
\end{lemma}

Our plan is to apply \Cref{lem:moving-cut} to the graph
$\mathcal{G}_{d,p}(H,\mathbf{1},\mathbf{1})$, where $\mathbf{1}$ denotes the all-one vector. Crucially, the topology of this graph is independent of the set-disjointness inputs: all edges of $H$ are present in both copies, while the bits {$x_e$ and $y_e$} are provided only as private inputs to the corresponding terminals. 

We then show how the terminals can use their private input bits to simulate the input-dependent graph $\mathcal{G}_{d,p}(H,x,y)$, thereby reducing set-disjointness to $(k+1-\epsilon)$-$\apxMWC$.

\subsection{Set-Disjointness Lower Bound on \texorpdfstring{$\mathcal{G}_{d,p}(H,\mathbf{1},\mathbf{1})$}{Gdp(H,1,1)}}
\label{subsec:disj-lb-fixed-graph}

We now prove that set-disjointness is hard on  
$\mathcal{G}_{d,p}(H,\mathbf{1},\mathbf{1})$ via \Cref{lem:moving-cut}. 
For every edge $e=\{\ell_i,\rho_j\}\in E_H$, define
\[
    s_e=a_i
    \qquad\text{and}\qquad
    t_e=v_j,
\]
and consider this collection of source--sink pairs
\[
    S(H)
    =
    \left\{
        (s_e,t_e)
        :
        e\in E_H
    \right\}.
\]
Thus, in the distributed set-disjointness problem on $(\mathcal{G}_{d,p}(H,\mathbf{1},\mathbf{1}), S(H))$, for every edge $e=\{\ell_i,\rho_j\} \in E_H$, the source $s_e=a_i$ receives $x_e$
and the sink $t_e=v_j$ receives $y_e$. A node may participate in several
terminal pairs and hence hold several input bits.

\begin{lemma}[Set-disjointness lower bound on $\mathcal{G}_{d,p}(H,\mathbf{1},\mathbf{1})$]
\label{lem:disj-fixed-graph}
Let $H=(L\cup R,E_H)$ be any bipartite graph with
$|L|=|R|=\gamma$. For any integers $d\ge 2$ and $p\ge 1$, there exists a
constant $\delta>0$ such that every $\delta$-error algorithm for computing
$\disj_{|E_H|}$ on
\(
    \left(\mathcal{G}_{d,p}(H,\mathbf{1},\mathbf{1}),S(H)\right)
\)
requires
\[
    \tilde{\Omega}\left(
        \min\left\{
            d^p,\frac{|E_H|}{dp}
        \right\}
    \right)
\]
rounds in the $\CONGEST$ model.
\end{lemma}

\begin{proof}
Write $b=|E_H|$. We construct a moving cut on
$\mathcal{G}_{d,p}(H,\mathbf{1},\mathbf{1})$. For every overlay-tree edge
between depths $h-1$ and $h$, where $h\in[p]$, set
\[
    \ell(e)
    =
    1+\left\lfloor\frac{b}{2pd^h}\right\rfloor.
\]
Set $\ell(e)=1$ for every remaining edge.

\paragraph{Capacity.}
There are $d^h$ tree edges between depths $h-1$ and $h$. Hence, the capacity
is at most
\[
    \sum_{h=1}^p
    d^h\left\lfloor\frac{b}{2pd^h}\right\rfloor
    \le
    \frac{b}{2}
    <
    b.
\]

\paragraph{Distance.}
Consider any path $P$ from a source $s_e=a_i=q^i_0$ to a sink
$t_{e'}=v_j=r^j_{d^p-1}$. Associate position $s$ with every path node
$q^i_s$ and $r^j_s$. Thus, $P$ starts at position $0$ and ends at position
$d^p-1$. Edges inherited from $H$ preserve the position, so there are only
two ways for $P$ to move between positions.

First, $P$ may use path edges. Each such edge changes the position by one and
has $\ell$-length $1$, so moving $\Delta$ positions in this way costs at least
$\Delta$.

Second, $P$ may traverse the overlay tree. Consider a maximal subpath that
enters the overlay tree from position $s$ and leaves it at position $t$.
Let $h-1$ be the depth of the lowest common ancestor of the leaves $t^p_s$
and $t^p_t$. The tree path between these leaves contains an edge $e^\ast$
between depths $h-1$ and $h$, and
$|s-t|<d^{p-h+1}$. Therefore,
\[
    \ell(e^\ast)
    \ge
    \frac{b}{2pd^h}
    >
    \frac{b|s-t|}{2pd^{p+1}}.
\]
Hence, moving $|s-t|$ positions through the overlay tree costs
$\Omega(b|s-t|/(pd^{p+1}))$.

Thus, regardless of how $P$ moves between positions, each unit of positional
progress costs at least
$\Omega(\min\{1,b/(pd^{p+1})\})$. Since $P$ must move from position $0$ to
position $d^p-1$, its $\ell$-length is at least
\[
    \Omega\left(
        \min\left\{
            1,\frac{b}{pd^{p+1}}
        \right\}
        (d^p-1)
    \right)
    =
    \Omega\left(
        \min\left\{
            d^p,\frac{b}{dp}
        \right\}
    \right).
\]
Hence, the moving cut has distance
$\Omega(\min\{d^p,b/(dp)\})$. Since its capacity is strictly less than
$b=|S(H)|$, the claim follows from \Cref{lem:moving-cut}.
\end{proof}

\subsection{From Set-Disjointness to \texorpdfstring{$\apxMWC$}{Apx-MWC}}
\label{subsec:disj-to-mwc}

We now combine the set-disjointness lower bound from
\Cref{lem:disj-fixed-graph} with the approximation gap from
\Cref{lem:congest-cycle-gap}. The former is proved on 
$\mathcal{G}_{d,p}(H,\mathbf{1},\mathbf{1})$, where $x$ and $y$ are given
as private terminal inputs, whereas the latter concerns the input-dependent
graph $\mathcal{G}_{d,p}(H,x,y)$. In the following proof, we connect the two
settings.

\begin{proposition}[$\CONGEST$ lower bound from a high-girth graph]
\label{prop:congest-lb-from-H}
Let $H=(L\cup R,E_H)$ be a bipartite graph with
$|L|=|R|=\gamma$ and $\operatorname{girth}(H)>2k$. Let $\epsilon > 0$ be any real number.
For any integers $d\ge 2$ and $p\ge 1$, there exists a constant
$\delta>0$ such that any $\delta$-error algorithm for
$(k+1-\epsilon)$-$\apxMWC$ in the $\CONGEST$ model requires
\[
    \tilde{\Omega}\left(
        \min\left\{
            d^p,\frac{|E_H|}{dp}
        \right\}
    \right)
\]
rounds on the family of graphs
\[
    \left\{
        \mathcal{G}_{d,p}(H,x,y)
        :
        x,y\in\bin^{E_H}
    \right\}.
\]
The lower bound holds for directed unweighted graphs and undirected graphs
with non-negative integer weights.
\end{proposition}

\begin{proof}
Let $\delta>0$ be the constant from \Cref{lem:disj-fixed-graph}.
Suppose there is a $\delta$-error distributed algorithm $\mathcal{A}$
that solves $(k+1-\epsilon)$-$\apxMWC$ on every graph in the above family
in $T$ rounds. We use $\mathcal{A}$ to compute $\disj_{|E_H|}$ on 
$\mathcal{G}_{d,p}(H,\mathbf{1},\mathbf{1})$ with terminal pairs $S(H)$.

For each edge $e=\{\ell_i,\rho_j\}\in E_H$, the source $s_e=a_i$ holds
$x_e$, while the sink $t_e=v_j$ holds $y_e$. Recall that, in $\mathcal{G}_{d,p}(H,x,y)$,  $x_e$ determines whether there is
 an edge between $a_i$ and $b_j$, whereas $y_e$ determines whether there is
 an edge between $u_i$ and $v_j$.

\paragraph{Preparing the simulated instance.}
In one round, $a_i$ sends $x_e$ to $b_j$, while $v_j$ sends $y_e$ to $u_i$,
simultaneously for all $e=\{\ell_i,\rho_j\}\in E_H$. This is possible because both edges $\{a_i, b_j\}$ and $\{u_i,v_j\}$ are present in
$\mathcal{G}_{d,p}(H,\mathbf{1},\mathbf{1})$. After that, both endpoints of every edge in $\mathcal{G}_{d,p}(H,\mathbf{1},\mathbf{1})$ learn whether the edge is present in
$\mathcal{G}_{d,p}(H,x,y)$ under the input strings $x$ and $y$.
This allows the nodes of $\mathcal{G}_{d,p}(H,\mathbf{1},\mathbf{1})$ to simulate $\mathcal{A}$ on
$\mathcal{G}_{d,p}(H,x,y)$ without any overhead.

\paragraph{Recovering set-disjointness.}
Consider first the directed unweighted case. If $\inprod{x}{y}\neq 0$, then
\Cref{lem:congest-cycle-gap} gives $\OPT\le 2d^p$, so $\mathcal{A}$ returns
a cycle of length at most
\(
    (k+1-\epsilon)\OPT
    \le
    2(k+1-\epsilon)d^p
    <
    2(k+1)d^p
\).
If $\inprod{x}{y}=0$, then every cycle has length at least
$2(k+1)d^p$.

Similarly, in the undirected weighted case, if $\inprod{x}{y}\neq 0$, then
$\OPT\le 2$, so $\mathcal{A}$ returns a cycle of weight strictly below
$2(k+1)$. If $\inprod{x}{y}=0$, then every cycle has weight at least
$2(k+1)$.

By the output convention for $\apxMWC$, every node learns the weight of the
returned cycle. Hence, after the simulation, every node can determine
$\disj_{|E_H|}(x,y)$ by comparing this weight with the corresponding
threshold. The resulting set-disjointness algorithm has error at most
$\delta$ and runs in $T+1$ rounds, where the $+1$ term is due to the preparation of the simulated instance.

Applying \Cref{lem:disj-fixed-graph} yields
\(
    T
    \in
    \tilde{\Omega}\left(
        \min\left\{
            d^p,\frac{|E_H|}{dp}
        \right\}
    \right)
\), 
which proves the proposition.
\end{proof}

We are ready to prove \Cref{thm:klower}.

\begin{proof}[Proof of \Cref{thm:klower}]
Assume the Erd\H{o}s girth conjecture. For every integer $k\ge 1$ and every
integer $\gamma\ge 1$, there exists a bipartite graph
$H=(L\cup R,E_H)$ with
\[
    |L|=|R|=\gamma,
    \qquad
    \operatorname{girth}(H)>2k,
    \qquad
    |E_H|\in\Omega\left(\gamma^{1+1/k}\right).
\]
Set $d=2$ and
\[
    p
    =
    \left\lceil
        \frac{k+1}{k}\log\gamma
    \right\rceil.
\]
Then $d^p\in\Theta(\gamma^{1+1/k})$ and
$p\in\Theta(\log\gamma)$. By \Cref{prop:congest-lb-from-H}, any
$\delta$-error algorithm for $(k+1-\epsilon)$-$\apxMWC$ requires
\[
    \tilde{\Omega}\left(
        \min\left\{
            d^p,\frac{|E_H|}{2p}
        \right\}
    \right)
    =
    \tilde{\Omega}\left(\gamma^{1+1/k}\right)
\]
rounds.

By \Cref{obs:congest-lb-size}, the number of nodes satisfies
\[
    n
    \in
    \Theta(\gamma d^p)
    =
    \Theta\left(\gamma^{2+1/k}\right).
\]
Therefore,
\[
    \gamma^{1+1/k}
    \in
    \Theta\left(
        n^{\frac{k+1}{2k+1}}
    \right),
\]
which gives the lower bound
\[
    \tilde{\Omega}\left(
        n^{\frac{k+1}{2k+1}}
    \right).
\]
Moreover, the hard instances have diameter
$O(p)=O(\log n)$.

This proves the result for directed unweighted graphs and for undirected graphs with \emph{non-negative} integer weights. Applying the same transformation as in the proof of \Cref{thm:cliqueLB} allows us to extend the result to undirected graphs with polynomially bounded \emph{positive} integer weights.
\end{proof}
 
\section{Conclusions and Open Problems}

We studied the distributed complexity of $\apxMWC$ in undirected weighted graphs. Our results  {give a continuous} round--approximation tradeoff in the $\CONGEST$ model, obtaining a nearly tight characterization of the round complexity of $\apxMWC$, assuming the Erd\H{o}s girth conjecture.

From a technical perspective, our main contribution is a new connection between $\apxMWC$ and the MPX low-diameter decomposition~\cite{miller2013parallel}. MPX has become a central tool in parallel and distributed graph algorithms, with applications to spanners~\cite{elkin2018efficient,forsterOPODIS2021}, network decompositions~\cite{ghaffari2024near}, shortest paths~\cite{andoni2020parallel,ghaffari2024near2}, approximation algorithms for packing and covering problems~\cite{chang2023complexity,davies2026distributed}, potential problems~\cite{balliu2025distributed}, radio network algorithms~\cite{Chang20bfs,Chang18broadcast,CzumajD17,dani2022wake,haeupler2016faster}, and many other problems. Our work adds $\apxMWC$ to this growing list. More broadly, our results suggest that MPX is the right tool for this problem: it yields upper bounds matching our lower bounds not only in the $\CONGEST$ model, but also in the broadcast congested clique model. We hope this connection will inspire further uses of MPX in parallel and distributed graph algorithms.

Several intriguing questions about $\apxMWC$ remain open.

\paragraph{Undirected unweighted graphs.}
For the undirected unweighted case, \citet{chechik2026girth} recently showed that, for every {constant integer $f>2$}, an $f$-approximation can be computed in $\tilde{O}(n^{1/f}+D)$ rounds in the $\CONGEST$ model. This substantially improves the state of the art, but no matching lower bound is known. Closing this gap is a natural next step toward a complete understanding of the complexity of $\apxMWC$.

\paragraph{Directed graphs.}
Another natural direction is to study directed graphs. We conjecture that the same round--approximation tradeoff as in the undirected weighted setting should hold. Our lower bound construction already extends to directed graphs, so the main challenge is algorithmic.

A plausible route is to replace MPX by a directed analogue of the decomposition. Existing directed low-diameter decompositions, however, incur an additional $O(\log\log n)$-factor loss~\cite{haeupler2025stronger}, which would translate into the same loss in the approximation ratio. Our algorithm does not need the full strength of directed low-diameter decomposition, so it would be very interesting to identify a weaker directed decomposition, tailored to capturing cycles of small weight, that can be computed efficiently in distributed models.

Such a result could also have implications beyond distributed algorithms. For directed graphs, a similar $O(\log\log n)$-factor loss appears in the state-of-the-art tradeoff in the centralized setting: \citet{chechik2020constant} showed that $O(k\log\log n)$-$\apxMWC$ can be solved in $\tilde{O}(m^{1+1/k})$ time, for every integer $k\geq 1$. Can this time--approximation tradeoff be improved?

\paragraph{Universally optimal algorithms.}
A distributed algorithm is \emph{universally optimal} if, on every input graph, its round complexity matches that of the best distributed algorithm tailored to that graph. Recent work~\cite{haeupler2021universally,haeupler2022hop} has shown that several problems in the complexity class $\tilde{\Theta}(\sqrt n+D)$ admit approximately universally optimal algorithms in the $\CONGEST$ model, including minimum spanning tree, $(1+\epsilon)$-approximate SSSP, and $(1+\epsilon)$-approximate minimum cut.

Given the close connection established in this paper between $\apxMWC$ and $(1+\epsilon)$-approximate SSSP via MPX, it is natural to ask whether this line of research can be extended to $\apxMWC$. Does $\apxMWC$ admit universally optimal distributed algorithms?

A natural intermediate goal is the regime of $O(\log n)$-approximation. Our approach already shows that this approximation ratio can be achieved using only polylogarithmically many calls to $(1+\epsilon)$-approximate SSSP. The substantially more challenging case is to obtain universal optimality for smaller approximation ratios, where the worst-case round complexity is strictly higher than $\tilde{\Theta}(\sqrt n+D)$.

\section*{Acknowledgments}
We thank Bernhard Haeupler and Thatchaphol Saranurak for helpful discussions, particularly for clarifying prior work on shortest paths and low-diameter decompositions. Thatchaphol Saranurak suggested using tree-like approximate SSSP, which allowed us to eliminate additional terms from the upper bound in an earlier version  of the paper (\href{https://arxiv.org/abs/2603.25368v2}{arXiv:2603.25368v2}).

\section*{AI Methodology}
The main results and overall proof strategy were developed by the human authors without AI assistance. During manuscript preparation and revision, we used OpenAI's ChatGPT to assist in proving specific lemmas and refining their presentation. ChatGPT also helped with language and style editing and the creation and revision of figures. ChatGPT's assistance materially affected the entire manuscript.

In particular, ChatGPT suggested using random parity bits to correct an error that an anonymous reviewer identified in the patching construction of an earlier version of the paper (\href{https://arxiv.org/abs/2603.25368v3}{arXiv:2603.25368v3}); see \Cref{subsect:extended-cycles} for the corrected construction.

All AI-assisted content was carefully reviewed and revised by the authors, who assume responsibility for the entire paper.

\printbibliography

\end{document}